%% file: 18TIT-lama.tex
\documentclass[draftclsnofoot,12pt,onecolumn]{IEEEtran}

\usepackage{mleftright}
\usepackage{cite}
\usepackage[T1]{fontenc}
\usepackage{graphicx}
\usepackage{pgfplots}
\usepackage{grffile}
\usepackage{sansmath}
\pgfplotsset{every tick label/.append style={font=\sansmath\LARGE}}
\pgfplotsset{every axis legend/.append style={font=\sansmath\LARGE}}
\pgfplotsset{every axis label/.append style={font=\sansmath\LARGE}}
\pgfplotsset{/tikz/font=\sansmath\sffamily}
\pgfplotsset{compat=newest}
\usetikzlibrary{plotmarks}
\usepgfplotslibrary{patchplots}
\usepackage{hyperref}

\usepackage{mathtools}
\usepackage{amssymb}
\usepackage{amsmath}
\usepackage{paralist}
\usepackage{subfigure}
\usepackage{booktabs} %\toprule \midrule \bottomrule
\usepackage{algorithm}
\usepackage{algorithmic}
\usepackage{amsthm}
\usepackage{multirow}
\usepackage{microtype}
\usepackage{tablefootnote}
\usepackage{standalone}

\usepackage{adjustbox}
\input{macros/vmr-symbols-vecbold}
\input{macros/standard-macros}

\input{macros/defs}

\allowdisplaybreaks % THIS ALLOWS EQUATIONS TO BREAK THE PAGE

\newcommand{\ar}{a_\textnormal{R}}%
\newcommand{\ai}{a_\textnormal{I}}%
\newcommand{\Sr}{S_\textnormal{R}}%
\newcommand{\dd}{\textnormal{d}}%

\usepackage{color}

\newcommand{\shat}[1]{\hat\bms#1}
\newcommand{\shate}[1]{\hat{s}#1}

\newcommand{\srv}[1]{S#1}

\newcommand{\resid}[1]{\bmr#1}
\newcommand{\reside}[1]{r#1}

\newcommand{\realpart}[1]{\textnormal{Re}\!\left\{ #1 \right\}\!}

\newcommand{\imagpart}[1]{\textnormal{Im}\!\left\{ #1 \right\}\!}

\safemath{\LAMA}{\textrm{LAMA}}
\safemath{\MRT}{\textrm{MRT}}
\safemath{\betamax}{\beta^\textnormal{max}_\setO}
\safemath{\betamaxno}{\beta^\textnormal{max}}
\safemath{\betamin}{\beta^\textnormal{min}_\setO}
\safemath{\betaminno}{\beta^\textnormal{min}}

\safemath{\Nomin}{\No^\textnormal{min}(\beta)}
\safemath{\Nominnobeta}{\No^\textnormal{min}}
\safemath{\Nomax}{\No^\textnormal{max}(\beta)}
\safemath{\Nomaxnobeta}{\No^\textnormal{max}}

\safemath{\MAP}{\textrm{MAP}}
\safemath{\IO}{\textrm{IO}}
\safemath{\JO}{\textrm{JO}}
\safemath{\Nopost}{N_{0}^\textnormal{post}}
\safemath{\MT}{{M_\textnormal{T}}}
\safemath{\MR}{{M_\textnormal{R}}}
\safemath{\Tran}{\textnormal{T}}
\safemath{\Herm}{\textnormal{H}}
\safemath{\row}{\textnormal{r}}
\safemath{\col}{\textnormal{c}}

\begin{document}

\title{Optimal Data Detection in Large MIMO}\author{Charles Jeon, Ramina Ghods, Arian Maleki, and Christoph Studer\thanks{This paper was presented in part at the 2015 International Symposium on Information Theory (ISIT)\cite{JGMS2015conf}.}\thanks{C. Jeon, R.~Ghods, and C.~Studer are with the School of Electrical and Computer Engineering, Cornell University, Ithaca, NY; e-mail: {jeon@csl.cornell.edu}, {rghods@csl.cornell.edu},  {studer@cornell.edu}; web: vip.ece.cornell.edu} \thanks{A. Maleki is with Department of Statistics at Columbia University, New York City, NY; e-mail: {arian@stat.columbia.edu}.} \thanks{MATLAB code to reproduce our numerical simulations will be made available upon (possible) acceptance of paper.}}
\maketitle

\begin{abstract}
Large multiple-input multiple-output (MIMO) appears in massive multi-user  MIMO and randomly-spread code-division multiple access (CDMA)-based wireless systems. In order to cope with the excessively high complexity of optimal data detection in such systems, a variety of efficient yet sub-optimal algorithms have been proposed in the past. 
In this paper, we propose a data detection algorithm that is computationally efficient and optimal in a sense that it is able to achieve the same error-rate performance as the individually optimal (IO) data detector under certain assumptions on the MIMO system matrix and constellation alphabet.
Our algorithm, which we refer to as \LAMA (short for LArge MIMO AMP),  builds on complex-valued Bayesian approximate message passing (AMP), which enables an exact analytical characterization of the performance and complexity in the large-system limit via the state-evolution framework.
We derive optimality conditions for LAMA and investigate performance/complexity trade-offs. As a byproduct of our analysis, we recover classical results of IO data detection for randomly-spread CDMA.
We furthermore provide practical ways for LAMA to approach the theoretical performance limits in realistic, finite-dimensional systems at low computational complexity. 
\end{abstract}

%Large multiple-input multiple-output (MIMO) appears in massive multi-user  MIMO and randomly-spread code-division multiple access (CDMA)-based wireless systems. In order to cope with the excessively high complexity of optimal data detection in such systems, a variety of efficient yet sub-optimal algorithms have been proposed in the past. In this paper, we propose a data detection algorithm that is computationally efficient and optimal in a sense that it is able to achieve the same error-rate performance as the individually optimal (IO) data detector under certain assumptions on the MIMO system matrix and constellation alphabet. Our algorithm, which we refer to as LAMA (short for large MIMO AMP),  builds on complex-valued Bayesian approximate message passing (AMP), which enables an exact analytical characterization of the performance and complexity in the large-system limit via the state-evolution framework. We derive optimality conditions for LAMA and investigate performance/complexity trade-offs. As a byproduct of our analysis, we recover classical results of IO data detection for randomly-spread CDMA. We furthermore provide practical ways for LAMA to approach the theoretical performance limits in realistic, finite-dimensional systems at low computational complexity. 


%\cj{$\beta>\beta^{max}$ multiple fixed points discussion.}
%
%Questions for Arian:
%\begin{enumerate}
%\item A function F has to be pseudo-Lipschitz for AMP to work. How do you show this for complex-valued functions? (Section 4)
%\item Assume that $p(\vecs)=\prod_\ell p(s_\ell)$ and the $p(s_\ell)$ are \emph{not} i.i.d.\ but independent and known. Can we still develop something like an AMP algorithm (I think this one is true) and is there still something like state evolution that would still hold? Or would one need GAMP-like stuff?
%\end{enumerate}

\input{sec1-intro}

\input{sec2-system}

\input{sec3-damp}

\input{sec4-BER.tex}

\input{sec5-conclusion.tex}

\appendices
\input{secc-app.tex}

\bibliographystyle{IEEEtran}
\bibliography{bib/VIPabbrv,bib/confs-jrnls,bib/publishers,bib/VIP_180924}

\end{document}

%% file: macros/vmr-symbols-vecbold.tex
\usepackage{amssymb}
\usepackage{amsfonts}
\usepackage{mathrsfs}
\usepackage{xspace}
\usepackage{bm}
\usepackage{upgreek}

\newcommand{\safemath}[2]{\newcommand{#1}{\ensuremath{#2}\xspace}}

\safemath{\bma}{\mathbf{a}}
\safemath{\bmb}{\mathbf{b}}
\safemath{\bmc}{\mathbf{c}}
\safemath{\bmd}{\mathbf{d}}
\safemath{\bme}{\mathbf{e}}
\safemath{\bmf}{\mathbf{f}}
\safemath{\bmg}{\mathbf{g}}
\safemath{\bmh}{\mathbf{h}}
\safemath{\bmi}{\mathbf{i}}
\safemath{\bmj}{\mathbf{j}}
\safemath{\bmk}{\mathbf{k}}
\safemath{\bml}{\mathbf{l}}
\safemath{\bmm}{\mathbf{m}}
\safemath{\bmn}{\mathbf{n}}
\safemath{\bmo}{\mathbf{o}}
\safemath{\bmp}{\mathbf{p}}
\safemath{\bmq}{\mathbf{q}}
\safemath{\bmr}{\mathbf{r}}
\safemath{\bms}{\mathbf{s}}
\safemath{\bmt}{\mathbf{t}}
\safemath{\bmu}{\mathbf{u}}
\safemath{\bmv}{\mathbf{v}}
\safemath{\bmw}{\mathbf{w}}
\safemath{\bmx}{\mathbf{x}}
\safemath{\bmy}{\mathbf{y}}
\safemath{\bmz}{\mathbf{z}}
\safemath{\bmzero}{\mathbf{0}}
\safemath{\bmone}{\mathbf{1}}

\bmdefine{\biad}{a}
\bmdefine{\bibd}{b}
\bmdefine{\bicd}{c}
\bmdefine{\bidd}{d}
\bmdefine{\bied}{e}
\bmdefine{\bifd}{f}
\bmdefine{\bigd}{g}
\bmdefine{\bihd}{h}
\bmdefine{\biid}{i}
\bmdefine{\bijd}{j}
\bmdefine{\bikd}{k}
\bmdefine{\bild}{l}
\bmdefine{\bimd}{m}
\bmdefine{\bind}{n}
\bmdefine{\biod}{o}
\bmdefine{\bipd}{p}
\bmdefine{\biqd}{q}
\bmdefine{\bird}{r}
\bmdefine{\bisd}{s}
\bmdefine{\bitd}{t}
\bmdefine{\biud}{u}
\bmdefine{\bivd}{v}
\bmdefine{\biwd}{w}
\bmdefine{\bixd}{x}
\bmdefine{\biyd}{y}
\bmdefine{\bizd}{z}

\bmdefine{\bixid}{\xi}
\bmdefine{\bilambdad}{\lambda}
\bmdefine{\bimud}{\mu}
\bmdefine{\bithetad}{\theta}
\bmdefine{\biphid}{\phi}
\bmdefine{\bideltad}{\delta}

\safemath{\bmia}{\biad}
\safemath{\bmib}{\bibd}
\safemath{\bmic}{\bicd}
\safemath{\bmid}{\bidd}
\safemath{\bmie}{\bied}
\safemath{\bmif}{\bifd}
\safemath{\bmig}{\bigd}
\safemath{\bmih}{\bihd}
\safemath{\bmii}{\biid}
\safemath{\bmij}{\bijd}
\safemath{\bmik}{\bikd}
\safemath{\bmil}{\bild}
\safemath{\bmim}{\bimd}
\safemath{\bmin}{\bind}
\safemath{\bmio}{\biod}
\safemath{\bmip}{\bipd}
\safemath{\bmiq}{\biqd}
\safemath{\bmir}{\bird}
\safemath{\bmis}{\bisd}
\safemath{\bmit}{\bitd}
\safemath{\bmiu}{\biud}
\safemath{\bmiv}{\bivd}
\safemath{\bmiw}{\biwd}
\safemath{\bmix}{\bixd}
\safemath{\bmiy}{\biyd}
\safemath{\bmiz}{\bizd}

\safemath{\bmxi}{\bixid}
\safemath{\bmlambda}{\bilambdad}
\safemath{\bmmu}{\bimud}
\safemath{\bmtheta}{\bithetad}
\safemath{\bmphi}{\biphid}
\safemath{\bmdelta}{\bideltad}

\safemath{\bA}{\mathbf{A}}
\safemath{\bB}{\mathbf{B}}
\safemath{\bC}{\mathbf{C}}
\safemath{\bD}{\mathbf{D}}
\safemath{\bE}{\mathbf{E}}
\safemath{\bF}{\mathbf{F}}
\safemath{\bG}{\mathbf{G}}
\safemath{\bH}{\mathbf{H}}
\safemath{\bI}{\mathbf{I}}
\safemath{\bJ}{\mathbf{J}}
\safemath{\bK}{\mathbf{K}}
\safemath{\bL}{\mathbf{L}}
\safemath{\bM}{\mathbf{M}}
\safemath{\bN}{\mathbf{N}}
\safemath{\bO}{\mathbf{O}}
\safemath{\bP}{\mathbf{P}}
\safemath{\bQ}{\mathbf{Q}}
\safemath{\bR}{\mathbf{R}}
\safemath{\bS}{\mathbf{S}}
\safemath{\bT}{\mathbf{T}}
\safemath{\bU}{\mathbf{U}}
\safemath{\bV}{\mathbf{V}}
\safemath{\bW}{\mathbf{W}}
\safemath{\bX}{\mathbf{X}}
\safemath{\bY}{\mathbf{Y}}
\safemath{\bZ}{\mathbf{Z}}

\safemath{\bZero}{\mathbf{0}}
\safemath{\bOne}{\mathbf{1}}
\safemath{\bDelta}{\mathbf{\Delta}}
\safemath{\bLambda}{\mathbf{\UpLambda}}
\safemath{\bPhi}{\mathbf{\Upphi}}
\safemath{\bSigma}{\mathbf{\Upsigma}}
\safemath{\bOmega}{\mathbf{\Upomega}}
\safemath{\bTheta}{\mathbf{\Uptheta}}

\bmdefine{\biAd}{A}
\bmdefine{\biBd}{B}
\bmdefine{\biCd}{C}
\bmdefine{\biDd}{D}
\bmdefine{\biEd}{E}
\bmdefine{\biFd}{F}
\bmdefine{\biGd}{G}
\bmdefine{\biHd}{H}
\bmdefine{\biId}{I}
\bmdefine{\biJd}{J}
\bmdefine{\biKd}{K}
\bmdefine{\biLd}{L}
\bmdefine{\biMd}{M}
\bmdefine{\biOd}{N}
\bmdefine{\biPd}{O}
\bmdefine{\biQd}{P}
\bmdefine{\biRd}{R}
\bmdefine{\biSd}{S}
\bmdefine{\biTd}{T}
\bmdefine{\biUd}{U}
\bmdefine{\biVd}{V}
\bmdefine{\biWd}{W}
\bmdefine{\biXd}{X}
\bmdefine{\biYd}{Y}
\bmdefine{\biZd}{Z}

\bmdefine{\biDelta}{\Delta}
\bmdefine{\biLambda}{\Lambda}
\bmdefine{\biPhi}{\Phi}
\bmdefine{\biSigma}{\Sigma}
\bmdefine{\biOmega}{\Omega}
\bmdefine{\biTheta}{\Theta}

\safemath{\bimA}{\biAd}
\safemath{\bimB}{\biBd}
\safemath{\bimC}{\biCd}
\safemath{\bimD}{\biDd}
\safemath{\bimE}{\biEd}
\safemath{\bimF}{\biFd}
\safemath{\bimG}{\biGd}
\safemath{\bimH}{\biHd}
\safemath{\bimI}{\biId}
\safemath{\bimJ}{\biJd}
\safemath{\bimK}{\biKd}
\safemath{\bimL}{\biLd}
\safemath{\bimM}{\biMd}
\safemath{\bimN}{\biNd}
\safemath{\bimO}{\biOd}
\safemath{\bimP}{\biPd}
\safemath{\bimQ}{\biQd}
\safemath{\bimR}{\biRd}
\safemath{\bimS}{\biSd}
\safemath{\bimT}{\biTd}
\safemath{\bimU}{\biUd}
\safemath{\bimV}{\biVd}
\safemath{\bimW}{\biWd}
\safemath{\bimX}{\biXd}
\safemath{\bimY}{\biYd}
\safemath{\bimZ}{\biZd}

\safemath{\bimDelta}{\biDelta}
\safemath{\bimLambda}{\biLambda}
\safemath{\bimPhi}{\biPhi}
\safemath{\bimSigma}{\biSigma}
\safemath{\bimOmega}{\biOmega}
\safemath{\bimTheta}{\biTheta}

\safemath{\setA}{\mathcal{A}}
\safemath{\setB}{\mathcal{B}}
\safemath{\setC}{\mathcal{C}}
\safemath{\setD}{\mathcal{D}}
\safemath{\setE}{\mathcal{E}}
\safemath{\setF}{\mathcal{F}}
\safemath{\setG}{\mathcal{G}}
\safemath{\setH}{\mathcal{H}}
\safemath{\setI}{\mathcal{I}}
\safemath{\setJ}{\mathcal{J}}
\safemath{\setK}{\mathcal{K}}
\safemath{\setL}{\mathcal{L}}
\safemath{\setM}{\mathcal{M}}
\safemath{\setN}{\mathcal{N}}
\safemath{\setO}{\mathcal{O}}
\safemath{\setP}{\mathcal{P}}
\safemath{\setQ}{\mathcal{Q}}
\safemath{\setR}{\mathcal{R}}
\safemath{\setS}{\mathcal{S}}
\safemath{\setT}{\mathcal{T}}
\safemath{\setU}{\mathcal{U}}
\safemath{\setV}{\mathcal{V}}
\safemath{\setW}{\mathcal{W}}
\safemath{\setX}{\mathcal{X}}
\safemath{\setY}{\mathcal{Y}}
\safemath{\setZ}{\mathcal{Z}}
\safemath{\emptySet}{\varnothing}

\safemath{\colA}{\mathscr{A}}
\safemath{\colB}{\mathscr{B}}
\safemath{\colC}{\mathscr{C}}
\safemath{\colD}{\mathscr{D}}
\safemath{\colE}{\mathscr{E}}
\safemath{\colF}{\mathscr{F}}
\safemath{\colG}{\mathscr{G}}
\safemath{\colH}{\mathscr{H}}
\safemath{\colI}{\mathscr{I}}
\safemath{\colJ}{\mathscr{J}}
\safemath{\colK}{\mathscr{K}}
\safemath{\colL}{\mathscr{L}}
\safemath{\colM}{\mathscr{M}}
\safemath{\colN}{\mathscr{N}}
\safemath{\colO}{\mathscr{O}}
\safemath{\colP}{\mathscr{P}}
\safemath{\colQ}{\mathscr{Q}}
\safemath{\colR}{\mathscr{R}}
\safemath{\colS}{\mathscr{S}}
\safemath{\colT}{\mathscr{T}}
\safemath{\colU}{\mathscr{U}}
\safemath{\colV}{\mathscr{V}}
\safemath{\colW}{\mathscr{W}}
\safemath{\colX}{\mathscr{X}}
\safemath{\colY}{\mathscr{Y}}
\safemath{\colZ}{\mathscr{Z}}

\safemath{\opA}{\mathbb{A}}
\safemath{\opB}{\mathbb{B}}
\safemath{\opC}{\mathbb{C}}
\safemath{\opD}{\mathbb{D}}
\safemath{\opE}{\mathbb{E}}
\safemath{\opF}{\mathbb{F}}
\safemath{\opG}{\mathbb{G}}
\safemath{\opH}{\mathbb{H}}
\safemath{\opI}{\mathbb{I}}
\safemath{\opJ}{\mathbb{J}}
\safemath{\opK}{\mathbb{K}}
\safemath{\opL}{\mathbb{L}}
\safemath{\opM}{\mathbb{M}}
\safemath{\opN}{\mathbb{N}}
\safemath{\opO}{\mathbb{O}}
\safemath{\opP}{\mathbb{P}}
\safemath{\opQ}{\mathbb{Q}}
\safemath{\opR}{\mathbb{R}}
\safemath{\opS}{\mathbb{S}}
\safemath{\opT}{\mathbb{T}}
\safemath{\opU}{\mathbb{U}}
\safemath{\opV}{\mathbb{V}}
\safemath{\opW}{\mathbb{W}}
\safemath{\opX}{\mathbb{X}}
\safemath{\opY}{\mathbb{Y}}
\safemath{\opZ}{\mathbb{Z}}
\safemath{\opZero}{\mathbb{O}}
\safemath{\identityop}{\opI}

\safemath{\veca}{\bma}
\safemath{\vecb}{\bmb}
\safemath{\vecc}{\bmc}
\safemath{\vecd}{\bmd}
\safemath{\vece}{\bme}
\safemath{\vecf}{\bmf}
\safemath{\vecg}{\bmg}
\safemath{\vech}{\bmh}
\safemath{\veci}{\bmi}
\safemath{\vecj}{\bmj}
\safemath{\veck}{\bmk}
\safemath{\vecl}{\bml}
\safemath{\vecm}{\bmm}
\safemath{\vecn}{\bmn}
\safemath{\veco}{\bmo}
\safemath{\vecp}{\bmp}
\safemath{\vecq}{\bmq}
\safemath{\vecr}{\bmr}
\safemath{\vecs}{\bms}
\safemath{\vect}{\bmt}
\safemath{\vecu}{\bmu}
\safemath{\vecv}{\bmv}
\safemath{\vecw}{\bmw}
\safemath{\vecx}{\bmx}
\safemath{\vecy}{\bmy}
\safemath{\vecz}{\bmz}

\safemath{\veczero}{\bmzero}
\safemath{\vecone}{\bmone}
\safemath{\vecxi}{\bmxi}
\safemath{\veclambda}{\bmlambda}
\safemath{\vecmu}{\bmmu}
\safemath{\vectheta}{\bmtheta}
\safemath{\vecphi}{\bmphi}
\safemath{\vecdelta}{\bmdelta}

\safemath{\matA}{\bA}
\safemath{\matB}{\bB}
\safemath{\matC}{\bC}
\safemath{\matD}{\bD}
\safemath{\matE}{\bE}
\safemath{\matF}{\bF}
\safemath{\matG}{\bG}
\safemath{\matH}{\bH}
\safemath{\matI}{\bI}
\safemath{\matJ}{\bJ}
\safemath{\matK}{\bK}
\safemath{\matL}{\bL}
\safemath{\matM}{\bM}
\safemath{\matN}{\bN}
\safemath{\matO}{\bO}
\safemath{\matP}{\bP}
\safemath{\matQ}{\bQ}
\safemath{\matR}{\bR}
\safemath{\matS}{\bS}
\safemath{\matT}{\bT}
\safemath{\matU}{\bU}
\safemath{\matV}{\bV}
\safemath{\matW}{\bW}
\safemath{\matX}{\bX}
\safemath{\matY}{\bY}
\safemath{\matZ}{\bZ}
\safemath{\matzero}{\bmzero}

\safemath{\matDelta}{\bDelta}
\safemath{\matLambda}{\bLambda}
\safemath{\matPhi}{\bPhi}
\safemath{\matSigma}{\bSigma}
\safemath{\matOmega}{\bOmega}
\safemath{\matTheta}{\bTheta}

\safemath{\matidentity}{\matI}
\safemath{\matone}{\matO}

\safemath{\rnda}{A}
\safemath{\rndb}{B}
\safemath{\rndc}{C}
\safemath{\rndd}{D}
\safemath{\rnde}{E}
\safemath{\rndf}{F}
\safemath{\rndg}{G}
\safemath{\rndh}{H}
\safemath{\rndi}{I}
\safemath{\rndj}{J}
\safemath{\rndk}{K}
\safemath{\rndl}{L}
\safemath{\rndm}{M}
\safemath{\rndn}{N}
\safemath{\rndo}{O}
\safemath{\rndp}{P}
\safemath{\rndq}{Q}
\safemath{\rndr}{R}
\safemath{\rnds}{S}
\safemath{\rndt}{T}
\safemath{\rndu}{U}
\safemath{\rndv}{V}
\safemath{\rndw}{W}
\safemath{\rndx}{X}
\safemath{\rndy}{Y}
\safemath{\rndz}{Z}

\safemath{\rveca}{\bimA}
\safemath{\rvecb}{\bimB}
\safemath{\rvecc}{\bimC}
\safemath{\rvecd}{\bimD}
\safemath{\rvece}{\bimE}
\safemath{\rvecf}{\bimF}
\safemath{\rvecg}{\bimG}
\safemath{\rvech}{\bimH}
\safemath{\rveci}{\bimI}
\safemath{\rvecj}{\bimJ}
\safemath{\rveck}{\bimK}
\safemath{\rvecl}{\bimL}
\safemath{\rvecm}{\bimM}
\safemath{\rvecn}{\bimN}
\safemath{\rveco}{\bomO}
\safemath{\rvecp}{\bimP}
\safemath{\rvecq}{\bimQ}
\safemath{\rvecr}{\bimR}
\safemath{\rvecs}{\bimS}
\safemath{\rvect}{\bimT}
\safemath{\rvecu}{\bimU}
\safemath{\rvecv}{\bimV}
\safemath{\rvecw}{\bimW}
\safemath{\rvecx}{\bimX}
\safemath{\rvecy}{\bimY}
\safemath{\rvecz}{\bimZ}

\safemath{\rvecxi}{\bmxi}
\safemath{\rveclambda}{\bmlambda}
\safemath{\rvecmu}{\bmmu}
\safemath{\rvectheta}{\bmtheta}
\safemath{\rvecphi}{\bmphi}

\safemath{\rmatA}{\bimA}
\safemath{\rmatB}{\bimB}
\safemath{\rmatC}{\bimC}
\safemath{\rmatD}{\bimD}
\safemath{\rmatE}{\bimE}
\safemath{\rmatF}{\bimF}
\safemath{\rmatG}{\bimG}
\safemath{\rmatH}{\bimH}
\safemath{\rmatI}{\bimI}
\safemath{\rmatJ}{\bimJ}
\safemath{\rmatK}{\bimK}
\safemath{\rmatL}{\bimL}
\safemath{\rmatM}{\bimM}
\safemath{\rmatN}{\bimN}
\safemath{\rmatO}{\bimO}
\safemath{\rmatP}{\bimP}
\safemath{\rmatQ}{\bimQ}
\safemath{\rmatR}{\bimR}
\safemath{\rmatS}{\bimS}
\safemath{\rmatT}{\bimT}
\safemath{\rmatU}{\bimU}
\safemath{\rmatV}{\bimV}
\safemath{\rmatW}{\bimW}
\safemath{\rmatX}{\bimX}
\safemath{\rmatY}{\bimY}
\safemath{\rmatZ}{\bimZ}

\safemath{\rmatDelta}{\bimDelta}
\safemath{\rmatLambda}{\bimLambda}
\safemath{\rmatPhi}{\bimPhi}
\safemath{\rmatSigma}{\bimSigma}
\safemath{\rmatOmega}{\bimOmega}
\safemath{\rmatTheta}{\bimTheta}

%% file: macros/standard-macros.tex
\usepackage{amssymb}
\usepackage{amsfonts}
\usepackage{mathrsfs}
\usepackage{xspace}
\usepackage{bm}
\usepackage{fancyref}
\usepackage{textcomp}

\usepackage{multirow}
\usepackage{stmaryrd}

\newenvironment{textbmatrix}{	\setlength{\arraycolsep}{2.5pt}%
								\big[\begin{matrix}}{\end{matrix}\big]%
								\raisebox{0.08ex}{\vphantom{M}}}

\def\be{\begin{equation}}
\def\ee{\end{equation}}
\def\een{\nonumber \end{equation}}
\def\mat{\begin{bmatrix}}
\def\emat{\end{bmatrix}}
\def\btm{\begin{textbmatrix}}
\def\etm{\end{textbmatrix}}

\def\ba#1\ea{\begin{align}#1\end{align}}
\def\bas#1\eas{\begin{align*}#1\end{align*}}
\def\bs#1\es{\begin{split}#1\end{split}}
\def\bg#1\eg{\begin{gather}#1\end{gather}}
\def\bml#1\eml{\begin{multline}#1\end{multline}}
\def\bi#1\ei{\begin{itemize}#1\end{itemize}}

\newcommand{\lefto}{\mathopen{}\left}

\DeclareMathOperator*{\argmax}{arg\;max}		% arg max
\DeclareMathOperator{\Exop}{\opE}			% expectation operator
\DeclareMathOperator{\Varop}{\opV\!\mathrm{ar}} % variance operator
\newcommand{\abs}[1]{\lefto\lvert#1\right\rvert}		% absolute value

\newcommand{\vecnorm}[1]{\lefto\lVert#1\right\rVert}		% vector norm
\newcommand{\est}[1]{\ensuremath{\hat{#1}}} 	% estimate
\safemath{\dirac}{\delta}					% Dirac delta
\safemath{\krond}{\dirac}					% Kronecker delta
\safemath{\upto}{\uparrow}
\safemath{\downto}{\downarrow}
\safemath{\iu}{j}							% imaginary unit
\safemath{\ev}{\lambda}						% eigenvalue
\safemath{\hilseqspace}{l^{2}}				% Hilbert sequence space
\newcommand{\banachfunspace}[1]{\setL^{#1}}	% Banach function space
\safemath{\hilfunspace}{\banachfunspace{2}}	% Hilbert function space
\safemath{\SNR}{\textit{SNR}} 				% signal to noise ratio
\safemath{\PAR}{\textit{PAR}} 				% signal to noise ratio
\safemath{\No}{N_0}							% noise spectral density
\safemath{\Es}{E_s}							% energy per symbol
\safemath{\Eb}{E_b}							% energy per bit
\safemath{\EbNo}{\frac{\Eb}{\No}}
\safemath{\EsNo}{\frac{\Es}{\No}}

\DeclareMathOperator{\CHop}{\ensuremath{\opH}} % channel operator
\safemath{\tvir}{\rndh_{\CHop}}				% time-varying impulse response
\safemath{\tvtf}{\rndl_{\CHop}}				% 	-''- transfer function
\safemath{\spf}{\rnds_{\CHop}}				% spreading function
\safemath{\bff}{H_{\CHop}}					% bi-freuqency function

\safemath{\ircf}{r_{h}}						% impulse response correlation fn.
\safemath{\tftvcf}{r_{s}}					% scattering function
\safemath{\tfcf}{r_{l}}						% time-frequency correlation fn.
\safemath{\bfcf}{r_{H}}						% bi-frequency correlation fn.

\safemath{\tcorr}{c_h}						% time-correlation function
\safemath{\scf}{c_{s}}						% spreading function
\safemath{\tfcorr}{c_{l}}					% transfer-function correlation
\safemath{\fcorr}{c_{H}}						% frequency-correlation function

\safemath{\mi}{I}							% mutual information
\safemath{\capacity}{C}						% capacity

\safemath{\normal}{\mathcal{N}}			% normal distribution
\safemath{\jpg}{\mathcal{CN}}			% jointly proper Gaussian
\safemath{\mchain}{\leftrightarrow}		% Markov chain
\safemath{\dB}{\,\mathrm{dB}}
\safemath{\dBm}{\,\mathrm{dBm}}
\safemath{\Hz}{\,\mathrm{Hz}}
\safemath{\kHz}{\,\mathrm{kHz}}
\safemath{\MHz}{\,\mathrm{MHz}}
\safemath{\GHz}{\,\mathrm{GHz}}
\safemath{\s}{\,\mathrm{s}}
\safemath{\ms}{\,\mathrm{ms}}
\safemath{\mus}{\,\mathrm{\text{\textmu}s}}
\safemath{\ns}{\,\mathrm{ns}}
\safemath{\ps}{\,\mathrm{ps}}
\safemath{\meter}{\,\mathrm{m}}
\safemath{\mm}{\,\mathrm{mm}}
\safemath{\cm}{\,\mathrm{cm}}
\safemath{\m}{\,\mathrm{m}}
\safemath{\W}{\,\mathrm{W}}
\safemath{\mW}{\, \mathrm{mW}}
\safemath{\J}{\,\mathrm{J}}
\safemath{\K}{\,\mathrm{K}}
\safemath{\bit}{\,\mathrm{bit}}
\safemath{\nat}{\,\mathrm{nat}}

\safemath{\define}{\triangleq}			% definition

\safemath{\equivalent}{\sim}
\safemath{\distas}{\sim}					% distributed according to
\safemath{\sdiff}{\Delta}				% symmetric set difference

\safemath{\reals}{\mathbb{R}}
\safemath{\positivereals}{\reals_{+}}
\safemath{\integers}{\mathbb{Z}}
\safemath{\posint}{\integers_{+}}
\safemath{\naturals}{\mathbb{N}}
\safemath{\posnaturals}{\naturals_{+}}
\safemath{\complexset}{\mathbb{C}}
\safemath{\rationals}{\mathbb{Q}}

\newcommand*{\fancyrefapplabelprefix}{app}		% Appendix
\newcommand*{\fancyrefthmlabelprefix}{thm}		% Theorem
\newcommand*{\fancyreflemlabelprefix}{lem}		% Lemma
\newcommand*{\fancyrefcorlabelprefix}{cor}		% Corollary
\newcommand*{\fancyrefdeflabelprefix}{def}		% Definition
\newcommand*{\fancyrefproplabelprefix}{prop}		% Proposition
\newcommand*{\fancyrefexmpllabelprefix}{exmpl}
\newcommand*{\fancyrefalglabelprefix}{alg}		% Algorithm
\newcommand*{\fancyreftbllabelprefix}{tbl}    % Algorithm
\newcommand*{\fancyrefasmlabelprefix}{asm}    % Assumption

\frefformat{vario}{\fancyrefasmlabelprefix}{Assumption~#1}
\frefformat{vario}{\fancyrefseclabelprefix}{Section~#1}
\frefformat{vario}{\fancyrefthmlabelprefix}{Theorem~#1}
\frefformat{vario}{\fancyreftbllabelprefix}{Table~#1}
\frefformat{vario}{\fancyreflemlabelprefix}{Lemma~#1}
\frefformat{vario}{\fancyrefcorlabelprefix}{Corollary~#1}
\frefformat{vario}{\fancyrefdeflabelprefix}{Definition~#1}
\frefformat{vario}{\fancyreffiglabelprefix}{Fig.~#1}
\frefformat{vario}{\fancyrefapplabelprefix}{Appendix~#1}
\frefformat{vario}{\fancyrefeqlabelprefix}{(#1)}
\frefformat{vario}{\fancyrefproplabelprefix}{Proposition~#1}
\frefformat{vario}{\fancyrefexmpllabelprefix}{Example~#1}
\frefformat{vario}{\fancyrefalglabelprefix}{Algorithm~#1}

%% file: macros/defs.tex
 \newtheorem{thm}{Theorem}
 \newtheorem{cor}[thm]{Corollary}   % Turned off theorem numbering
 
 \newtheorem{defi}{Definition}
 \newtheorem{lem}{Lemma}
\safemath{\dictab}{[\,\dicta\,\,\dictb\,]}

\safemath{\ysig}{\bmy}
\safemath{\ysighat}{\hat{\ysig}}
\safemath{\ysigdim}{M}
\safemath{\xsig}{\bmx}
\safemath{\xsigdim}{N}
\safemath{\nx}{n_x}
\safemath{\zsig}{\bmz}
\safemath{\zsigdim}{\ysigdim}
\safemath{\rsig}{\bmr}
\safemath{\Adict}{\bA}
\safemath{\Adicttilde}{\widetilde{\Adict}}
\safemath{\Adictdim}{\outputdim\times\xsigdim}
\safemath{\avec}{\bma}
\safemath{\avectilde}{\tilde{\avec}}
\safemath{\Bdict}{\bB}
\safemath{\Bdicttilde}{\widetilde{\Bdict}}
\safemath{\Cdict}{\bC}
\safemath{\cvec}{\bmc}
\safemath{\Ddict}{\bD}
\safemath{\Ddictdim}{\ysigdim\times\xsigdim}
\safemath{\dvec}{\bmd}
\safemath{\Ddicttilde}{\widetilde{\bD}}
\safemath{\Bonb}{\bB}
\safemath{\bvec}{\bmb}
\safemath{\Bonbdim}{\ysigdim\times\ysigdim}
\safemath{\noise}{\bmn}
\safemath{\noisedim}{\ysigim}
\safemath{\err}{\bme}
\safemath{\errdim}{\ysigdim}
\safemath{\errset}{\setE}
\safemath{\nerr}{n_e}
\safemath{\delop}{\bP_\errset}
\safemath{\delopc}{\bP_{{\errset}^c}}

\safemath{\cplxi}{\imath}
\safemath{\cplxj}{\jmath}
\safemath{\dict}{\matD}
\safemath{\inputdim}{N}		% number of columns of dictionary D
\safemath{\outputdim}{M}		%number of rows of dictionary D
\safemath{\sparsity}{S}	%sparsity
\safemath{\inputdimA}{{N_a}}	%total number of elements in dictionary A
\safemath{\inputdimB}{{N_b}}	%total number of elements in dictionary B
\safemath{\elemA}{{n_a}}	%number of elements chosen from dictionary A
\safemath{\elemB}{{n_b}}	%number of elements chosen from dictionary B
\safemath{\resA}{\matR_a}	%restriction map to elements of dictionary A
\safemath{\resB}{\matR_b}	%restriction map to elements of dictionary B
\safemath{\subD}{\matS} %subdictionary
\safemath{\subA}{\matS_a} %subdictionary part of A
\safemath{\subB}{\matS_b} %subdictionary part of B
\safemath{\dicta}{\matA} 	% first subdictionary
\safemath{\dictb}{\matB} 	% second subdictionary
\safemath{\hollowS}{H}
\safemath{\hollowA}{H_a}
\safemath{\hollowB}{H_b}
\safemath{\cross}{Z}
\safemath{\coh}{\mu_d}			% coherence dictionary
\safemath{\coha}{\mu_a}			% coherence first subdictionary
\safemath{\cohb}{\mu_b}			% coherence second subdictionary
\safemath{\mubs}{\nu}	%block sub-coherence
\safemath{\cohm}{\mu_m} %mutual coherence
\safemath{\dictset}{\setD}	% set of dictionaries
\safemath{\dictsetp}{\dictset(\coh,\coha,\cohb)}	% set of dictionaries parametrized
\safemath{\dictsetgen}{\dictset_\text{gen}}
\safemath{\dictsetgenp}{\dictsetgen(\coh)}
\safemath{\dictsetonb}{\dictset_\text{onb}}
\safemath{\dictsetonbp}{\dictsetonb(\coh)}

\safemath{\leftside}{U}
\safemath{\rightsideA}{R_a}
\safemath{\rightsideB}{R_b}

\safemath{\indexS}{\setI_S} %set of indices participating in sub-dictionary S

\safemath{\na}{n_a}			% cardinality of set of linearly independent columns of first dictionary
\safemath{\nb}{n_b}			% cardinality of set of linearly independent columns of second dictionary
\safemath{\coeffa}{p_i}	%coefficients for columns of A
\safemath{\coeffb}{q_j}	%coefficients for columns of B
\safemath{\seta}{\setP}		% set of linearly independent columns of A
\safemath{\setb}{\setQ}     % set of linearly independent columns of B
\safemath{\setw}{\setW}	%set of n largest elements of w
\safemath{\setz}{\setZ}	%set of L-n largest elements of z
\safemath{\cola}{\veca}		% generic element of the dictionary A
\safemath{\colb}{\vecb}		% generic element of the dictionary B
\safemath{\cold}{\vecd}		% generic element of the dictionary D
\safemath{\inputvec}{\vecx} 	%coefficient vector (input)
\safemath{\error}{\vece}	%error vector
\safemath{\noiseout}{\vecz} 	%noisy output vector
\safemath{\inputvecel}{x}
\safemath{\inputveca}{\vecx_a}
\safemath{\inputvecb}{\vecx_b}
\safemath{\outputvec}{\vecy}	%output of Dictionary
\safemath{\lambdamin}{\lambda_{\mathrm{min}}}
\safemath{\elltwo}{\ell_2}
\safemath{\ellone}{\ell_1}
\safemath{\ellzero}{\ell_0}
\safemath{\ellinf}{\ell_\infty}
\safemath{\ellinftilde}{\ell_{\widetilde\infty}}
\safemath{\licard}{Z(\coh,\coha,\cohb)}
\safemath{\xsol}{\hat{x}}
\safemath{\xbord}{x_b}		%Solution at the border
\safemath{\xstat}{x_s}		%Solution stationary in l0 prob
\safemath{\xstatLone}{\tilde{x}_s}
\safemath{\order}{\mathcal{O}} %order notation (big O)
\safemath{\scales}{\Theta} %scales as
\safemath{\ones}{\mathbf{1}} %all ones matrix
\safemath{\zeroes}{\mathbf{0}} %all zeroes matrix
\safemath{\thlone}{\kappa(\coh,\cohb)} %treshold l1 problem
\safemath{\constoneA}{\delta} %constant in l1 theorem to save space
\safemath{\constoneB}{\epsilon} %constant in l1 theorem to save space
\safemath{\nlarge}{L}				   %num large elements
\safemath{\sumlarge}{S_\nlarge}
\safemath{\maxlarger}{P_\nlarge}	   % maximum in Gribonval and Nielsen
\safemath{\Pzero}{\textrm{P0}}	
\safemath{\Pone}{\textrm{P1}}
\safemath{\vecfir}{\vecw}			 % \vecv element of the kernel of the dictionary, \vecv=[\vecfir \vecsec]
\safemath{\vecsec}{\vecz}
\safemath{\elvecfir}{w}              % element of vecfir
\safemath{\elvecsec}{z}				 % element of vecsec
\safemath{\nlargefir}{n}
\safemath{\normout}{\gamma}
\safemath{\auxfun}{h}
\safemath{\supp}{\textrm{supp}}%support

\safemath{\indexa}{\ell}
\safemath{\indexb}{r}
\safemath{\indexc}{i}
\safemath{\indexd}{j}

\safemath{\project}{P}%projector

%% file: sec1-intro.tex
% !TEX root = 18TIT-lama.tex

\section{Introduction}
\label{sec:intro}
We consider the problem of recovering the $\MT$-dimensional data vector $\vecs_0\in\setO^\MT$ from the noisy input-output relation \mbox{$\vecy=\bH\vecs_0+\bmn$}, by solving the  
% either 
individually-optimal (IO) data detection problem \cite{guo2003multiuser,GV2005}
\begin{align*}
(\IO) \quad s_\ell^\IO & = \argmax_{\tilde s_\ell\in\setO} \,p\!\left(\tilde s_\ell \,|\, \bmy, \bH\right)\!,
\quad \ell=1,2,\ldots,\MT,
\end{align*}
% or jointly-optimal (JO) data detection
% \begin{align*}
% (\JO) \quad \vecs^\JO & = \argmax_{\tilde \vecs\in\setO^\MT}\,p\!\left(\tilde\bms \,|\, \bmy, \bH\right)\!,
% \end{align*}
where $p\!\left(\tilde s_\ell \,|\, \bmy, \bH\right)$ 
% and $p\!\left(\tilde \bms \,|\, \bmy, \bH\right)$ are 
is the probability density function conditioned on observing the receive vector~$\vecy\in\complexset^\MR$ and assuming Gaussian noise for the noise vector \mbox{$\vecn\in\complexset^\MR$}.
The scalar $s_\ell^\IO$ corresponds to the $\ell$th IO estimate, $\setO$ is a finite constellation set (e.g., PAM, PSK, or QAM), $\MT$ and $\MR$ denote the number of transmitters and receivers, respectively,  and \mbox{$\bH\in\complexset^{\MR\times\MT}$} represents the (known)  multiple-input multiple-output (MIMO) channel matrix.

We develop a computationally efficient algorithm, referred to as \LAMA (short for \underline{la}rge \underline{M}IMO \underline{a}pproximate message passing), which is able to achieve the error-rate performance of the \IO data-detector under certain assumptions on the MIMO channel matrix and the constellation alphabet.
We show that in the large system limit, i.e., for $\beta=\MT/\MR$ and $\MT\to\infty$, and for i.i.d.\ Rayleigh fading MIMO channels, LAMA decouples the noisy MIMO system into a set of independent additive white Gaussian noise (AWGN) channels with equal signal-to-noise ratio (SNR); see \fref{fig:introfigure2} for an illustration of this decoupling property. 
LAMA is iterative in nature and enables one to compute the noise variance  $\sigma_t^2$ of each decoupled AWGN channel in each iteration $t$. This property allows for a precise analysis of the algorithm's performance (in terms of achievable rates and error rate) and complexity (in terms of the number of LAMA iterations).

Figure \ref{fig:introfigurea} demonstrates that \LAMA (i) is able to achieve the error-rate performance of the individually optimal detector for a square MIMO system (i.e.,  $\MR=\MT$) in the large-system limit, and (ii) closely approaches the error-rate performance of the IO data detector in finite-dimensional systems.
Furthermore, we can accurately characterize the performance/complexity trade-offs without the need for expensive system simulations; see \fref{fig:introfigureb} for an illustration.

%computing the input parameter $\gamma_t^2$ in the non-linear MMSE function $\mathsf{F}$ and the final noise variance  
%represents the decoupling capabilities of LAMA. In the large system limit (i.e., $\beta=\MT/\MR$ and $\MT\to\infty)$, the output of LAMA (see \fref{fig:introfigurec}) after $t$ iteration is equivalent to a parallel AWGN system (see \fref{fig:introfigured}) with the same signal-to-noise ratio. 

%
%Moreover, the final output function $\mathsf{F}$ achieves the non-linear minimum mean-squared (MMSE), which we will derive analytically. This shows that \LAMA is not only capable of decoupling the system into parallel channels, but also can achieve non-linear MMSE detection. Finally, our framework will enable us to analyze the theoretical performance of \LAMA by computing the input parameter $\gamma_t^2$ in the non-linear MMSE function $\mathsf{F}$ and the final noise variance $\sigma_t^2$ for the AWGN channel in \fref{fig:introfigured}. 

\begin{figure}[tp]
\centering
%
% \hspace*{0.4cm}
\subfigure[]{\includegraphics[width=0.35\textwidth]{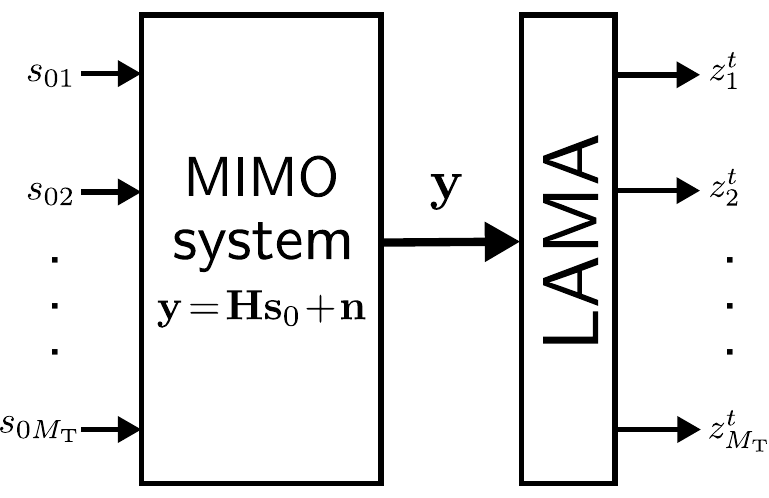}\label{fig:introfigurec}}
\hspace{0.4cm}
\subfigure[]{\includegraphics[width=0.35\textwidth]{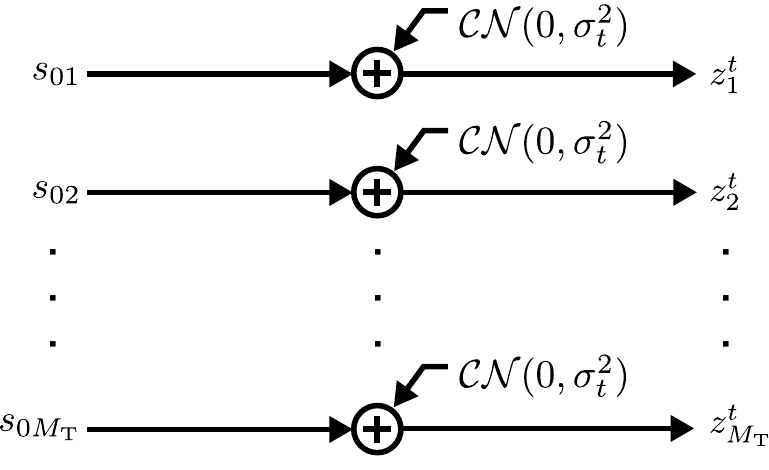}\label{fig:introfigured}}
\caption{Decoupling property of \LAMA. (a) Large MIMO system. (b) \LAMA decouples the system into parallel and independent AWGN channels with equal noise variance in the large-system limit  ($\beta=\MT/\MR$ and $\MT\to\infty)$.
% \LAMA also computes non-linear MMSE estimates $\hat{s}^{t+1}_i$, which correspond to the outputs of the function $\mathsf{F}(\cdot,\gamma_t^2)$. 
The state-evolution (SE) framework provides exact expressions for the AWGN noise variance $\sigma_t^2$ at iteration~$t$, which enables a precise analysis of LAMA's performance (in terms of achievable rates and error rate) and complexity (in terms of the number of algorithm iterations).} \label{fig:introfigure2}
\end{figure}

\begin{figure}[tp]
\centering
\subfigure[]{\includegraphics[height=0.25\textheight]{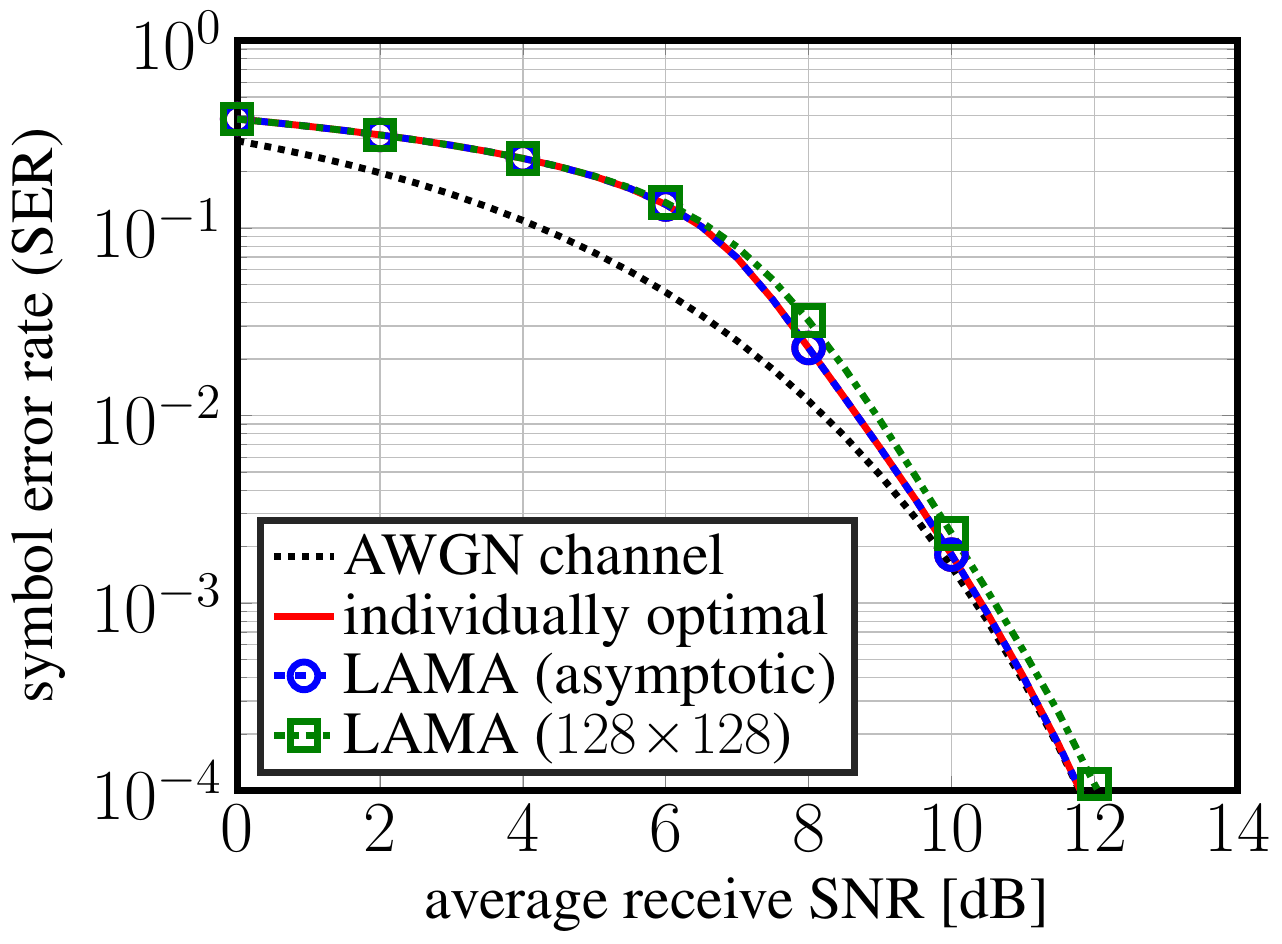}\label{fig:introfigurea}}
\hspace{0.4cm}
\subfigure[]{\includegraphics[height=0.25\textheight]{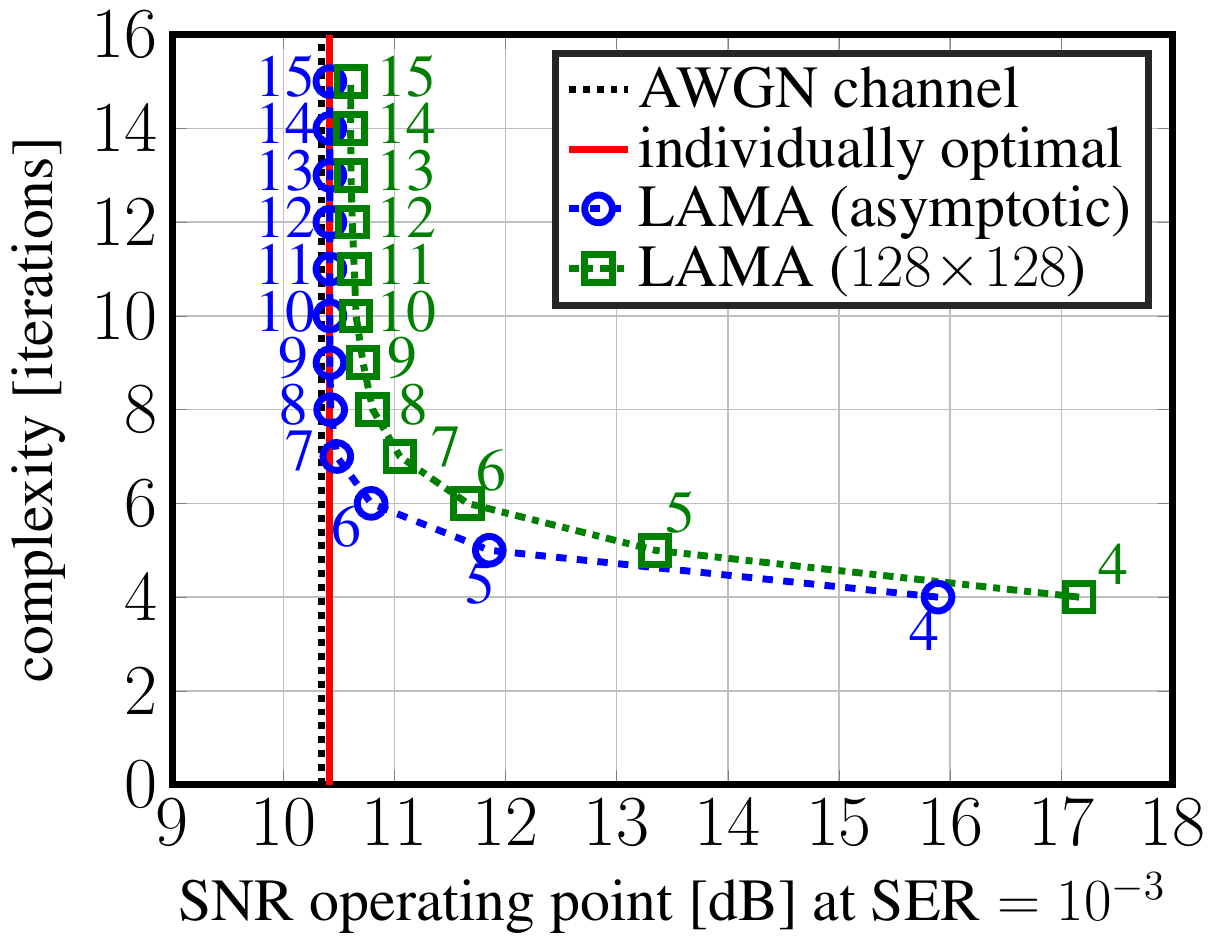}\label{fig:introfigureb}}
\caption{Capabilities of \LAMA in large MIMO with a square i.i.d.\ Gaussian system matrix and QPSK modulation. (a) Symbol error-rate (SER) in the large-system limit ($\beta=\MT/\MR=1$ and $\MT\to\infty$) compared to the optimal SER and the SER of an AWGN channel. 
\LAMA achieves the same error-rate performance as the IO data detector and approaches AWGN performance for sufficiently large SNR values; we also see that \LAMA closely approaches the theoretical performance limits for finite dimensions (i.e., for a $128\times128$ MIMO system). (b) Performance/complexity trade-off in the large-system limit (analytical) and for finite dimensions (simulated); a small number of LAMA iterations is sufficient  to approach the theoretical performance limits.} \label{fig:introfigure}
\end{figure}

\subsection{Application Examples}
The considered MIMO system model covers a variety of applications, including the following examples. 

\subsubsection{Massive Multi-User (MU) MIMO}

Massive MU-MIMO (also known as large-scale or full-dimensional MIMO) will be a key technology to meet the demands for higher spectral efficiency and quality-of-service-in fifth-generation (5G) wireless systems \cite{rusek2013scaling,LETM2014,ABCHLAZ2014}.
Massive MU-MIMO relies on  hundreds of antennas at the base-station (BS) that serve tens of users simultaneously and in the same frequency band. 
This technology
% Massive MU-MIMO 
promises significant gains in terms of spectral efficiency as well as lower operational power consumption compared to that of existing, small-scale MIMO systems \cite{LETM2014}. 
In addition, in the large BS-antenna limit, i.e., where $\MR\to\infty$ and the total number $\MT$ of user antennas remains constant, low-complexity data detection and precoding methods (such as the matched filter) turn out to be optimal \cite{Marzetta10}. 
However, as demonstrated in \cite{HBD11,WBVSCJD2013,WYWDCS2014}, practical (finite-dimensional) antenna configurations require more sophisticated data detection algorithms, which entail high computational complexity.
The proposed \LAMA algorithm enables high-performance and low-complexity data detection in practical massive MU-MIMO systems with higher-order modulation schemes, and allows for an accurate prediction of the fundamental performance/complexity trade-offs.

\subsubsection{Code-Division Multiple Access (CDMA)} 

CDMA is a classical transmission technology, in which multiple users simultaneously access a common resource (such as time or frequency) by modulating their individual information signals using spreading sequences \cite{L1991,V1995,GJPVWW1991,HP1997,LV1989}.
A significant portion of the CDMA literature studied the limits (such as the achievable rates for a given modulation scheme) of randomly spread CDMA. 
In the considered system model, the spreading matrix corresponds to $\bH$ with i.i.d.\ zero-mean Gaussian entries, $\MT$ denotes the number of users, and the spreading sequences are of length~$\MR$.
For common constellations (such as PAM, PSK, or QAM), we provide conditions
 that depend on the \emph{system ratio} $\beta=\MT/\MR$ (also known as the loading factor) in the large-system limit for which \LAMA achieves the same error-rate performance of the IO data detector \cite{GV2005}.
Our analysis recovers classical results from the CDMA literature \cite{GV2005,TanakaCDMA,K2003} while providing practical means for closely approaching these limits in finite-dimensional systems at low computational complexity.

\subsubsection{Finding Discrete Solutions to Systems of Linear Equations}
The considered system model also enables one to study the recovery of integer solutions to the (noisy) system of linear equations \mbox{$\vecy=\bH\vecs+\bmn$}. For noiseless observations, i.e., $\vecy=\bH\vecs$, and for the case of $\setO$ being (a subset of) the integers, \LAMA is able to perfectly recover $\vecs\in\setO^\MT$ provided that the  entries of the system matrix $\bH$ are i.i.d.\ zero-mean Gaussian distributed  and the system ratio $\beta=\MT/\MR$ does not exceed a certain \emph{exact recovery threshold (ERT)}. 
This result is relevant for solving systems of linear Diophantine equations, which finds, for example, use in number theory, cryptography, or closest vector problems in lattices; see \cite{Brown1958,contejean1994efficient,fincke1985improved,agrell2002closest} and the references therein. 

\subsection{Relevant Prior Results}
Early results on optimal data detection in large MIMO systems reach back to \cite{VS1999} where Verd\'u and Shamai analyzed the spectral efficiency of multi-user detectors in randomly-spread CDMA systems. 
The authors provided a precise characterization of the achievable rates with optimal data detection and demonstrated that the system's randomness (due to the random spreading sequences) disappears in the large-system limit.  Tanaka \cite{TanakaCDMA} derived analytical expressions for the error-rate performance and the multi-user efficiency (equivalent to the noise variance in a single AWGN channel) for the IO data detector using the {replica method} \cite{MPV1987}; Tanaka's results were obtained for BPSK constellations using the replica method in \cite{TanakaCDMA} and later proven rigorously in \cite{MT06}. 
Guo and Verd\'u provided an extension of these results to arbitrary discrete inputs \cite{GV2005}. Moreover, it was shown that for a certain family of multi-user detectors, referred as posterior mean estimators (PMEs), the communication system decouples into a set of parallel and independent AWGN channels with equal SNR \cite{TanakaCDMA,GV2005,GW2006,GW2007}. 
All of these results study the fundamental performance of IO detection in the large-system limit, i.e., for $\beta=\MT/\MR$ with $\MT\rightarrow\infty$. 
Corresponding practical algorithms have been proposed for  BPSK in real-valued systems \cite{K2003}, \cite{CMT2004}---in contrast, LAMA is a practical  algorithm  for
% no algorithms for 
general constellations and complex-valued systems, and enables a corresponding theoretical performance analysis.
 % have been proposed in the open literature. 

LAMA builds upon approximate message passing (AMP) \cite{donoho2009message,BM2011,Maleki2010phd}, which was initially proposed for sparse signal recovery and compressive sensing~\cite{Donoho06,CSintro_Candes08,MMB2014}.  
In the large-system limit, the estimates obtained by AMP correspond to the true signal perturbed by i.i.d.\ Gaussian noise \cite{andreaGMCS}.
In addition, the variance of the Gaussian random variables in each AMP iteration can be tracked exactly via the state evolution (SE) framework \cite{donoho2009message,BM2011}; this feature enables an exact performance analaysis.
AMP has been generalized to i.i.d. signal priors using the Bayesian AMP framework \cite{DMM10a,DMM10b,Maleki2010phd} and to sparse recovery in complex-valued systems \cite{malekiCAMP}. More recently, AMP and the SE framework have been extended to more general observation models in \cite{GAMP2011,JM2012,RSF2017}.  Within the last few years, AMP has been successfully deployed in a variety of applications~\cite{VS2013,KRFU2014,RSRFC2013,VS2011}, including signal restoration \cite{KGR2012,maechler2012vlsi}, imaging \cite{SS2012}, phase retrieval \cite{SR2012},  and de-noising \cite{MMB2014,MZB2014}.
AMP-related algorithms have also been used for data detection in many different communication systems \cite{NSE14,NL2014,S2011,WKNLHG14,KGR2012}. While these results showcase the potential of AMP for data detection in wireless systems, they lack of a rigorous performance analysis.
In this paper, we focus on a theoretical performance analysis of AMP for data detection and provide conditions for which it achieves IO performance.

\subsection{Contributions}

In this paper, we build upon complex-valued AMP \cite{malekiCAMP} and (real-valued) Bayesian AMP~\cite{Maleki2010phd} in order to develop the complex Bayesian AMP (cB-AMP) algorithm and its complex state evolution (cSE) framework. 
Our derivations incorporate the possibility of having a mismatch between the actual and a postulated noise variance.
We then specialize cB-AMP to data detection in large MIMO, resulting in the  \LAMA algorithm. Our key contributions are as follows.
\begin{itemize}
\item We study \LAMA in the massive MU-MIMO limit, i.e., when $\beta\to0$, and show that for such a scenario, simple low-complexity algorithms achieve IO performance.
\item We demonstrate that the SE recursions of \LAMA are identical to the fixed-point equations that predict the optimal multiuser efficiency developed in \cite{TanakaCDMA,GV2005,GW2006}. 
\item We develop conditions for which \LAMA achieves the same error-rate performance as the IO data detector.
\item We derive exact recovery thresholds (ERTs), for which \LAMA perfectly recovers signals from PAM, PSK, and QAM alphabets in noiseless systems.
\item We investigate the achievable rates and error-rate performance of \LAMA  for PAM, PSK, and QAM constellations, and analyze the impact of the system ratio $\beta$.
\item We characterize the performance/complexity trade-off of  \LAMA and show that only a few algorithm iterations are sufficient to achieve near-IO performance. 
\item We discuss the efficacy and limits of the proposed \LAMA algorithms in practical (finite-dimensional) large-MIMO systems and provide corresponding numerical results. 
\end{itemize}

\subsection{Notation}
Lowercase and uppercase boldface letters represent column vectors and matrices, respectively. For a matrix $\bH$, we define its transpose and Hermitian to be $\bH^\Tran$ and $\bH^\Herm$, respectively. The $\ell$th column and the $k$th row vector of the matrix $\bH$ are denoted by $\bmh_\ell^\col$ and $\bmh_k^\row$ respectively, the entry on the $k$th row and $\ell$th column is $H_{k,\ell}$, and the $k$th entry of a vector $\vecx$ is $x_k$. 
For a $N$-dimensional vector $\bmx$, we define its complex conjugate by $\bmx^*$ and its $k$th entry by~$x_k$. The $M\times M$ identity matrix is denoted by $\bI_M$ and the $M\times N$ all-zeros matrix by $\bm{0}_{M\times N}$. 
The real and imaginary parts of scalars, vectors, and matrices are denoted by $\realpart{\cdot}$ and $\imagpart{\cdot}$, respectively. We use $\left\langle\,\cdot\,\right\rangle$ to represent the averaging operator $\left\langle \bmx \right\rangle = \frac{1}{N}\sum_{k=1}^N x_k$. Multivariate real-valued and complex-valued Gaussian probability density (pdf) functions are denoted by $\setN(\bmm,\bK)$ and $\setC\setN(\bmm,\bK)$, respectively, where~$\bmm$ is the mean vector and $\bK$ the covariance matrix; $\Exop_X\!\left[\,\cdot\,\right]$ denotes expectation and $\Varop_X\!\left[\,\cdot\,\right]$ denotes  variance with respect to the pdf of the random variable~$X$.

\subsection{Paper Outline}

The rest of the paper is organized as follows. \fref{sec:systemmod} introduces the complex Bayesian AMP (cB-AMP) algorithm and the complex state evolution (cSE) framework. \fref{sec:DAMP} derives the LAMA algorithm. \fref{sec:SE_MI_ER} provides optimality conditions. \fref{sec:numerical} presents  corresponding numerical results and discusses practical considerations. \fref{sec:prior_art} summarizes prior art relevant to LAMA and \fref{sec:conclusion} concludes the paper. All proofs are relegated to the appendices.

%The remainder of the paper is organized as follows. In \fref{sec:systemmod}, the system model and complex Bayesian AMP is described. The proposed D-AMP receiver and its derivation is described in \fref{sec:DAMP}. The performance of D-AMP receiver by state evolution analysis and its verification by numerical simulations are shown in \fref{sec:SE_MI_ER}. Finally, conclusions and remarks are presented in \fref{sec:conclusion}.

%% file: sec2-system.tex
% !TEX root = 18TIT-lama.tex

\section{cB-AMP: Complex Bayesian AMP} \label{sec:systemmod}

We start by developing the complex Bayesian AMP (cB-AMP) framework which builds the foundation of the \LAMA algorithm developed in \fref{sec:DAMP}.
We specify our model assumptions, derive cB-AMP, and detail the complex-valued state-evolution (cSE) framework.

\subsection{System Model and Assumptions}

We estimate the complex-valued {data vector} $\vecs_0\in\complexset^\MT$ with known i.i.d.\ prior distribution $p(\vecs_0)=\prod_{\ell=1}^{\MT}p(s_{0\ell})$ from the following MIMO  input-output relation:
\begin{align} \label{eq:systemmodel}
\vecy = \bH\vecs_0+\vecn.
\end{align}
The number of transmitters and receivers are denoted by $\MT$ and $\MR$, respectively, and we do not impose any assumption on the so-called  \emph{system ratio} (also known as the loading factor in CDMA literature \cite{GV2003}), which we define as $\beta=\MT/\MR$. 
% 
% The large-system limit definition is given below in \fref{def:largesystemlimit}.
We will often use the following definition:

\begin{defi}\label{def:largesystemlimit}
% [Large-system limit] 
For a MIMO system with $\MT$ and $\MR$ transmitters and receivers respectively, we define the \emph{large-system limit} by fixing the system ratio $\beta=\MT/\MR$ and letting $\MT\to\infty$.
\end{defi}

In what follows, we will consider underdetermined ($\beta\leq1$) as well as  overdetermined ($\beta>1$) systems. 
% in massive MU-MIMO wireless systems, one typically assumes $\beta\ll1$, whereas in CDMA systems the case $\beta\geq 1$ is also relevant.
%
The {receive vector} in \eqref{eq:systemmodel} is given by $\vecy\in\complexset^\MR$ and the entries of the {noise vector} \mbox{$\vecn\in\complexset^\MR$} are assumed to be  i.i.d.\ circularly-symmetric complex Gaussian with variance $\No$ per complex entry.
The MIMO system matrix \mbox{$\bH\in\complexset^{\MR\times\MT}$} is assumed to be perfectly known to the receiver.
We will frequently use of the following assumptions on the MIMO system matrix~$\bH$~\cite{andreaGMCS}:
\sloppy
\begin{description}
\item[(A1)] The entries of $\bH$ are normalized so that the columns are zero mean and have unit $\ell_2$-norm; the real and imaginary parts are independent with identical variance.
Furthermore, all entries have similar magnitude $O(1/\sqrt{\MR})$ and are pairwise independent. 
%We will use the pairwise independence criterion to derive cB-AMP.
\item[(A2)] The entries $H_{k,\ell}\sim\setC\setN(0,1/\MR)$ are i.i.d.\ circularly-symmetric complex Gaussian.
\end{description}
\fussy
We note that (A2) implies (A1) in the large-system limit; see \cite{andreaGMCS} for the details. 
Throughout the paper, we define the average receive signal-to-noise-ratio \SNR as:
\begin{align}\label{eq:SNR}
\SNR = \frac{\Exop[\vecnorm{\bH\bms_0}^2_2]}{\Exop[\vecnorm{\bmn}^2_2]} 
% = \frac{\MT}{\MR}\frac{E_s}{\No} 
= \beta\frac{E_s}{\No},
\end{align}
where $E_s=\Exop[\abs{s_{0\ell}}^2]$ for all $\ell=1,\ldots,\MT$.
\sloppy
We also consider the case in which the receiver assumes the following (possibly) mismatched input-output relation:
\begin{align} \label{eq:systemmodel2}
\vecy = \bH\vecs_0+\vecn^\text{post}.
\end{align}
Here, $\bmn^\text{post}\sim\setC\setN(\mathbf{0}_{\MR\times1}, \Nopost\bI_\MR)$ models noise with postulated noise variance $\Nopost$ (not necessarily equal to $\No$). The model in \eqref{eq:systemmodel2} allows us to analyze a mismatch between the true noise variance $\No$ and the postulated noise variance~$\Nopost$ assumed by the detector. The case $\No = \Nopost$ corresponds to an ideal system with perfect knowledge of the noise variance. 

\fussy

\subsection{Complex Bayesian AMP (cB-AMP)}\label{sec:ampprior}

To arrive at an efficient algorithm that achieves the same error-rate performance as the IO data detector, we start with the Bayesian AMP (B-AMP) algorithmm proposed in \cite{DMM10a,DMM10b,Maleki2010phd} to obtain a marginalized distribution $p(\tilde{s}_{\ell}\vert\bmy,\bH)$ for each stream $\ell$ (also called layer).  
With the marginalized distribution, B-AMP enables the estimation of a vector~$\vecs_0$ from a real-valued version of the system model \fref{eq:systemmodel}. 
While B-AMP can---in certain cases---be applied to complex-valued systems using the well-known real-valued decomposition\footnote{The complex-valued model \fref{eq:systemmodel} can be rewritten as the following real-valued model: 
\begin{align*}
\left[\!\!
\begin{array}{c}
\realpart{\bmy}\\
\imagpart{\bmy}
\end{array}
\!\!
\right]
= 
\left[\!\!
\begin{array}{cr}
\realpart{\bH} \!\!\!\! & -\imagpart{\bH} \\
\imagpart{\bH} \!\!\!\! & \realpart{\bH}
\end{array}\!\!
\right]
\left[\!\!
\begin{array}{c}
\realpart{\bms}\\
\imagpart{\bms}
\end{array}
\!\!
\right]
+
\left[\!\!
\begin{array}{c}
\realpart{\bmn}\\
\imagpart{\bmn}
\end{array}
\!\!
\right].
\end{align*}},
the effective, real-valued system matrix $\overline\bH\in\reals^{2\MR\times2\MT}$ 
(i) violates the independence assumptions on the entires of $\bH$ of (A1),
% , which may prevent a rigorous analysis of the algorithm's performance, 
and (ii) prevents the use of non-separable symbol alphabets, such as phase-shift keying (PSK) constellations. 
To overcome both of these drawbacks, we develop a complex-valued version of B-AMP, which we refer to as cB-AMP.
% 
% \sloppy
We start with Bayes' rule and factorize  
\begin{align} \label{eq:factoredMAP}
 p\!\left(\bmy \,|\, \bms, \bH\right)p\!\left(\bms\right)
= \prod_{k=1}^{\MR} p\!\left(y_k\,|\,\bms,\vech^\row_k\right) \prod_{\ell=1}^{\MT}p\!\left(s_\ell\right)\!,
\end{align}
where we assume (i) complex Gaussian noise with postulated noise variance $\Nopost$ 
given by
\begin{align*}
p\!\left(y_k\,|\,\bms,\vech^\row_k\right) &= \frac{1}{Z}\exp\!\left(-\frac{1}{\Nopost}\left| y_k - \bmh_k^\row\bms\right|^2\right)\!,
\end{align*}
with the constant $Z$ so that $\int_\complexset p\!\left(y_k \,|\, \bms, \vech^\row_k\right) \text{d}y_k = 1$, and (ii) that the transmitted symbols are i.i.d. 

To arrive at an efficient inference method, we deploy the sum-product message-passing algorithm~\cite{MM2009book}.
% , which is known to provide excellent approximations of the marginal distributions of \fref{eq:factoredMAP} for numerous applications. 
However, as noted in \cite{donoho2009message}, a corresponding full-fledged message passing scheme is impractical. Hence, as in \cite{andreaGMCS,Maleki2010phd}, we simplify the algorithm by assuming a Gaussian distribution for the marginal densities of the messages $p(\shate_\ell\,|\,s_\ell,\tau)\sim \setC\setN(s_\ell,\tau)$ so that~\cite{andreaGMCS}
\begin{align}\notag
f\!\left(s_\ell \,|\, \shate_\ell,\tau\right) &= \frac{p(\shate_\ell\,|\,s_\ell,\tau)p(s_\ell)}{p(\shate_\ell,\tau)}\\
&= \frac{1}{Z'} \exp\!\left(-\frac{1}{\tau}\!\left\vert s_\ell - \shate_\ell\right\vert^2 \right)\! p\!\left(s_\ell\right)\!\label{eq:gaussianmeasure},
\end{align}
with the normalization constant $Z'$. We denote the conditional mean $\mathsf{F}\!\left(\shate_\ell,\tau \right)$ and variance $\mathsf{G}\!\left(\shate_\ell,\tau \right)$ of a random variable $S$ distributed according to \fref{eq:gaussianmeasure} as the message mean and message variance; both quantities are defined as follows:
\begin{align}\label{eq:F_compute}\mathsf{F}\!\left(\shate_\ell,\tau \right)&=\Exop_S[S\,|\, \shate_\ell,\tau] \\
\label{eq:G_compute}\mathsf{G}\!\left(\shate_\ell,\tau \right)&=\Varop_S[S\,|\, \shate_\ell,\tau].
\end{align}
% and variance by $\mathsf{F}\!\left(\shate_\ell,\tau \right)=\Exop_S[S\,|\, \shate_\ell,\tau]$ and $\mathsf{G}\!\left(\shate_\ell,\tau \right)=\Varop_S[S\,|\, \shate_\ell,\tau]$, 
%respectively, where.

With the methods developed in \cite{Maleki2010phd,malekiCAMP}, we can simplify the sum-product message-passing computations for \fref{eq:factoredMAP} which stems from the Gaussian assumption for the marginal densities of the messages.
We refer to the resulting algorithm as complex Bayesian AMP (cB-AMP), which is summarized below  (and derived in detail in \fref{app:damp-complex}):

\newtheorem{alg}{Algorithm}
\begin{alg}
% [Complex Bayesian AMP]
\label{alg:cB-AMP} Suppose that $\bH$ satisfies (A1) and \cite[Lem.~5.56]{Maleki2010phd} holds. Then, the \emph{complex Bayesian AMP (cB-AMP)} algorithm performs the following steps for each iteration $t=1,2,\ldots,$:
\begin{align}\nonumber
\shat^{t+1} &= \, \mathsf{F}\!\left(\shat^{t} + \bH^\Herm\resid^t,\Nopost(1+\tau^t)\right)\\\nonumber
\resid^{t+1} &= \, \bmy - \bH \shat^{t+1}\\\nonumber
&\quad+ \frac{\beta\resid^{t}}{2}\!\left\langle \!\left(\partial_1\mathsf{F}^\textnormal{R} + \partial_2\mathsf{F}^\textnormal{I} \right)\!\left(\shat^{t} + \bH^\Herm\resid^t,\Nopost(1+\tau^t)\right) \right\rangle \\\label{eq:cb_amp_resid}
&\quad- i\frac{\beta\resid^{t}}{2}\!\left\langle \!\left(\partial_2\mathsf{F}^\textnormal{R} - \partial_1\mathsf{F}^\textnormal{I} \right)\!\left(\shat^{t} + \bH^\Herm\resid^t,\Nopost(1+\tau^t)\right) \right\rangle\\\nonumber
\tau^{t+1} &=\, \frac{\beta}{\Nopost}\!\left\langle \mathsf{G}\left( \shat^{t} + \bH^\Herm \resid^t ,\Nopost(1+\tau^t)\right)\right\rangle\!,
\end{align}
where the functions $\partial_{\{1,2\}}\mathsf{F}^{\{\textnormal{R},\textnormal{I}\}}(x+iy,\tau)$ are defined as
\begin{align*}
\partial _1\mathsf{F}^\textnormal{R} &\triangleq \frac{\partial \realpart{\mathsf{F}(x + iy,\tau)}}{\partial x},\quad \partial _2\mathsf{F}^\textnormal{R} \triangleq \frac{\partial \realpart{\mathsf{F}(x + iy,\tau)}}{\partial y},\\
\partial _1\mathsf{F}^\textnormal{I} &\triangleq \frac{\partial \imagpart{\mathsf{F}(x + iy,\tau)}}{\partial x},\quad \partial _2\mathsf{F}^\textnormal{I} \triangleq \frac{\partial \imagpart{\mathsf{F}(x + iy,\tau)}}{\partial y},
\end{align*}
and $\partial_{\{1,2\}}\mathsf{F}^{\{\text{R},\text{I}\}}$, $\mathsf{F}$, as well as $\mathsf{G}$ operate element-wise on vectors.
\end{alg}

We note that $\shat^{t+1}$ in \fref{alg:cB-AMP} corresponds to the (nonlinear) minimum mean-squared error (MMSE) estimate defined in \fref{eq:F_compute}.
For a real-valued system with $\No=\Nopost$, cB-AMP reduces to the real-valued Bayesian AMP (B-AMP) proposed in \cite{Maleki2010phd};  \fref{app:B_AMP_same_cB_AMP} establishes this fact.

\begin{lem}
% [Equivalence of real-valued cB-AMP and B-AMP]
\label{lem:B_AMP_same_cB_AMP}  Let \mbox{$\Nopost=\No$} and assume~$\bH$ satisfies (A1). If $\bH$, $\bms$, and $\bmn$ are real-valued, then cB-AMP reduces to B-AMP in \cite{Maleki2010phd}.
\end{lem}

\subsection{cSE: Complex State Evolution (with Mismatch)}\label{sec:stateevolution}
Two unique features of AMP-based algorithms are 
(i) the output decouples the system into parallel independent channels with additive Gaussian noise (see \fref{fig:introfigure2} for an illustration), and 
(ii) the noise variance of the decoupled AWGN channel 
% \cs{we have not really defined the MSE right? should we rather continue with the story and say that the decoupled noise variance can be predicted?} mean-squared error (MSE) 
can be predicted analytically via fixed-point equations in the large-system limit, which is known as \emph{state evolution} (SE) \cite{andreaGMCS}.
% 
% \cs{not only the MSE; actually we can track far more and also we can establish the decoupling. this statement here must be way more general and powerful}%
% other existing iterative algorithms used for sparse signal recovery \cite{BT09,Saad03,VO96,OBG05,YOGD08},
 %
The SE framework has been investigated in detail in \cite{Maleki2010phd} for B-AMP and in \cite{malekiCAMP} for CAMP, which is a special case of cB-AMP proposed here\footnote{The SE framework presented \cite{malekiCAMP} focused on sparse signal recovery; we present SE framework for general prior distributions.
}. 
Before we delve into the complex SE (cSE) framework for analysis on the noise variance of the decoupled AWGN channels, we first define the mean-squared error (MSE) of cB-AMP's MMSE output.

% \cs{I think now here we can say how the decoupled noise variance and the MSE relates an then define the MSE:}

% Before we show the complex SE (cSE) framework, we need to define the MSE, the output noise variance, and the postulated output variance of cB-AMP.

\begin{defi}
% [Mean-squared error (MSE) of cB-AMP] 
Suppose that $\bmy = \bH \bms_0 + \bmn$, where the signal $\bms_0$ is distributed according to $\bms_0\sim p(\bms_0)$, $\bmn\sim\setC\setN(\mathbf{0}_{\MR\times1},\No\bI_\MR)$, and the postulated noise variance is $\Nopost$. Let $\shat^{t+1}$ be the MMSE output of cB-AMP after $t$ iterations. We define the  MSE of the MMSE output of cB-AMP after $t$ iterations as follows:
\begin{align}\notag
\text{MSE}_{t}& =\lim_{\MT\rightarrow\infty} \frac{1}{\MT}\vecnorm{\shat^{t+1} - \bms_0}_2^2 
\\&= 
% \lim_{\MT\rightarrow\infty} \frac{1}{\MT}\vecnorm{\mathsf{F}(\shat^t+\bH^\Herm\bmr^t,\Nopost(1+\tau^t))- \bms_0}^2.\\
\lim_{\MT\rightarrow\infty} \frac{1}{\MT}
\sum_{\ell=1}^\MT\abs{\mathsf{F}\!\left(\shate_\ell^t+(\bmh_\ell^\col)^\Herm\resid^t,\Nopost(1+\tau^t)\right)\!- s_{0\ell}}^2\!\label{eq:MSE_eq}.
\end{align}
\end{defi}

We now define effective noise variance $\sigma^2_t$, which represents the noise variance of the decoupled AWGN channel after $t$ iterations in the large-system limit (see Figs.~\ref{fig:introfigurec} and~\ref{fig:introfigured} for illustrations).

% 
% As shown in Figs.~\ref{fig:introfigurec} and~\ref{fig:introfigured}, cB-AMP after $t$ iterations decouples the MIMO system into parallel channels with additive Gaussian noise with effective noise variance $\sigma^2_t$ in the large-system limit.
%
% We stress that the effective noise variance $\sigma^2_{t}$ for each AWGN channel defined below corresponds to the actual noise variance seen at the output of cB-AMP after $t$ iterations.

\begin{defi}
% [Effective noise variance of cB-AMP]
\label{def:effective_var} The \emph{effective noise variance} for the MMSE estimate of cB-AMP after $t$ iterations is given by 
% \cs{wouldn't it be smart to insert the MSEt here in the 2nd equation?}
\begin{align}
% \notag
\sigma_{t+1}^2 &= 
\lim_{\MR\rightarrow\infty}\frac{1}{\MR}\vecnorm{\bmr^{t+1}}^2_2
% \\&
= \No + \beta
\text{MSE}_t
% \lim_{\MT\rightarrow\infty}\frac{1}{\MT}\vecnorm{\shat^{t+1} - \bms_0}^2_2\!
.\label{eq:output_noise_var_eq}
\end{align}
\end{defi}

We note that a proof of \fref{eq:output_noise_var_eq} was given in \cite[Lem. 4.1]{BM2012}. While $\sigma_{t+1}^2$ corresponds to the effective noise variance (shown in \fref{fig:introfigured}), the postulated output variance $\gamma^2_{t+1}$ defined below corresponds to the \emph{predicted} value of $\sigma_{t+1}^2$ at iteration $t$ of cB-AMP. 
If there is a mismatch in the noise variance $\Nopost\neq\No$, then the postulated output variance $\gamma^2_t$ differs from the actual noise variance $\sigma^2_t$, i.e., $\gamma^2_t\neq\sigma_t^2$.
\begin{defi}
% [Postulated output variance of cB-AMP]
\label{def:postulated_var} The \emph{postulated output variance} of cB-AMP after $t$ iterations is given by
\begin{align}\notag
\gamma_{t+1}^2 &= 
\lim_{\MT\rightarrow\infty}\Nopost(1+\tau^{t+1})\\
&= \Nopost + \beta\!\!\lim_{\MT\rightarrow\infty}\frac{1}{\MT}\sum_{\ell=1}^\MT \mathsf{G}\!\left(\shate_\ell^t+(\bmh_\ell^\col)^\Herm\resid^t,\Nopost(1+\tau^t)\right)\!\label{eq:input_var_eq}.
\end{align}
\end{defi}

% We will show in \fref{cor:cSE_IO} that in the case of no noise-variance mismatch ($\No=\Nopost$) we have $\sigma^2_t=\gamma^2_t$ for all $t$. 
% 
We can now formulate the complex SE (cSE) framework with noise variance mismatch for cB-AMP. 
The complex SE framework was proven rigorously in \cite{BM2011}; for completeness, we resort to a heuristic derivation of the proof in \fref{app:statederiv}.

\begin{thm}
% [cSE with noise-variance mismatch]
\label{thm:CSE} 
Suppose the entries of $\bms_0$ are i.i.d.\ $p(\bms_0)\sim\prod_{\ell=1}^\MT p(s_{0\ell})$ and the entries of the MIMO system matrix $\bH$ satisfy (A2). Let $\bmn\sim\setC\setN(\bm{0}_{\MR\times 1},\No\bI_\MR)$ and $\mathsf{F} : \complexset \rightarrow \complexset$ be a pseudo-Lipschitz function as defined in \cite[Sec. 1.1, Eq. 1.5]{BM2011}. Assume the large-system limit and that the postulated noise variance is $\Nopost$. Then, the effective noise variance $\sigma^2_{t+1}$ in \fref{eq:output_noise_var_eq} and postulated output variance $\gamma_{t+1}^2$ in \fref{eq:input_var_eq} of cB-AMP in iteration $t$ are given by the following coupled recursion:
\begin{align}
\sigma_{t+1}^2 &  
%= \No + \lim_{\MT\rightarrow\infty}\frac{1}{\MT}\!\left\Vert\bms^{t+1} - \bms_0\right\Vert^2
 = \No +\beta\Psi(\sigma_t^2,\gamma_t^2), \label{eq:SErecursion} \\
\gamma_{t+1}^2 & = \Nopost + \beta \Phi(\sigma_t^2,\gamma_t^2).\label{eq:SErecursion2}
\end{align}
The MSE function $\Psi$ and variance function $\Phi$ are defined by
\begin{align}
\label{eq:Psi}
% \lim_{\MT\rightarrow\infty}\frac{1}{\MT}\vecnorm{\shat^{t+1} - \bms_0}^2 &= 
\Psi(\sigma_t^2,\gamma_t^2) &= \Exop_{\srv,Z}\!\left[\abs{ \mathsf{F}\!\left(\srv + \sigma_t Z,\gamma_t^2\right) - \srv}^2 \right]\!,
\\
% \quad \text{and} \quad 
% \lim_{\MT\rightarrow\infty}\Nopost\tau^{t+1} &= 
\label{eq:Phi}
\Phi(\sigma_t^2,\gamma_t^2) &= \Exop_{\srv,Z}\!\left[\mathsf{G}\!\left(\srv + \sigma_t Z,\gamma_t^2\right)\right]\!,
\end{align}
% \begin{align*}
% \Psi(\sigma^2,\gamma^2) &= \Exop_{\srv,Z}\!\left[\left( \mathsf{F}\!\left(\srv + \sigma Z,\gamma^2\right) - \srv\right)^2 \right]  \\
% \Phi(\sigma^2,\gamma^2) &= \Exop_{\srv,Z}\!\left[\mathsf{G}\!\left(\srv + \sigma Z,\gamma^2\right)\right]
% \end{align*}
respectively, with $\srv\sim p(\srv)$, $Z \sim\setC\setN(0,1)$. The recursion is initialized at $t=1$ with
\begin{align*}
\sigma_1^2 = \No + \beta \Varop_S[S] \quad \text{and} \quad \gamma_1^2 = \Nopost + \beta\Varop_S[S].
\end{align*}
\end{thm}

We note that the MSE function $\Psi(\sigma^2,\sigma^2)$ is identical to the ``$\mathsf{mmse}\!\left(\mathsf{snr}\right)$'' function  in \cite{DSV2005,GV2003,GV2005,WV2010} with the relation $\mathsf{snr}=1/\sigma^2$ used to derive the relationship between the mutual information and the MSE function.
We also note that the MSE of cB-AMP at iteration $t$ as defined in \fref{eq:MSE_eq} is equivalent to $\Psi(\sigma_t^2,\gamma_t^2)$ in the large-system limit. \fref{thm:CSE} implies that the effective noise variance of cB-AMP $\sigma^2$ can be predicted exactly by the variance of a \emph{single} random variable mixed with additive Gaussian noise in the large-system limit.
If $\No=\Nopost$, then we obtain the following result.
 % (a short proof is given in \fref{app:csE_IO}).

\begin{cor}
% [cSE without noise-variance mismatch]
\label{cor:cSE_IO} Let $\Nopost=\No$ in \fref{thm:CSE}. Then \fref{eq:SErecursion} is identical to \fref{eq:SErecursion2}, and the cSE reduces to the following recursion:
\begin{align}\label{eq:SE_IO}
%\sigma_{t+1}^2 = \No + \beta\Exop_{X,Z}\!\left[\left( \mathsf{F}\left(X + \sigma_t Z,\sigma_t^2\right) - X\right)^2 \right].
\sigma_{t+1}^2 = \No +\beta\Psi(\sigma_t^2,\sigma_t^2).
\end{align}
\end{cor}

The proof of \fref{cor:cSE_IO} follows from the fact that the MSE equals to the conditional variance, i.e., $\Phi(\sigma_t^2,\sigma_t^2)=\Psi(\sigma_t^2,\sigma_t^2)$. 
We note that \fref{cor:cSE_IO} corresponds to the cSE derived originally in \cite{malekiCAMP} in absence of noise-variance mismatch. Furthermore, for real-valued systems, \fref{cor:cSE_IO} coincides with the original SE framework in \cite{donoho2009message,Maleki2010phd}. 
In \fref{sec:ERTs}, we will rely on cSE to analyze the performance and complexity of \LAMA. 
In what follows, we assume that there is no mismatch in the prior distribution---this case was studied in \cite{JMS2016}.

%% file: sec3-damp.tex
\section{LAMA: Large-MIMO Detection using cB-AMP}\label{sec:DAMP}

We now derive the \LAMA algorithm.
 % by using a specific set of signal priors within cB-AMP. 
We specify the missing aspects of the large-MIMO system model and detail the \LAMA algorithm along with the corresponding cSE framework.

\subsection{Large MIMO and Optimal Data Detection}

We consider a communication system in which the entries $s_\ell$, $\ell=1,\ldots,\MT$, of the transmit data vector $\bms$ are taken from a finite constellation set $\setO=\{a_j : j=1,\ldots,|\setO|\}$ with points~$a_j$ chosen, from e.g., a pulse amplitude modulation (PAM), phase-shift keying (PSK), or quadrature amplitude modulation (QAM) alphabet.
We assume i.i.d.\ priors $p(\vecs)=\prod_{\ell=1}^{\MT}p(s_\ell)$, with the following distribution for each transmit symbol $s_\ell$:
\begin{align} \label{eq:prior}
p\!\left(s_\ell\right) = \sum_{a\in \setO}p_a\delta\!\left(s_\ell - a\right)\!.
\end{align}
Here, $p_{a}$ is the (known) prior probability of each constellation point $a\in\setO$ and $\delta(\cdot)$ is the Dirac delta distribution; for uniform priors we have $p_{a}=1/|\setO|$.

The vector $\bms_0$ is transmitted through a MIMO channel as in~\fref{eq:systemmodel}. We assume perfect knowledge of the MIMO system matrix $\bH$ at the receiver and the noise vector $\vecn$ to be i.i.d.\ circularly complex Gaussian with variance $\No$ per complex entry.
For these assumptions, the individually optimal (IO) data-detection problem in~\cite{guo2003multiuser,GV2005} is given by
\begin{align}\label{eq:IO}
% (\IO) \quad  
s_\ell^\IO = \argmax_{\tilde s_\ell\in\setO} \sum_{\tilde\vecs_\ell\in\setO^{(\MT-1)}_\ell} \exp\!\left(-\frac{\!\left\Vert \bmy - \bH\tilde\bms \right\Vert^2}{\No}+ \log p(\tilde\bms)\right)\!,
\end{align}
where $\setO^{(\MT-1)}_\ell$ stands for the subset of $\setO^\MT$ that excludes the $\ell$th entry and $\tilde\vecs_\ell\in \setO^{(\MT-1)}_\ell$ is a $\MT-1$ dimensional vector from this subset. 

The detection problem in \fref{eq:IO} is of combinatorial nature and requires prohibitive complexity in systems with large $\MT$ \cite{V1998,JO05,SJSB11}. 
We note that IO data detection achieves the minimum probability of symbol errors (see \cite[Sec. 4.1]{V1998} for a detailed discussion).
While computationally efficient algorithms exist for small-scale MIMO systems (up to about eight transmit streams), such as sphere-decoding (SD) based  methods \cite{HT2003,studer2008soft,studer2010soft}, their average computational complexity still scales exponentially in $\MT$ \cite{JO05,SJSB11}.\footnote{In the case of BPSK transmission, soft-input soft-output MAP detectors, such as the one in \cite{studer2010soft}, can exactly solve the IO problem at low average computational complexity for a small number of transmit streams $\MT$.
For higher-order modulation schemes, no known method exists to solve \fref{eq:IO} at low complexity.
} Consequently,  such methods are not suitable for large MIMO systems. 
In order to enable data detection for such systems, a variety of sub-optimal algorithms have been proposed in the past; see, e.g.,  \cite{indiachemp,ASW2014,CL2014,SKIM2012,CLSK2013,WYWDCS2014} and the references therein.
%

% \cj{added jan 30,2016}
Instead of solving the 
% )\cs{I think we have to be consistent; if we say IO problem, there should be no brackets, because it is its name; if we say the problem (IO), then the brackets should be there as we refer to the problem. Maybe we should never use (IO). I think that would make it more consistent? What do you think?} 
IO problem in \fref{eq:IO} directly, we first compute the marginalized distribution $p(s_\ell\vert\bmy,\bH)$ using cB-AMP as in \fref{alg:cB-AMP}. 
Once we obtain the marginalized distribution $p(s_\ell\vert\bmy,\bH)$, $\ell=1,\ldots,\MT$, the IO data-detection problem is transformed in an entry-wise data detection problem that can be solved at low complexity. 
% by computing the most probable symbol in $\setO$ from $p(s_\ell\vert\bmy,\bH)$.
% 

%We now use cB-AMP together with the prior $p(\bms)$ in~\fref{eq:prior} to develop \LAMA (short for \underline{la}rge \underline{M}IMO \underline{a}pproximate message passing), which enables us to obtain the same error-rate performance of IO data detection at low computational complexity, given certain assumptions on the MIMO system matrix hold (see \fref{sec:SE_MI_ER} for precise optimality conditions).

%
\subsection{Derivation of the \LAMA Algorithm}\label{sec:DAMP-real}

With the prior distribution in \fref{eq:prior}, we can write the posterior distribution \fref{eq:gaussianmeasure} for the transmit symbol $s_\ell$ as
\begin{align}\label{eq:prob_dist}
f\!\left(s_\ell \,|\, \shate_\ell,\tau\right)
%&=
%\frac{1}{Z\!\left(x_\ell,\sigma^2\right)}\exp\!\left( -\frac{1}{\sigma^2}\abs{s_\ell-x_\ell}^2\right)p\!\left(s_{\ell}\right)\\
=
\frac{1}{Z\!\left(\shate_\ell,\tau\right)}\exp\!\left( -\frac{\abs{s_\ell-\shate_\ell}^2}{\tau}\right) \sum_{a\in\setO}p_{a}\delta\!\left(s_\ell - a\right)\!.
\end{align}
Since the normalization constant $Z\!\left(\shate_\ell,\tau \right)$ is chosen so that $\int_\complexset f\!\left(s_\ell\mid \shate_\ell,\tau \right) \text{d}s_\ell=1$, we have 
\begin{align*}
Z\!\left(\shate_{\ell},\tau \right)=\sum_{a\in\setO}p_{a}\exp\!\left( -\frac{1}{\tau}\abs{\shate_\ell-a}^2\right)\!,
\end{align*}
which enables us to write the message mean in \fref{eq:F_compute} as follows
\begin{align}\notag
\mathsf{F}\!\left(\shate_{\ell},\tau \right) & =\int_{\complexset} s_{\ell}f\!\left(s_{\ell}\vert \shate_{\ell},\tau\right)\mathrm{d}s_{\ell} 
\\
&= \frac{\sum_{a\in\setO}a p_{a}\exp\!\left( -\frac{1}{\tau} \abs{\shate_{\ell}-a}^2\right)}
{\sum_{a'\in\setO}p_{a'}\exp\!\left( -\frac{1}{\tau}\abs{\shate_{\ell}-a'}^{2}\right)}
= \sum_{a\in\setO}w_a(\shate_\ell,\tau) a,\label{eq:DAMP_F}
\end{align}
where we use the shorthand notation
\begin{align*}
w_a(\shate_\ell,\tau) = \frac{p_{a}\exp\!\left(-\frac{1}{\tau} \abs{\shate_\ell -a}^2\right)}{\sum_{a'\in\setO}p_{a'}\exp\!\left( -\frac{1}{\tau} \abs{\shate_{\ell}-a'}^2\right)}.
\end{align*}
The message variance $\mathsf{G}$ defined in \fref{eq:G_compute} is given by
\begin{align*}
\mathsf{G}\!\left(\shate_\ell,\tau \right) &= \int_{\complexset} \abs{s_\ell}^2 f\!\left(s_{\ell}\vert \shate_{\ell},\tau \right)\mathrm{d}s_{\ell} - \abs{\mathsf{F}(\shate_\ell,\tau)}^2\!,
\end{align*}
which can be simplified to
\begin{align} \label{eq:DAMPrealG}
\mathsf{G}\!\left(\shate_\ell,\tau \right) 
%&= \frac{\sum \abs{a}^2 g_a(x_\ell)}{\sum g_a(x_\ell)} - \abs{\frac{\sum ag_a(x_\ell)}{\sum g_a(x_\ell)}}^2\notag\\
= \sum_{a\in\setO} w_a(\shate_\ell,\tau)  \!\left\vert a - \mathsf{F}(\shate_\ell,\tau)\right\vert^2\!.
\end{align}
The final step in the derivation of \LAMA involves a simplification of the partial derivatives of~\fref{eq:cb_amp_resid} in \fref{alg:cB-AMP}. 
The result is summarized by \fref{lem:LAMA_alg} with proof given in \fref{app:LAMA_alg}. 

\begin{lem}[Message variance of the \LAMA algorithm]\label{lem:LAMA_alg} Suppose that the assumptions of \fref{alg:cB-AMP} hold, and the mean $\mathsf{F}(\shate_\ell,\tau)$ as well as the variance $\mathsf{G}(\shate_\ell,\tau)$ functions are given by~\fref{eq:DAMP_F} and \fref{eq:DAMPrealG}, respectively. Then, the message variance is given by:
\begin{align*}
\mathsf{G}(\shate_\ell,\tau) = \frac{\tau}{2} \!\left[\partial_1\mathsf{F}^R + \partial_2\mathsf{F}^I \right]\!\left(\shate_\ell,\tau\right)
\end{align*}
and cB-AMP leads to \fref{alg:LAMA}.
% set of equations:
%\begin{align*}
%\bmx^{t+1} &= \mathsf{F}\!\left(\bmx^t + \bH^H\bmz^t,\No(1+\tau^t)\right)\\
%\tau^{t+1} &= \frac{\beta}{\No}\!\left\langle \mathsf{G}\left( \bmx^t + \bH^H \bmz^t ,\No(1+\tau^t)\right)\right\rangle\\
%\bmz^{t+1} &= \bmy - \bH \bmx^{t+1} + \frac{\tau^{t+1}}{1+\tau^t}\bmz^{t}.
%\end{align*}
\end{lem}

With \fref{lem:LAMA_alg} and \fref{alg:cB-AMP}, we arrive at the LAMA algorithm summarized next.

\begin{alg}[LAMA]
% [\LAMA: \underline{La}rge \underline{M}IMO data detection using cB-\underline{A}MP]
\label{alg:LAMA} 
Suppose that $\bH$ satisfies (A1) and \cite[Lem.~5.56]{Maleki2010phd} holds. Then, the \emph{LAMA} algorithm is given by following procedure 
% The LAMA algorithm is given by the following set of recursions:
% 
\begin{align}\nonumber
\bmz^{t} &=\hat\bms^t+\bH^\Herm\bmr^t\\
\label{eq:LAMA_Fstep}
\shat^{t+1} &= \mathsf{F}\!\left(\bmz^{t},\Nopost(1+\tau^t)\right)\\
\label{eq:LAMA_Gstep}
\tau^{t+1} &= \frac{\beta}{\Nopost}\!\left\langle\mathsf{G}\!\left(\bmz^{t},\Nopost(1+\tau^t)\right)\right\rangle\\\label{eq:LAMA_resid}
\resid^{t+1}  &= \bmy-\bH\shat^{t+1}+\frac{\tau^{t+1}}{1+\tau^t}\resid^{t}
\end{align}
for each iteration $t=1,2,\ldots$.
The LAMA algorithm is initialized at iteration $t=1$ with $\shat^t=\Exop_S[S]\bm1_{\MT\times1}$, $S\sim p(S)$, $\resid^t= \bmy-\bH\shat^t$, and $\tau^t=\beta\Varop[S]/\Nopost$.
\end{alg}

The main difference between the cB-AMP in \fref{alg:cB-AMP} and LAMA in \fref{alg:LAMA} is that the update in \fref{eq:cb_amp_resid} for cB-AMP is simplified to \fref{eq:LAMA_resid} for LAMA and we utilize the prior distribution $p(S)$ to initialize the algorithm. 
We note that \LAMA as summarized in \fref{alg:LAMA} makes use of the postulated noise variance~$\Nopost$; this allows us not only to model a mismatch in the noise variance, but also enables us to perform IO detection and matched filter (MF) data detection solely by selecting appropriate values for $\Nopost$; see \fref{sec:lama-iojo}.

% \begin{algorithm}[t]
% \caption{\LAMA: \underline{La}rge \underline{M}IMO data detection using cB-\underline{A}MP}\label{alg:LAMA}
% \begin{algorithmic}[1]
% \STATE {\bf inputs:} $\bH\in\complexset^{\MR\times \MT}$, $\bmy\in\complexset^{\MR}$, $\beta=\MT/\MR$, $\setO$, and $\Nopost$
% \STATE {\bf initialize:} $t=1$, $\shat^t=\Exop_S[S]\bm1_{\MT\times1}$, $\resid^t= \bmy-\bH\shat^t$, and $\tau^t=\beta \Varop[S]/\Nopost$
% \WHILE{not converged}
% \STATE $\bmz^{t}=\hat\bms^t+\bH^\Herm\bmr^t$
% \STATE $\shat^{t+1} = \mathsf{F}\!\left(\bmz^{t},\Nopost(1+\tau^t)\right)$
% \STATE $\tau^{t+1} = \frac{\beta}{\Nopost}\!\left\langle\mathsf{G}\!\left(\bmz^{t},\Nopost(1+\tau^t)\right)\right\rangle$
% \STATE $\resid^{t+1}  = \bmy-\bH\shat^{t+1}+\frac{\tau^{t+1}}{1+\tau^t}\resid^{t}$
% \STATE $t\gets t+1$
% \ENDWHILE
% \STATE {\bf output 1:} Gaussian output $\bmz^{t}$ with postulated variance $\Nopost(1+\tau^t)$ 
% \STATE {\bf output 2:} Non-linear MMSE $\shat^{t+1}$
% \end{algorithmic}
% \end{algorithm}

%%
%%%%%%%%%%%%%%%%%%%%%%%%%%%%%%%%%%%%%%%%%%%%%%%%%%%%%%
\subsection{\LAMA Decouples Large-MIMO Systems}\label{sec:LAMAdecouple}

We now show that LAMA decouples a MIMO system into a set of parallel and independent AWGN channels with identical noise variance  in the large system limit (cf. Figs.~\ref{fig:introfigurec} and~\ref{fig:introfigured}). 
%
%
%From the MIMO system model \fref{eq:systemmodel}, we have the relation for each receive antenna $k$ by  with $\beta=1$
%\begin{align*}
%y_\ell &= \bmh^\row_k \bms_0 + n_k\\
% &= H_{\ell,\ell}s_{0\ell} + \sum_{k\neq \ell}H_{\ell,k}s_{0k} + n_\ell,
%\end{align*}
%with $n_\ell \sim\setC\setN(0,\No)$ for all $k$.
%
%The middle term $\sum_{k\neq \ell}H_{\ell,k}s_{0k}$ is the unwanted interference between the transmit signals. 
 % and \cs{which classic method? I would not even mention the SVD here.} the classic method \cite{RC1998} to remove this is usually done by singular value decomposition (SVD).
%
First, we discuss the outputs of \LAMA: (i) the Gaussian output vector $\vecz^{t}$, (ii) the postulated variance $\Nopost(1+\tau^t)$, and (iii) the non-linear MMSE output vector $\hat\vecs^{t}$. 

\subsubsection*{(i) Gaussian output vector $\vecz^t$}
In each iteration $t$, cB-AMP computes the marginal distribution for $s_\ell$ for $\ell=1,\ldots,\MT$, which corresponds to a Gaussian distribution centered around the original signal $s_{0\ell}$ with variance $\sigma_{t+1}^2$.
These properties on $\bmz^t$ follow from 
% \fref{eq:stateevoeta} and \fref{eq:stateevosigma} of 
\fref{thm:CSE}, which shows that $\bmz^t=\shat^t+\bH^\Herm\bmr^t$ is distributed according to $\setC\setN(\bms_0,\sigma_t^2\bI_\MT)$ in the large-system limit~\cite{BM2011,BM2012}.
Therefore, the input--output relation for each transmit stream $z_\ell^t = \shate_\ell^t+ (\bmh_\ell^\col)^\Herm\bmr_\ell^t$ is equivalent to the following single-input single-output AWGN channel:
\begin{align} \label{eq:decoupledAWGNsystem}
z_\ell^t = s_{0\ell} + n_\ell^t.
\end{align}
Here, $s_{0\ell}$ is the transmitted signal and $n_\ell^t\sim\setC\setN(0,\sigma^2_t)$ is AWGN with effective noise variance $\sigma^2_t$ per complex entry. 
Since $p(z_\ell^t\,\vert\,s_{0\ell})\sim\setC\setN(s_{0\ell},\sigma_t^2)$, the posterior distribution of \fref{eq:decoupledAWGNsystem} as defined in \fref{eq:prob_dist} is given by
$f(s_{0\ell}\,\vert\,z_\ell^t,\sigma_t^2)$.
An immediate consequence of these properties is the fact that LAMA decouples the MIMO system (cf.~\fref{fig:introfigured}). 
We note that the decoupling behavior of \LAMA was observed for posterior mean estimators (PMEs) in randomly spread CDMA systems~\cite{GV2005,V1998} for which no practical data detection algorithm was given.

% The equivalent noise variance $\sigma^2_t$ can be computed from cSE in \fref{eq:SErecursion}. In particular, for $t\rightarrow\infty$, \fref{thm:CSE} converges to the following coupled fixed-point equations:
% \begin{align}
% \sigma^2 & = \No + \beta \Psi\left(\sigma^2,\gamma^2\right), \label{eq:sigma_gamma1} \\
% \gamma^2 & = \Nopost+\beta\Phi\left(\sigma^2,\gamma^2\right) \label{eq:sigma_gamma},
% \end{align}
% % Here, $\Psi(\sigma^2,\gamma^2)$ and $\Phi(\sigma^2,\gamma^2)$ are defined in \fref{thm:CSE} with the \cs{message? always use the same wording} mean and variance functions $\mathsf{F}$ and $\mathsf{G}$ are defined by the constellation set $\setO$ in \fref{eq:DAMP_F} and \fref{eq:DAMPrealG} respectively.
% whose fixed point(s) $\sigma^2$ correspond(s) to the equivalent noise variance in the large system limit. Since, in general, multiple fixed points to \fref{eq:sigma_gamma} may exist, we discuss the corresponding optimality conditions in \fref{sec:LAMA_optimal}. 
% %
% In the case for \LAMA, when $\Nopost=\No$, we note that $\sigma^2_t$ is tracked directly within the \LAMA algorithm. This follows from \fref{cor:cSE_IO} because $\No(1+\tau^t)\rightarrow\sigma_t^2$ in the large system limit.

% This is possible because \LAMA is leveraging the ``decoupling'' behavior that naturally occurs by the use of posterior mean detectors \cite{GV2005,V1998}. Also, for large-dimensional wireless MIMO systems, the resulting behavior can be seen as ``wires in the air'' \cite{MB2007}.
% similar ``crystallization'' behavior was also observed in a two-hop relay network in \cite{MB2007}. 

\subsubsection*{(ii) Postulated output variance $\Nopost(1+\tau^t)$}
In the large-system limit, there exist two noise variances: effective noise variance $\sigma_t^2$ from \fref{def:effective_var} and the postulated output variance~$\gamma_t^2$ from~\fref{def:postulated_var}. We note that $\Nopost(1+\tau^t)\to\gamma_t^2$ in the large-system limit.
We clarify the difference between the two quantities below.

The effective noise variance $\sigma_t^2$ is the \emph{true} noise variance in \fref{eq:decoupledAWGNsystem}, whereas the postulated output variance $\gamma_t^2$ is the estimate for $\sigma_t^2$ each iteration $t$.
The postulated output variance $\gamma_t^2$ is used as an input to the posterior mean function $\mathsf{F}$ (see \fref{eq:LAMA_Fstep} and Figs.~\ref{fig:retro1}~and~\ref{fig:retro2}) to the Gaussian vector $\bmz^t$ to obtain the MMSE estimate $\hat\vecs^{t+1}$.
Therefore, when the exact value of $\sigma_t^2$ is unknown at the receiver, a possible performance mismatch can result in using an incorrect value  $\gamma_t^2$ for obtaining the MMSE estimate. 
The cSE framework shown in \fref{thm:CSE} enables us to analyze the performance loss due to such a (possible) mismatch in the noise variance $\Nopost$ \emph{exactly}.

If there is no mismatch in the postulated noise variance, we have $\sigma_t^2=\gamma_t^2$ by \fref{cor:cSE_IO} and hence, the correct noise variance statistic is used for the MMSE estimate in \fref{eq:LAMA_Fstep} every iteration.
However, if $\Nopost\neq\No$, then $\sigma_t^2\neq\gamma_t^2$, and therefore, \LAMA applies the MMSE estimate on the Gaussian vector $\bmz^t$ according to an incorrect statistic, which may cause \LAMA to converge to an incorrect solution.
To illustrate how \LAMA may converge to an incorrect solution, consider the case where $\No=0$, and $\Nopost\rightarrow\infty$. In this case, $\bmz^t$ corresponds to the MF detector; see \fref{sec:lama-iojo} for more details.

\subsubsection*{(iii) Non-linear MMSE output vector $\hat\vecs^{t+1}$}
The non-linear MMSE output vector $\hat\vecs^{t+1}$ is given by $\hat\vecs^{t+1}=\mathsf{F}(\vecz^{t},\Nopost(1+\tau^{t}))$ in \fref{eq:DAMP_F}, which can be seen as a conditional mean of the Gaussian output vector $\vecz^{t}$ for the postulated output variance $\Nopost(1+\tau^{t})$. 
The non-linear MMSE output vector $\hat\vecs^{t+1}$ is identical to the PME \cite{GV2005}, 
where each $\ell$th output of PME is obtained by the expectation with respect to the conditional distribution $f\!\left(s_\ell\,\vert\,z^{t}_\ell,\Nopost(1+\tau^{t})\right)$ in~\fref{eq:prob_dist}.\footnote{The conditional distribution $f\!\left(s_\ell\,\vert\,z^{t}_\ell,\Nopost(1+\tau^{t})\right)$ is called ``retro-channel'' in~\cite{GV2005}.}
%
% Note that in the large system limit, the output of PME converges to the 
% Since each stream $\ell$ can be expressed the system by \fref{eq:decoupledAWGNsystem} in the large limit with $\Nopost(1+\tau^{t})$ converging to $\gamma_t^2$, we can express the non-linear MMSE output $\hat{s}^{t+1}_\ell$ as the PME under the distribution $f\!\left(s_\ell\,\vert\,z^{t}_\ell,\gamma_t^2\right)$.
%
The equivalence of \LAMA and the equivalent AWGN relation for the non-linear MMSE estimate is shown in \fref{fig:retro}. In \fref{fig:retro1}, the quantity $\hat\bms^{t+1}$ is the non-linear MMSE estimate with the postulated noise variance $\Nopost(1+\tau^t)$. In the large system limit, the input-output relation for each stream $\ell$ is an AWGN channel in \fref{fig:retro2} with equivalent variance $\sigma_t^2$ and the postulated variance $\gamma_t^2$.
%%%%
%%%%%%%%%

\begin{figure}[tp]
\centering
\hspace*{0.4cm}
\subfigure[]{\includegraphics[height=0.1\textheight]{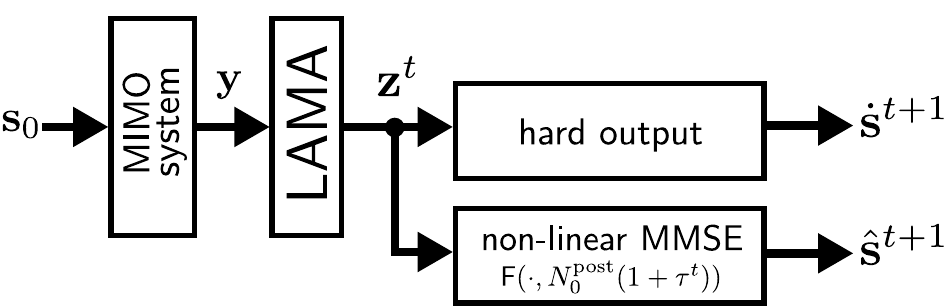}\label{fig:retro1}}
\hspace*{0.4cm}
\subfigure[]{\includegraphics[height=0.1\textheight]{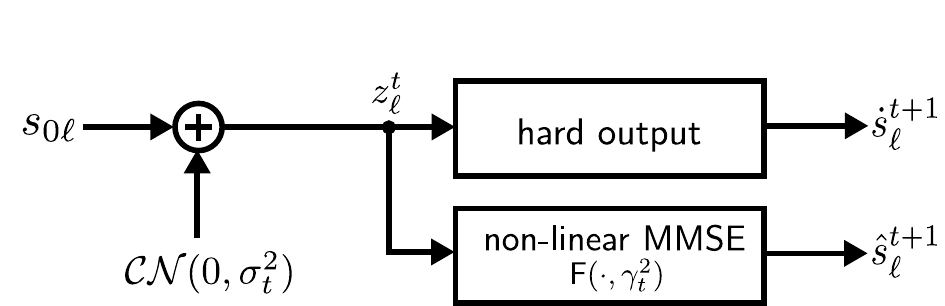}\label{fig:retro2}}
\caption{The system with \LAMA and its outputs (a) and the statistically-equivalent decoupled AWGN system as seen at the output of \LAMA (b). 
\LAMA generates a Gaussian output $\bmz^t$ and a non-linear MMSE estimator output $\shat^{t+1}$. Hard-output estimates~$\dot{s}^{t+1}$ are generated via \fref{eq:hard_decision}. 
% To account for the possible mismatch in the noise variance, \LAMA uses $\Nopost(1+\tau^t)$ to apply the MMSE estimate on the Gaussian vector $\bmz^t$.
%
In the large system limit, \LAMA decouples the MIMO system into independent, parallel AWGN channels with equivalent output noise variance $\sigma_t^2$.  
% In the large system limit, \LAMA applies the MMSE estimate on the Gaussian term $z_\ell^t$ by the postulated output variance $\gamma_t^2$ computed by cSE.
%
}
\label{fig:retro}
\end{figure}

\subsection{\LAMA 
% , JO-\LAMA, 
and MF Data Detection}\label{sec:lama-iojo}
Since LAMA decouples the MIMO system, data detection reduces to element-wise hard decisions for each entry in $\bmz^t$ subject to the postulated output variance $\Nopost(1+\tau^{t})$ as 
% ; this output is denoted by $\dot{\bms}^{t+1}$ in \fref{fig:retro1} with the $\ell$-th output $\dot{s}_\ell^{t+1}$ in \fref{fig:retro2}. 
%
% We note the hard decision under the postulated output variance $\Nopost(1+\tau^{t})$ corresponds to computing the element-wise MAP detection according to
\begin{align}\label{eq:hard_decision}
\dot{s}_\ell^t=\argmax_{s_\ell\in\setO} f(s_\ell\,\vert\,z_\ell^t,\Nopost(1+\tau^{t})).
 \end{align}

By setting the postulated noise variance $\Nopost$, \LAMA can perform IO and MF data detection.
In particular, (i) for $\Nopost=\No$, \LAMA corresponds to the IO detector
% and is referred to as \LAMA
% , (ii) for $\Nopost\rightarrow 0$, \LAMA corresponds to the JO (or MAP) detector and is referred to as JO-\LAMA, 
and (ii) for $\Nopost\rightarrow\infty$, \LAMA corresponds to the MF detector. These two ``operation modes''  are detailed next.

\subsubsection*{(i)}
Consider $\Nopost=\No$. From \fref{cor:cSE_IO}, we have that the equivalent output noise variance and the postulated noise variance are equal, which implies $\sigma_t^2=\gamma_t^2$ in the large-system limit.
Since there is no noise variance mismatch, the output \fref{eq:hard_decision} achieves the same error-rate performance as the IO data detector which in \fref{eq:IO} given certain conditions are met;
% \cj{check ref}
see \fref{sec:LAMA_optimal} for precise optimality conditions.
% \cs{for what? the optimality conditions?}. \cs{didn't we say that way earlier when we introduced the IO problem if not we should and then we can remove this redundant statement here} We note that IO data detection achieves the minimum probability of symbol errors (see \cite[Sec. 4.1]{V1998} for a detailed discussion).
% 
% As in \cite{V1998,TanakaCDMA,GV2005,RFG2012}, letting $\Nopost\rightarrow0$ yields jointly optimal detection.

% 
% 
% 
% 
% 
% 
% 
% 
% 
% 
% 

\subsubsection*{(ii)}
By letting $\Nopost\rightarrow\infty$, it was shown in \cite[Eq. (12)]{GV2005} that the output of the non-linear MMSE estimator \fref{eq:DAMP_F} corresponds to the MF output for real-valued signals with $\Exop[S]=0$. We now provide conditions for which \LAMA with $\Nopost\rightarrow\infty$ performs MF data detection for arbitrary system ratios $\beta$.
% provided that the constellation set $\setO$ is separable.
%
The proof of the following Lemma is given in \fref{app:LAMA_MF}.

\begin{lem}
% [MF data detection with LAMA]
\label{lem:LAMAtoMF} 
Fix the constellation set $\setO$, and let $S\sim p(S)$. If $\Exop_S[S]=0$, then as $\Nopost\rightarrow\infty$, the Gaussian output at every iteration $t=1,2,\ldots,$ from \LAMA corresponds to the MF output:
\begin{align*}
\lim_{\Nopost\rightarrow\infty}\bmz^t = \bH^\Herm\bmy.
\end{align*}
If $\Exop_S[\,\realpart{S}\imagpart{S}\,] = \Exop_S[S\abs{S}^2]=0$, 
 % and $\Exop_{\realpart{S}}[\realpart{S}^2]=\Exop_{\imagpart{S}}[\imagpart{S}^2]=E_s/2$. 
then, as $\Nopost\rightarrow\infty$, the  scaled version of the non-linear MMSE estimate also corresponds to the MF output:
\begin{align*}
\lim_{\Nopost\rightarrow\infty} \frac{\Nopost}{E_s}\shat^t = \bH^\Herm \bmy.
\end{align*}
\end{lem}

% \cs{any words to summarize this? or is that it?}

% In \fref{sec:LAMA_optimal}, we will provide the exact conditions for which the hard estimates of LAMA correspond to the outputs of the IO and JO problems defined in \fref{eq:IO} and \fref{eq:JO}.

\subsection{\LAMA in the Massive MU-MIMO Limit}\label{sec:massiveMIMO}

We now study the properties of \LAMA in the massive MU-MIMO limit, where we fix the number of streams (or layers) $\MT$ and let the number of BS antennas $\MR\to\infty$.
% where $\beta=\MT/\MR\rightarrow0$. 
%
As shown in~\cite{rusek2013scaling,LETM2014}, MF data detection is optimal in such scenarios. The following Lemma reveals that LAMA corresponds to the MF detector in the massive MU-MIMO limit; a proof is given in \fref{app:LAMAbeta0}.
%
%Results from \cj{refs} suggest that in massive MU-MIMO, the matched filter detection is optimal. In this subsection, we show that this can also be seen by LAMA. Similar to the results in \fref{lem:LAMAtoMF}, we show that for $\beta\rightarrow 0$, the output of LAMA behaves similar to a matched filter detector. This is shown in \fref{lem:LAMAbeta0} with proof in \fref{app:LAMAbeta0}.

\begin{lem}
% [LAMA in the Massive MU-MIMO Limit]
\label{lem:LAMAbeta0} Assume that $\setO$ is fixed and let $\Nopost\geq 0$. Then, for $\beta\to0$, the Gaussian output $\bmz^t$ of \LAMA corresponds to the MF data detector, i.e., $\bmz^t = \bH^\Herm\bmy$, for all $t\geq1$.
Furthermore, the effective noise variance is $\sigma_t^2 = \No$ for all $t\geq1$.
\end{lem}

This result is in accordance with \cite{rusek2013scaling,LETM2014} and implies that a simple one-shot algorithm (performing a single iteration) is sufficient to perform IO data detection in the massive MU-MIMO limit. Furthermore, \LAMA decouples the MIMO system into parallel and independent AWGN channels with variance $\sigma_t^2=\No$ (see \fref{fig:introfigurec} and \fref{fig:introfigured}) in every iteration. 
We emphasize that \LAMA can be used in more-realistic massive MU-MIMO systems, i.e., where 
% $0<\beta\ll 1$ and 
the number of BS antennas is finite. As we will show in \fref{sec:performance_tradeoff}, \LAMA quickly converges and provides near-optimal performance for realistic massive MU-MIMO antenna configurations. 

%\fref{lem:LAMAbeta0} reconciles the results obtained from massive MIMO and AMP, and show that \LAMA can be adapted to such settings. Moreover, when $\beta\rightarrow 0$, LAMA is able to fully decouple the system into parallel AWGN system with no performance loss because the equivalent noise variance is given as $\sigma_t^2=\No$ for all $t>1$.

% In \fref{sec:SE_MI_ER}, we will show the optimality of \LAMA and JO-\LAMA, and detail the limits ofthe performance and complexity trade-off of IO and JO-\LAMA by cSE, respectively.

%% file: sec4-BER.tex
% !TEX root = 18TIT-lama.tex

\section{Optimality of LAMA}\label{sec:SE_MI_ER}

We now provide exact conditions for which \LAMA achieves the performance of the IO data detector. We furthermore study the noiseless case in which \LAMA is able to perform error-free data recovery. 

% \cs{some section/subsection titles are still not capitalized; I did this one}
\subsection{Existing Results of IO and Multiuser Detection in Large MIMO Systems}\label{sec:GV2005_discussion}
An spectral efficiency analysis of IO data detection in large systems with BPSK was presented by Tanaka in \cite{TanakaCDMA}. 
% analyzed the performance of IO detection in a CDMA system with BPSK constellation.
% 
% While the analysis relied on the non-rigorous replica method, Tanaka's results were proved rigorously in \cite{MT06} for $\beta\leq1.49$.
%
These results were generalized to arbitrary constellation sets in~\cite{GV2005}. 
Under the assumption that replica method is correct, Guo and Verd\'u showed in  \cite{GV2005} that by using 
% a special family of multi-user detectors, % referred as posterior mean estimators (
PMEs,
 the multi-user channel in the large-system limit decouples into an AWGN channel for each transmit stream, where the noise is amplified by a factor $\eta^{-1}$ due to the interference of other streams.
The factor $\eta\in(0,1)$, known as the multi-user efficiency, can be computed exactly by solving the following coupled equations for $\eta$ and $\xi$:
\begin{align}\label{eq:GV_psi}
\No/\eta &= \No + \beta\Exop_{\srv,Z}\!\left[\abs{\mathsf{F}\!\left(\srv + \sqrt{\No/\eta}Z,\No/\xi \right) - \srv}^2 \right]\!,\\\label{eq:GV_phi}
\No /\xi &= \Nopost + \beta\Exop_{\srv,Z}\!\left[\mathsf{G}\!\left(\srv + \sqrt{\No/\eta}Z,\No/\xi\right)\right]\!.
\end{align}
Here, the functions $\mathsf{F}$ and $\mathsf{G}$ depend on constellation set $\setO$ as in \fref{eq:F_compute} and \fref{eq:G_compute}, respectively.

We note that the performance of IO data detection corresponds to the case with $\No=\Nopost$. In this case, the right-hand side of \fref{eq:GV_phi} is equal to \fref{eq:GV_phi}, and therefore, $\eta=\xi$.
Thus, the multi-user efficiency $\eta$ is given by a  
% to compute the noise variance of the AWGN channel from the decoupling principle 
single fixed-point equation
% 
% The performance of IO detection, which is when there is no mismatch in noise variance, i.e., $\No=\Nopost$, is computed by noting that \fref{eq:GV_psi} and \fref{eq:GV_phi} are equal under $\eta=\xi$. 
% 
% Under IO detection, the decoupling principle also applies---the noise variance from the AWGN channel is given by $\No/\eta$, which satisfies the fixed-point equation
% 
\begin{align}\label{eq:GV_IO_fp}
\No /\eta &= \No + \beta\Exop_{\srv,Z}\!\left[\abs{\mathsf{F}\!\left(\srv + \sqrt{\No/\eta}Z,\No /\eta \right) - \srv}^2 \right]\!.
\end{align}
If there exist multiple fixed points to \fref{eq:GV_psi} and \fref{eq:GV_phi}, we pick the tuple $(\eta,\xi)$ that minimizes the so-called ``free energy'' (as done in  \cite[Sec. 2-D]{GV2005}) given by:
\begin{align}\notag
\mathcal{F} &=\int_\complexset p(z,\No/\eta) \log_2 p(z,\No/\xi) \dd z
\\\notag
 &\quad+ 
\frac{1}{\beta}\!\left( (\xi -1) \log_2 e - \log_2\xi
\right) + \log_2 \frac{\xi}{\pi} - \frac{\xi}{\eta}\log_2 e\\
&\quad+ \frac{\Nopost}{\beta\No}\frac{\xi}{\eta}(\eta-\xi)\log_2 e + \frac{1}{\beta}\log_2(2\pi) + \frac{\xi}{\eta\beta}\log_2e,
\label{eq:GV_free_energy}
\end{align}
where the term  $p(z,x)$ in \fref{eq:GV_free_energy} is obtained by marginalizing the joint distribution \mbox{$p(z,x\vert s)\sim\setC\setN(s,x)$} with respect to the prior distribution \mbox{$s\sim p(s)$}, i.e., \mbox{$p(z,x)=\int_\complexset p(z,x\vert s) p(s)\dd s$}.
% , $p(z,x\vert s)\sim\setC\setN(s,x)$, and $p(s)$ is the prior distribution.
% 

% \cj{added Oct 3,2016 - this part is regarding sparse matrix-> send infty -> make dense matrix}

We note that the aforementioned results rely on the replica method, which build on the replica assumptions \cite{TanakaCDMA,GV2005}.
Montanari and Tse in \cite{MT06} proposed an alternative approach to prove Tanaka's results in \cite{TanakaCDMA} up to certain system ratios $\beta$ for BPSK systems. 
Instead of directly analyzing a dense MIMO system matrix, Montanari and Tse first introduce a ``sparse signature'' scheme, in which only a sparse subset of the channel matrix is active.
For this system, the performance of belief propagation (BP) can be analyzed via density evolution. Once the density evolution expressions were established in the large-system limit, one can ``densify'' the MIMO system matrix to ensure that the each entry is distributed (A1);
we shall refer to this setup as \emph{large-sparse} limit~\cite{WG2006}.
 % for the subsequent discussions.
% 
By doing so, one recovers Tanaka's results derived under the replica method without relying on the replica assumptions.
% 
% We note that the approach in \cite{MT06} is fundamentally different than replica-based approaches \cite{TanakaCDMA,GV2005}, as Montanari and Tse starts with a \emph{sparse} system, takes the large-system limit, and then ``densifies'' the system for analyses. 
% 
The analysis of BPSK systems using this sparse signature scheme has been generalized to arbitrary prior input distributions in \cite{GW2006,GW2007,WG2006}. 
Not surprisingly, these results 
% analyses \cite{GW2006,GW2007,WG2006} 
agree with the replica results \cite{GV2005} 
when the fixed-point $\eta$ to \fref{eq:GV_IO_fp} is unique.
In addition, in \cite{WG2006}, Wang and Guo showed that BP is equivalent to element-wise MAP estimation, and the detection performance of the BP is identical to that given by a AWGN system with noise amplified by $\eta^{-1}$ obtained in \fref{eq:GV_IO_fp}.
% 
%We conclude this section by noting that these results are analytical in nature; BP algorithm was utilized for analysis only and the complexity of BP was not discussed.

\subsection{Fixed Points of \LAMA}

Before we provide exact optimality conditions for \LAMA, we highlight that under \fref{thm:CSE}, as $t\to\infty$ the cSE converges to the following fixed-point equations: for $\No=\Nopost$, we have
\begin{align}\label{eq:fixed_pt}
% \sigma_\IO^2 = \No + \beta \Exop_{X,Z} \left[ (\mathsf{F}(X+\sigma_\IO Z,\sigma_\IO^2) - X)^2\right].
\sigma_\IO^2 = \No + \beta \Psi(\sigma_\IO^2,\sigma_\IO^2),
\end{align}
whereas for $\Nopost\neq\No$, we have
\begin{align}
\label{eq:fixed_pt_JO}
\sigma_{\text{m}}^2 = \No + \beta \Psi\left(\sigma_{\text{m}}^2,\gamma_{\text{m}}^2\right) \quad \text{and} \quad
\gamma_{\text{m}}^2 = \Nopost+&\beta\Phi\left(\sigma_{\text{m}}^2,\gamma_{\text{m}}^2\right).
\end{align}

As mentioned above, the fixed-point equation for LAMA in \fref{eq:fixed_pt} and \fref{eq:fixed_pt_JO} corresponds to the fixed-point equations for IO data detection in \fref{eq:GV_IO_fp}, and \fref{eq:GV_psi} and \fref{eq:GV_phi}, respectively, with $\sigma^2_{\text{m}}=\No/\eta$ and $\gamma^2_{\text{m}}=\No/\xi$. 
% 
% We note that while the optimal performance of IO was described by Guo and Verd\'u in \cite{GV2005}, no constructive algorithm that delivers this fixed-point solution was proposed. 

% 
In general, the above fixed-point equations may have multiple solutions.
In the case of a unique fixed point, then \LAMA \emph{always} recovers the solution with the minimal effective noise variance $\sigma^2$ regardless of initialization, and thus, achieves the same error-rate performance as IO data detection (see \fref{sec:optimality} for the details).
% 
% We emphasize that if the fixed-point solution to \fref{eq:fixed_pt} is unique, 
% 
In the case of such non-unique fixed points, Guo and Verd\'u choose the solution that minimizes free-energy\footnote{The solution that minimizes the free energy in \fref{eq:GV_free_energy} is equivalent to the thermodynamically dominant solution in statistical physics \cite{TanakaCDMA,GV2005}.} 
given in~\fref{eq:GV_free_energy}, whereas the fixed point obtained by \LAMA depends on the initialization\footnote{Convergence to another fixed-point solution is possible if LAMA is initialized sufficiently close to such a fixed point \cite{ZMWL2015}.} 
of the algorithm and thus, we cannot expect it to converge to the same fixed point that minimizes the free-energy~\fref{eq:GV_free_energy}. 
We note that depending on the initialization of \LAMA presented in \fref{alg:LAMA}, \LAMA converges to the fixed-point solution with the largest effective noise variance $\sigma^2$ in \fref{eq:fixed_pt} and \fref{eq:fixed_pt_JO}, respectively. 
Therefore, if there are multiple fixed points to \fref{eq:fixed_pt}, then \LAMA is, in general, sub-optimal and does not necessarily converge to the fixed-point solution with minimal free-energy.

Before we delve into the optimality analysis of \LAMA, we note that the fixed-point analysis for \LAMA with noise mismatch is more involved as it requires finding fixed points for the coupled fixed-point equations in \fref{eq:fixed_pt_JO}. 
Hence, we focus  on the case $\No=\Nopost$.

\subsection{When Does \LAMA Achieve the Same Performance as IO Data Detector?}
\label{sec:optimality}
We note that the performance of \LAMA (in the large-system limit) is fully described by the SE framework. However, characterizing the performance of the IO data detector is a non-trivial task.
% 
% One viable approach for analyzing the performance of IO data detector is using the replica method \cite{GV2005}. 
Although an analysis via the replica method \cite{GV2005} was recently proved to be correct under mild assumptions \cite{RP2016},  a verification of the assumptions still requires extensive work for each prior distribution.
Therefore, to establish optimality of \LAMA, we first introduce an additional assumption to characterize the performance of the IO data detector, and then show that under this assumption, \LAMA achieves the same data detection performance as the IO data detector. 
% 

% We start by defining the large-sparse-system limit as done in \cite{WG2006}. 
% 
We define a specific example of a large-sparse limit that will be used for our analysis of \LAMA. The general definition of large-sparse limit is provided in \cite{WG2006}.

\begin{defi}
% [Large-sparse limit] 
\label{def:LarseSparselimit}
We define the large-sparse limit as the following procedure:
% 
% \cs{this definition is not clear at all; }
First, start by defining a binary-valued matrix $\bB \in \{0,1\}^{\MR\times\MT}$. 
Pick a constant $\Gamma \leq \MT$ and generate each entry $B_{k,\ell}$ as an i.i.d. Bernoulli random variable with probability $\Gamma/\MT$.
Define a normalization constant $\Gamma_\ell = \sum_{k=1}^\MR B_{k,\ell}$ for each $\ell=1,\ldots,\MT$. 
Then, generate the channel matrix $\bH$ with each entry being i.i.d. $H_{k,\ell}\in \setC\setN(0,1/\Gamma_\ell)$ if $B_{k,\ell}=1$ and 0 otherwise.
Based on this construction of $\bH$ for a fixed $\Gamma$, 
we define the large-sparse limit when we first let $\MR,\MT\to\infty$ with $\MT/\MR=\beta$. Then, we let $\Gamma\to\infty$.
\end{defi}
We note that the large-system limit corresponds to the case when we first set $\Gamma = \MT$ and then let $\MR,\MT \to\infty$ with $\MT/\MR = \beta$. 
However, we will assume that we first fix a constant $\Gamma < \MT$, and then let $\MR,\MT\to\infty$; 
this formulation of the large-sparse limit is needed to prevent the factor graph for the input-output relation in \fref{eq:systemmodel} from having short cycles~\cite{WG2006}. 
We need an additional assumption to establish optimality of \LAMA. 
We assume that exchanging the order of the large-system limit still holds true for cSE:

\begin{description}
% \label{asm:LarseSparselimit}
\item[(A3)]
We assume that cSE for LAMA remains valid in the large-sparse limit.
\end{description}

With \fref{def:LarseSparselimit} and (A3), we will now establish optimality of \LAMA in two parts.
First, we show that in the large-sparse limit, BP achieves the same performance as the IO data detector and the input-output relation is asymptotically decoupled into AWGN channels with equal decoupled noise variance. 
Second, we show that \LAMA achieves the same noise variance as that given by BP using state evolution. 
Since the input-output relation is decoupled into AWGN channels and \LAMA achieves the lowest (unique) decoupled noise variance, \LAMA achieves the same detection performance as the IO data detector.
We show the first part by \cite[Thm. 4]{WG2006}:

\begin{thm}\label{thm:wg2006_BP} Assume the large-sparse limit and the system ratio $\beta_\textnormal{BP}$ is chosen such that the fixed-point solution of BP $\eta_\textnormal{BP}$ to \fref{eq:GV_IO_fp} is unique.
Then, BP achieves the same performance as the IO data detector. 
In addition, the posterior distribution of each user after BP converges to that given by an AWGN channel with variance $\No/\eta_\textnormal{BP}$.
\end{thm}

\fref{thm:wg2006_BP} shows that in the large-sparse limit and for unique fixed points, one can use BP to achieve the same performance as IO data detector. 
The proof in \cite[Sec. V]{WG2006} uses a sandwiching argument between genie-aided BP and classical BP to achieve IO performance.
Interestingly, the posterior distribution of each transmit stream after BP converges to that given by an AWGN channel. 
In addition, the noise variance of the equivalent AWGN channel can be characterized by solving a fixed-point equation \fref{eq:GV_IO_fp}; this fixed-point equation coincides exactly to that given by the replica method shown in \cite{GV2005}. 
Now that we have shown that BP achieves IO performance and characterized the decoupling of AWGN, we now establish optimality of \LAMA.
\begin{cor}
Assume the large-system limit and $\beta_\textnormal{LAMA}=\beta_\textnormal{BP}$ from \fref{thm:wg2006_BP}.
% is chosen that the fixed-point solution \fref{eq:fixed_pt} is unique.  
\label{cor:LAMA_opti}
Then,  \LAMA decouples the MIMO system into parallel AWGN channels with variance $\sigma_\textnormal{IO}^2$, which is a unique fixed-point solution to \fref{eq:fixed_pt} with $\sigma_\textnormal{IO}^2=\No/\eta_\textnormal{BP}$ from \fref{thm:wg2006_BP}.
\end{cor}
The proof of \fref{cor:LAMA_opti} follows from first noting that \fref{eq:GV_IO_fp} and \fref{eq:fixed_pt} are equal. Hence, since $\beta_\textnormal{LAMA}=\beta_\textnormal{BP}$, \LAMA has a unique fixed-point solution to \fref{eq:fixed_pt} given by $\sigma_\textnormal{IO}^2$ which is equivalent to $\No/\eta_\textnormal{BP}$.
Since \LAMA decouples the MIMO system into parallel AWGN channels \cite{BM2011} and the decoupled variances are equal, \LAMA achieves the same performance as the IO data detector.
 % by \fref{thm:wg2006_BP}.
% 
In \fref{sec:LAMA_optimal}, we provide conditions for which there is exactly one (unique) fixed point with minimum effective noise variance $\sigma^2$.

\subsection{Exact Recovery Thresholds (ERTs)}\label{sec:ERTs}

We start by analyzing \LAMA in a noiseless setting and for $\No=\Nopost=0$.
We provide sharp bounds on the system ratio $\beta=\MT/\MR$, which guarantee exact recovery of an unknown transmit signal $\vecs_0$ in the large-system limit. 
We show that if $\beta<\betamax$, where $\betamax$ is the so-called \emph{exact recovery threshold (ERT)}, then \LAMA perfectly recovers $\vecs_0$.
Note that the ERT depends on the constellation $\setO$ and resembles to the phase-transition behavior observed in sparse signal recovery~\cite{SR2012,DT2011,DT2010}; the key difference is that \LAMA operates with dense vectors.
% , as the entries in $\vecs_0$ are chosen from a finite constellation $\setO$.  

We will show in \fref{thm:recovery} that if $\beta<\betamax$, there exists a unique fixed point at $\sigma^2=0$ to the fixed-point equation in \fref{eq:fixed_pt}. The unique fixed point at $\sigma^2=0$ implies that the effective noise variance output for the decoupled AWGN channel will be zero. Therefore, the output from the non-linear MMSE estimate from \LAMA will be $\mathsf{F}(\bms_0,\sigma^2) = \bms_0$ from \fref{eq:DAMP_F}, and hence \LAMA perfectly recovers $\bms_0$.
% % $=\lim\limits_{\MT\rightarrow\infty}\frac{1}{\MT}\vecnorm{\hat\bms-\bms_0}^2=0$
% , where $\hat\bms$ is the non-linear MMSE estimate computed by \LAMA and $\bms_0$ is the original signal. 
% % 
% Note that the result $\sigma^2=0$ implies that $\hat\bms=\bms_0$ from construction of the non-linear MMSE estimate function $\mathsf{F}(\cdot,\sigma^2)$ in \fref{eq:DAMP_F}, and hence \LAMA perfectly recovers $\bms_0$.
% % 
% This result can also been seen by noting that $\sigma^2=0$ implies the output from the decoupled AWGN channel will be equivalent to $\bms_0$.
% 
% 
% 
For $\beta\geq\betamax$, perfect recovery cannot be guaranteed.\footnote{We assume the initialization as given in \fref{alg:LAMA}. LAMA may recover the original signal for $\beta\geq\betamax$ if initialized sufficiently close to the optimal fixed point; see \cite{ZMWL2015} for a discussion.} To make this behavior explicit, we need the following technical result with proof in \fref{app:DAMPsolvability}.

\begin{lem}\label{lem:DAMPsolvability} Fix the constellation set $\setO$ and let $\Varop_S[S]$ be finite. Then, there exists a non-negative gap $\sigma^2 - \Psi(\sigma^2,\sigma^2)\geq 0 $ with equality if and only if $\sigma^2 = 0$. As $\sigma^2\rightarrow 0$, we have MSE $\Psi(\sigma^2,\sigma^2)\rightarrow 0$; as $\sigma^2\rightarrow\infty$, we have the the MSE $\Psi(\sigma^2,\sigma^2)\rightarrow \Varop_S[S]$.
\end{lem}

%%%%%%%%%%%%%%%%%%%%%%%%%%%%%%%%%%%%%%%%%%%%%%%%

For a finite value of $\Varop_S[S]$, \fref{lem:DAMPsolvability} shows that we have $\Psi(\sigma^2,\sigma^2)<\sigma^2$ for all $\sigma^2 > 0$. 
Now, suppose that for some $\beta>1$, $\beta\Psi(\sigma^2,\sigma^2)<\sigma^2$ also holds for all $\sigma^2>0$.
Then, as long as $\beta>1$ is not too large to also ensure $\beta\Psi(\sigma^2,\sigma^2)<\sigma^2$, for all $\sigma^2>0$, there will only be a \emph{single} fixed point at $\sigma^2=0$.
Therefore, LAMA is able to perfectly recover the original signal $\bms_0$ by \fref{thm:CSE} since the unique fixed point at $\sigma^2=0$ implies that $\Psi(\sigma^2,\sigma^2)=0$. 
% 
% We define perfect recovery when the asymptotic normalized MSE goes to zero.
%
Leveraging the gap between $\Psi(\sigma^2,\sigma^2)$ and $\sigma^2$ will allow us to find the exact recovery threshold (ERT) of \LAMA for values of $\beta>1$. 
For the fixed (discrete) constellation $\setO$, the largest value of $\beta$ that ensures $\beta\Psi(\sigma^2,\sigma^2)<\sigma^2$ is precisely the ERT. 
% \cj{Arian wrote: ``Be careful about perfect recovery. We should either define what we mean by perfect recovery or .. We only know that asymptotic normalized MSE goes to zero.''}
% 

% \cj{``A comparison with replica is very helpful here. We know the recovery performance of LAMA for the noiseless setting. We can just mention replica now and derive the optimal performance. For finite constellation we realize that we can have perfect recovery even with one measurement. Does this make sense. So, in that sense LAMA is suboptimal in this region. I think it is very cool and it says that in the noiseless setting there is a lot of room for improvement here.  
% 
% We can also mention that LAMA under good initialization may converge to the right fixed point even if $\beta> \beta_O^{\rm max}$.''}
% 
% 
% 
% 
\begin{defi}
% [Exact Recovery Threshold] 
\label{def:ERT}Fix $\setO$ and let $\No=\Nopost=0$. Then, the exact recovery threshold (ERT) that enables perfect recovery by \LAMA is defined by
\begin{align}\label{eq:beta_recover}
\betamax &= \min_{\sigma^2\geq0}\!\left\{\!\left(
\frac{\Psi(\sigma^2,\sigma^2)}{\sigma^2}\right)^{\!\!-1}
\right\}\!.
\end{align}
\end{defi}
We are now ready to  establish perfect recovery with $\beta^\text{max}_\setO$; the proof is given in \fref{app:recovery}. 

% \fref{thm:recovery} provides precise conditions for which \LAMA is able to perfectly recover signals. 

\begin{thm} \label{thm:recovery} Let $\No=\Nopost=0$ and $\bH$ satisfy (A2). Fix the constellation $\setO$. 
If $\beta<\betamax$, then \LAMA perfectly recovers $\bms_0$ in \fref{eq:systemmodel} in the large-system limit.
\end{thm}
We emphasize that for a given constellation $\setO$, the ERT $\betamax$ can be computed numerically from \fref{eq:beta_recover}, where $\Psi(\sigma^2,\sigma^2)$ is given by \fref{thm:CSE}.
We emphasize that the signal variance, $\Varop_S[S]$ does not have an impact on the ERT as the MSE function $\Psi(\sigma^2,\sigma^2)$ and $\sigma^2$ both scale linearly with $\Varop_S[S]$.
In \fref{sec:LAMA_optimal}, we extend our analysis to the noisy case.
%% 

% \cs{capitalization of titles...}
\subsection{Optimality Conditions for \LAMA With Noise}\label{sec:LAMA_optimal}

We develop optimality conditions of \LAMA in the presence of noise, and we focus on mismatch-free case as the associated optimality conditions allow for an elegant analysis.\footnote{
The mismatch-free case requires us to identify all fixed points of \fref{eq:fixed_pt}, whereas mismatch case requires the identification of all fixed points to the \emph{coupled} fixed-point equations in \fref{eq:fixed_pt_JO}.
% , which depend on $\sigma^2$, $\gamma^2$, and $\beta$. 
% 
A detailed analysis of optimality conditions for \LAMA with noise variance mismatch is left for future work.}

In the presence of noise ($\No>0$), exact recovery is no longer guaranteed. Nevertheless, if \LAMA converges to a unique fixed-point, then we obtain the same error-rate performance as the IO data detector.
 % using \fref{eq:hard_decision}. 
% 
% \cs{redundant statements:}
In such situations, we call \LAMA to be optimal.
%
% We note that the unique fixed point of \LAMA obtained by \fref{eq:fixed_pt} coincides to the fixed-point solution of the IO data detector in \fref{eq:GV_IO_fp}. 
% 
Furthermore, if  multiple fixed-points exist, we call the fixed-point with minimum effective noise variance the \emph{optimal fixed point}, whereas all other fixed points are called \emph{suboptimal fixed points}. 

\begin{table}
\centering
\caption{Summary of (Sub-)Optimality Regimes of \LAMA}
\begin{tabular}{@{}cccc@{}}
\toprule
&$\beta<\betamin$\hspace*{-1mm}&$\betamin\!\leq\!\beta\!\leq\!\betamax$\hspace*{-1mm}&$\betamax<\beta$\\
\midrule
\hspace{15mm}$\No<\Nomin$\hspace*{-3mm}&\em optimal&\em optimal& suboptimal\\
$\Nomin\leq\No\leq\Nomax$\hspace*{-3mm}
&\em optimal& (sub-)optimal\tablefootnote{For some constellations, there may exist intervals in $[\Nomin,\Nomax]$ where \LAMA is still optimal; an example is shown in \fref{fig:SE_16PSK}.} &suboptimal\\
\hspace{-12mm}$\Nomax<\No$&\em optimal&\em optimal&\em optimal\\
\bottomrule
\end{tabular} %%% \phantom{}
\label{tbl:IOLAMAoptimal_reg}
\end{table}

In essence, there exist three different regimes for \LAMA (see \fref{tbl:IOLAMAoptimal_reg}), which depend on the system ratio $\beta$: 
(i) if $\beta$ is smaller than the so-called \emph{minimum recovery threshold} (MRT) $\betamin$, then \LAMA is \emph{always} guaranteed to converge to the unique fixed point (with minimal $\sigma^2$), i.e., the LAMA delivers IO data detection performance irrespective of the noise variance $\No$ 
(ii) if $\beta$ is larger or equal to than the MRT, but smaller than or equal to the ERT, then multiple fixed points exist. In this case, optimality of \LAMA depends on the noise variance $\No$. If the noise variance $\No$ is \emph{larger} than the so-called \emph{maximum guaranteed noise variance} $\Nomax$, then \LAMA converges to the unique fixed point. 
Similarly, if the noise variance $\No$ is strictly \emph{smaller} than the so-called \emph{minimum critical noise $\Nomin$}, then \LAMA converges to the optimal fixed point. However, if $\No\in[\Nomin,\Nomax]$, then \LAMA converges, in general, to a sub-optimal fixed point\footnote{We note that \LAMA can still be optimal if it was initialized close to the optimal fixed point \cite{Maleki2010phd}, but we exclude this case from our analysis.}.
We also note that for some constellations, there may exist intervals in $[\Nomin,\Nomax]$ in which \LAMA remains to be optimal. This behavior is shown in \fref{fig:SE_16PSK}.
Furthermore, as $\beta\rightarrow\betamax$, the minimum critical noise $\Nomin\rightarrow0$, which implies that \LAMA is optimal when 
% either optimal if $\No=0$, which agrees with the ERT discussed  in \fref{sec:ERTs}, \cs{this part of the sentence is not clear} or if
$\No>\Nomax$.
(iii) If $\beta$ exceeds the ERT, then \LAMA is optimal if $\No>\Nomax$. For all other values of $\No$, \LAMA converges, in general, to a sub-optimal fixed point.
In order to make these three regimes more explicit, we require the following definition.
\begin{defi}\label{def:MRT}Fix the constellation $\setO$ and let $\Nopost = \No$. Then, the minimum recovery threshold (MRT) $\betamin$ is defined as follows:
\begin{align}\label{eq:beta_badstate}
\betamin=\min_{\sigma^2\geq0}\!\left\{\!\left(\frac{\textnormal{d}\Psi(\sigma^2,\sigma^2)}{\textnormal{d}\sigma^2}\right)^{\!\!-1}\right\}\!.
\end{align}
\end{defi}

By the definition of the MRT, it is easy to observe that the fixed point of \fref{eq:fixed_pt} is unique for all system ratios $\beta<\betamin$, as $\beta \frac{\textnormal{d}\Psi(\sigma^2,\sigma^2)}{\textnormal{d}\sigma^2} < 1$ for all values of $\sigma^2$. 
The following lemma establishes an intuitive relationship between MRT and ERT; the proof is given in \fref{app:MRTERT}.
\begin{lem}The MRT never exceeds the ERT.\label{lem:MRTERT}
\end{lem}
\fref{lem:MRTERT} shows that if the system ratio $\beta$ is less than MRT, i.e., $\beta<\betamin$, then LAMA is not only optimal but also perfect recovery is possible in noiseless settings. 
We next define the minimum critical and maximum guaranteed noise variance, $\Nomin$ and $\Nomax$, that determine boundaries for the optimality regimes when $\beta\geq\betamin$. 

\begin{defi}\label{def:Nomin}Fix the system ratio $\beta\in[\betamin,\betamax]$. Then, the \emph{minimum critical} noise variance $\Nomin$ that ensures convergence to the optimal fixed-point is defined by
\begin{align*}
\Nomin &= \min_{\sigma^2\geq0}\!\left\{\sigma^2 - \beta\Psi\!\left(\sigma^2,\sigma^2\right)\!:\beta \frac{\textnormal{d}\Psi(\sigma^2,\sigma^2)}{\textnormal{d}\sigma^2} = 1
\right\}\!.
\end{align*}
\end{defi}
\begin{defi}\label{def:Nomax}Fix the system ratio $\beta\geq\betamin$. Then, the \emph{maximum guaranteed} noise variance $\Nomax$ that ensures convergence to the optimal fixed-point is defined by
\begin{align*}
\Nomax &= \max_{\sigma^2\geq0}\!\left\{\sigma^2 - \beta\Psi\!\left(\sigma^2,\sigma^2\right)\!:\beta \frac{\textnormal{d}\Psi(\sigma^2,\sigma^2)}{\textnormal{d}\sigma^2} = 1
\right\}\!.
\end{align*}
\end{defi}
Note that as $\beta\rightarrow\beta^\text{max}_\setO$, the minimum critical noise decreases to $\Nomin\rightarrow0$. 
To see this, consider the case when $\beta=\betamax$, so that there exists a $\sigma^2_\star>0$ such that $\beta^\text{max}_\setO\Psi(\sigma^2_\star,\sigma^2_\star)=\sigma^2_\star$. 
It is clear that $\betamax\!\left.\frac{\textnormal{d}\Psi(\sigma^2,\sigma^2)}{\textnormal{d}\sigma^2}\right\vert_{\sigma^2=\sigma_\star^2}\!=1$ and hence $\Nominnobeta(\betamax)=\sigma^2_\star-\betamax\Psi(\sigma_\star^2,\sigma_\star^2)=0$.  

Before we proceed with the analysis for optimality regimes of \LAMA, we present \fref{lem:fp_IOLAMA_0} (with proof in \fref{app:fp_IOLAMA_0}) that shows how the fixed-point $\sigma^2$ decreases  with $\No$ as $\No\to0$.

\begin{lem}
% [Highest fixed-point solution of \LAMA goes to 0 as $\No\to0$]
Fix the constellation $\setO$ and let $\beta<\betamax$. Denote $\sigma^2$ as the largest fixed-point solution of  \LAMA with noise variance $\No$. Then, as $\No\to0$, we have $\sigma^2\to0$. In addition, we have $\lim_{\No\to0}\frac{\sigma^2}{\No}=1$.\label{lem:fp_IOLAMA_0}
\end{lem}

\fref{lem:fp_IOLAMA_0} shows that not only the fixed-point solution $\sigma^2$ of LAMA goes to 0 as $\No\to0$, but also decreases linearly as $\lim_{\No\to0}\frac{\sigma^2}{\No}=1$. We now proceed to the optimality regime analysis.
We recall that all the zero-crossing points of the function
\begin{align}\label{eq:plotfixedfunction}
g(\sigma^2,\beta,\No,\setO)=\No+\beta\Psi(\sigma^2,\sigma^2)-\sigma^2
\end{align}
correspond to all the fixed points of the cSE of \LAMA. We will frequently refer to the function in \fref{eq:plotfixedfunction} for our optimality analysis of \LAMA.

Figure \ref{fig:SE_QPSK} illustrates our optimality analysis for a large MIMO system with QPSK. 
We plot the function \fref{eq:plotfixedfunction} depending on the effective noise variance $\sigma^2$ and for different system ratios $\beta$. The cases $\beta<\betamin$, $\beta\in[\betamin,\betamax]$, and $\beta>\betamax$ are shown in \fref{fig:SE_QPSK1}, \fref{fig:SE_QPSK2}, and \fref{fig:SE_QPSK3}, respectively.
The special case of $\beta=1$ in the noiseless setting $\No=0$ for \fref{eq:plotfixedfunction} corresponds to the solid blue line, along with the corresponding (unique) fixed point at the origin. In the following three paragraphs, we discuss the three operation regimes of \LAMA. 

\begin{figure*}[tp]
\centering
\subfigure[]{\includegraphics[height=0.25\textheight]{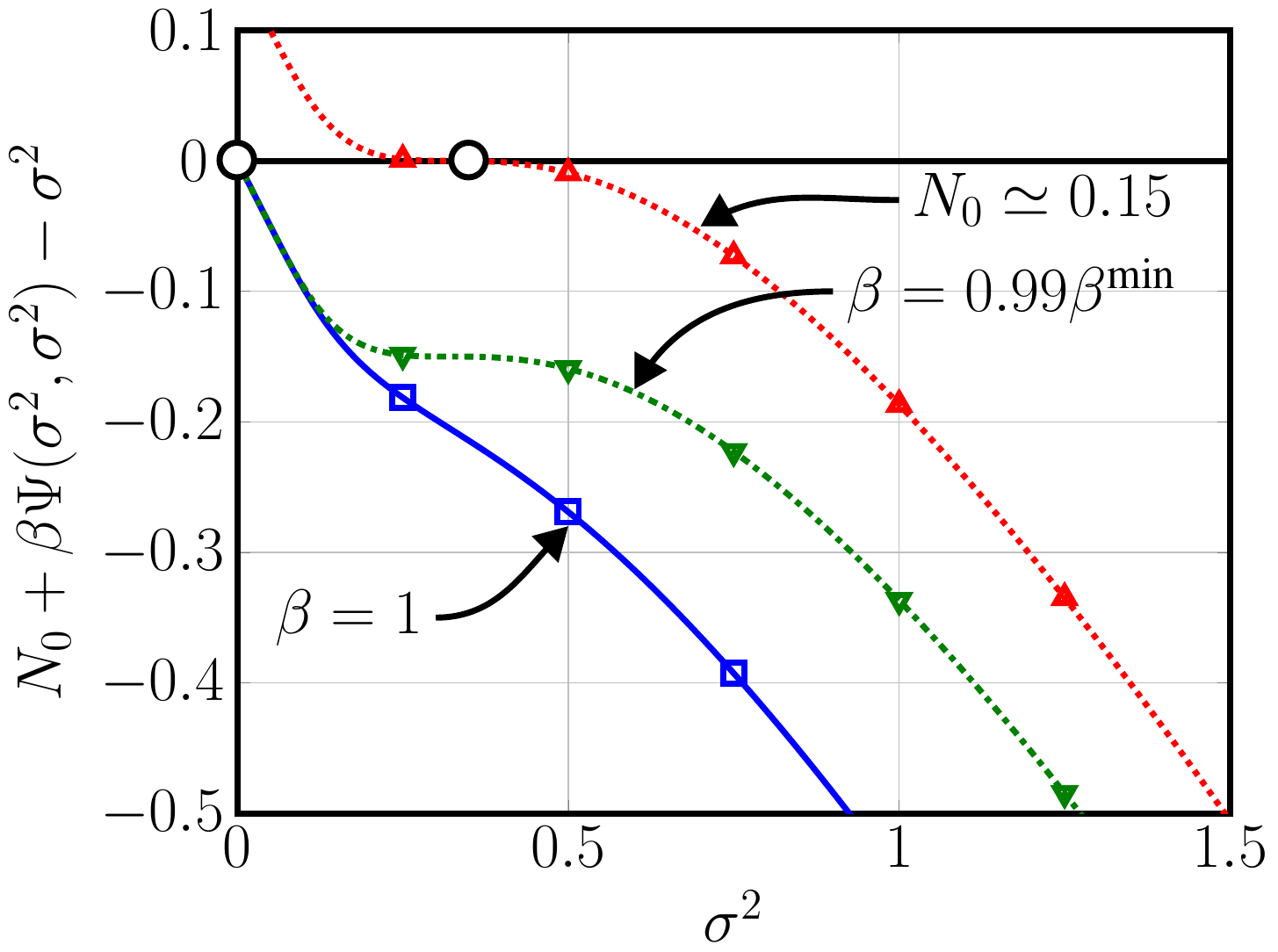}\label{fig:SE_QPSK1}}
\hspace{0.2cm}
\subfigure[]{\includegraphics[height=0.25\textheight]{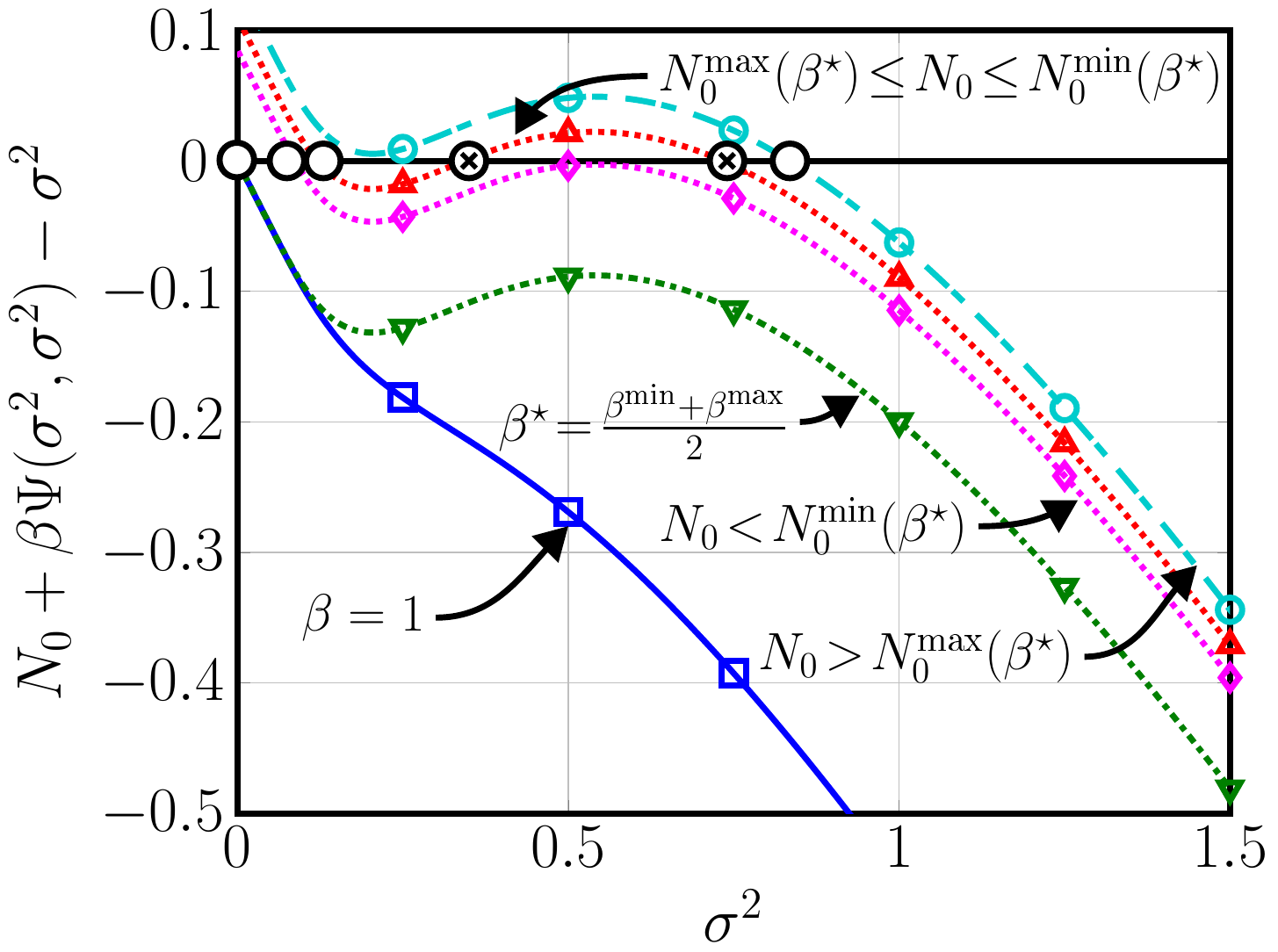}\label{fig:SE_QPSK2}}\\
\subfigure[]{\includegraphics[height=0.25\textheight]{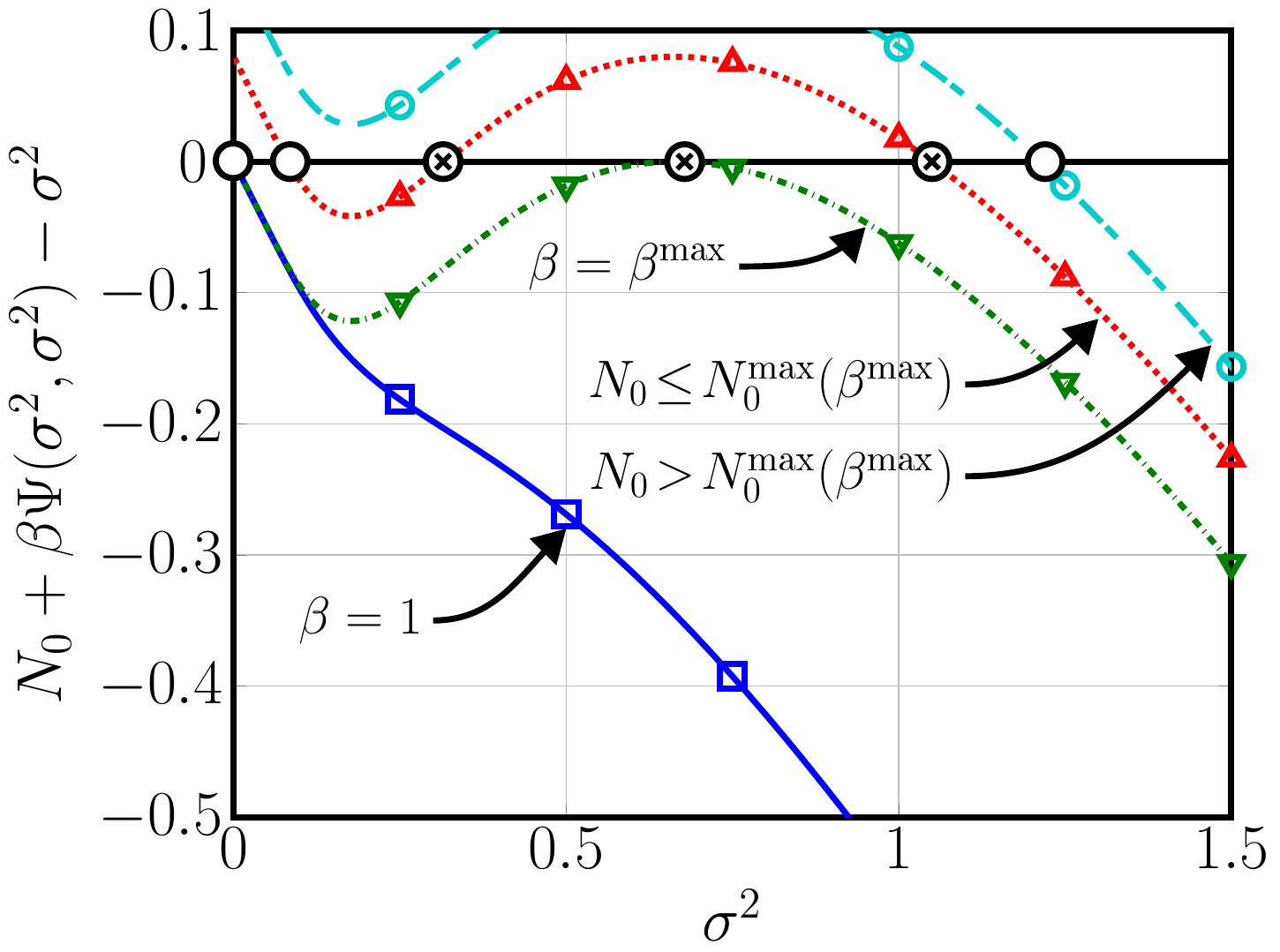}\label{fig:SE_QPSK3}}
\hspace{0.2cm}
\subfigure[]{\includegraphics[height=0.25\textheight]{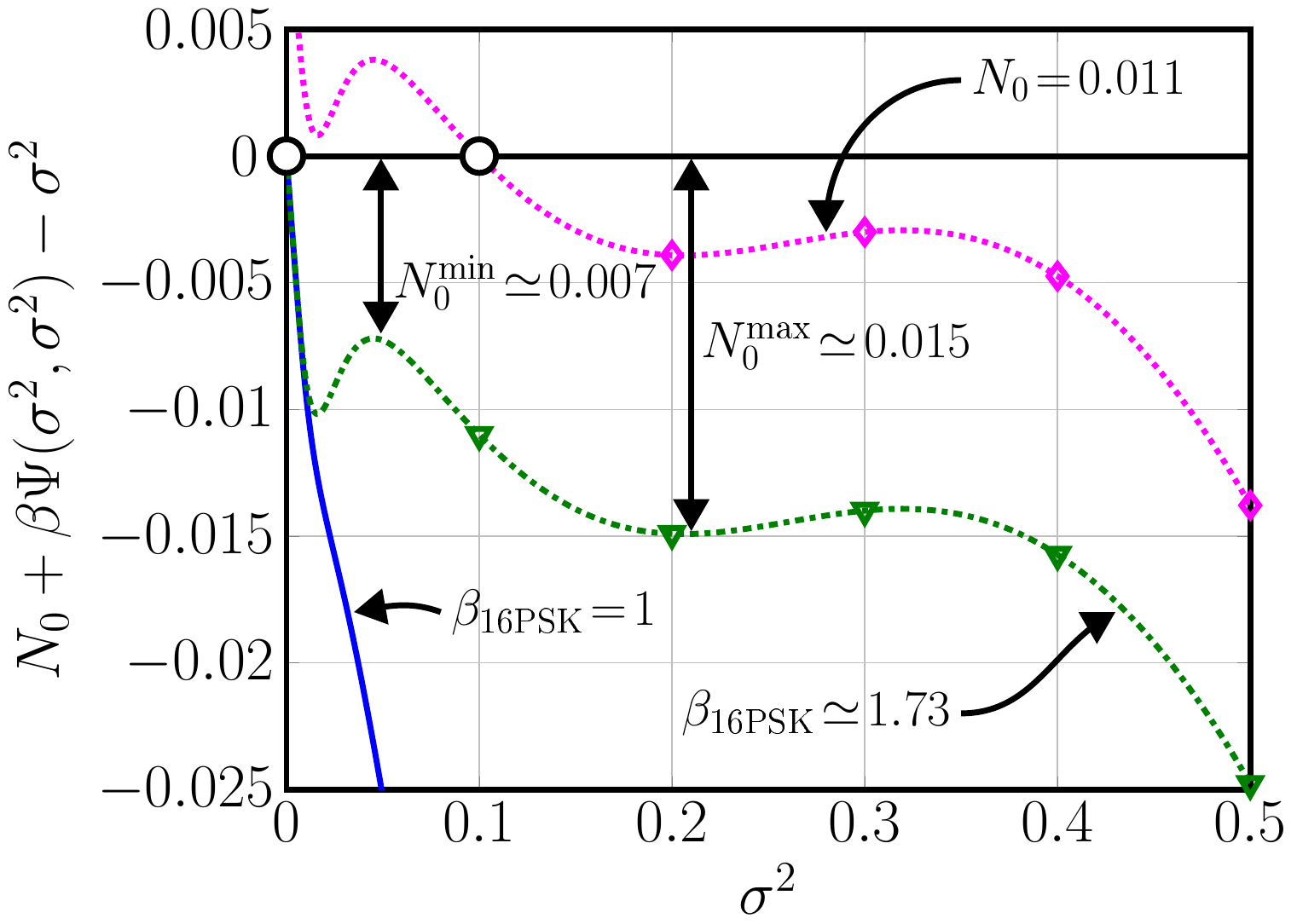}\label{fig:SE_16PSK}}

\caption{Function \fref{eq:plotfixedfunction} for three cases: (a) $\beta<\betamin$, (b) $\beta\in[\betamin,\betamax]$, and (c) $\beta>\betamax$ for QPSK normalized to $\Es=1$;  
(d) is for 16-PSK and $\beta=1.73\in[\betamin,\betamax]$. 
Optimal fixed points are indicated by {\large $\circ$}; suboptimal fixed points by $\otimes$.
(a) For $\beta<\betamin$, \LAMA always converges to the unique, optimal fixed point, irrespective of the noise variance $\No$. 
% 
% If $\beta=\betamin$ and $\Nomin=\Nomax$, then the fixed point corresponding to $\No=\Nomin$ remains to be unique, ensuring the optimality of \LAMA.
%
(b) For $\beta\in[\betamin,\betamax]$, we have two regimes for which \LAMA converges to an optimal fixed point: (i) $\No<\Nomin$ and (ii) $\No>\Nomax$. The situation $\beta^\star=\frac{\betamin+\betamax}{2}$ with (i) $\No<\Nominnobeta(\beta^\star)$ is shown with a purple dotted curve and (ii) $\No>\Nomaxnobeta(\beta^\star)$ is shown shown with a cyan dashed curve; we see that \LAMA exhibits a single (and hence, optimal) fixed point. However, if $\No\in[\Nomin,\Nomax]$, which is shown with a red dotted curve, the cSE of \LAMA exhibits multiple fixed points and hence, \LAMA is no longer IO. 
(c) For $\beta>\betamax$, the cSE of \LAMA converges to a suboptimal fixed point in the noiseless case $\No=0$, which is shown in green. However, when $\No>\Nomax$, the cSE of \LAMA converges to the optimal fixed point, which can be seen in the cyan dashed curve. If $\No\leq\Nomin$, then the cSE of \LAMA, shown in red dotted curve, has multiple fixed points and thus, is no longer IO.
(d) For 16-PSK and $\beta=1.73\in[\betamin,\betamax]$, $\Nomin$ and $\Nomax$ is computed to be $0.007$ and $0.015$ respectively. For 16-PSK, there exists regions where $\No\in[\Nomin,\Nomax]$ and \LAMA still achieves IO performance. 
}
\label{fig:SE_QPSK}
\end{figure*}

\subsubsection*{(i) $\beta<\betamin$} In this region, the cSE of \LAMA always converges to the unique, optimal fixed point.
For $\beta<\betamin$, the slope of \fref{eq:plotfixedfunction} is strictly-negative. Hence, as \fref{eq:plotfixedfunction} is always decreasing, there exists exactly one unique fixed point for the cSE of \LAMA regardless of the noise variance~\No. Thus,  \LAMA  achieves  IO  performance.
The green dash-dotted and red dotted line in \fref{fig:SE_QPSK1} show \fref{eq:plotfixedfunction} for \mbox{$\beta<\betamin$} with \mbox{$\No=0$} and \mbox{$\No\simeq 0.15$}, respectively. In both cases, we see that the cSE of \LAMA converges to the unique fixed point. 
% Note that for any $\beta<\betamin$ regardless of \No, the cSE of \LAMA always recovers the unique fixed point and achieves the same performance as (IO).

\subsubsection*{(ii) \mbox{$\betamin\leq\beta\leq\betamax$}}
In this region, the cSE of \LAMA converges to the unique, optimal fixed point if \mbox{$\No<\Nomin$} or if \mbox{$\No>\Nomax$} and consequently, \LAMA achieves IO performance in both of these regimes. 
% 
% We note that there may exist one fixed point when $\beta=\betamin$, if there exists a unique $\sigma_\star^2$ that satisfies the equation $\betamin\frac{\textnormal{d}}{\textnormal{d}\sigma^2}\Psi(\sigma^2,\sigma^2)\big\vert_{\sigma^2=\sigma_\star^2}=1$. 
% % 
% If $\sigma_\star^2$ is unique, then for all other $\sigma^2\neq\sigma_\star^2$, we have $\betamin\frac{\textnormal{d}}{\textnormal{d}\sigma^2}\Psi(\sigma^2,\sigma^2)<1$, and thus, the fixed point of \fref{eq:plotfixedfunction} remains to be unique even when $\beta=\betamin$. 
% % 
% In this setting, $\Nominnobeta(\betamin) = \Nomaxnobeta(\betamin)$ and \LAMA is optimal even when $\No=\Nominnobeta(\betamin) = \Nomaxnobeta(\betamin)$.
%
The green dash-dotted line, cyan dashed line, and magenta dotted line in \fref{fig:SE_QPSK2} show \fref{eq:plotfixedfunction} for $\beta^\star=(\betamin+\betamax)/{2}$ with \mbox{$\No=0$}, \mbox{$\No>\Nomaxnobeta(\beta^\star)$} and \mbox{$\No<\Nominnobeta(\beta^\star)$}, respectively. We note that for the three cases the fixed point is unique, labeled in \fref{fig:SE_QPSK2} by a circle. 
The red, dotted line in \fref{fig:SE_QPSK2} shows \fref{eq:plotfixedfunction} with~$\beta^\star$ for noise $\No\in[\Nominnobeta(\beta^\star),\Nomaxnobeta(\beta^\star)]$.  
In this case, however, we observe that the cSE of \LAMA converges to the rightmost suboptimal fixed point labeled by the crossed circle $\otimes$. Hence, \LAMA is able to achieve IO performance if $\Nomin\leq\No\leq\Nomax$.

\subsubsection*{(iii) \mbox{$\beta>\betamax$}}
In this region, the cSE of \LAMA converges to the unique, optimal fixed point when \mbox{$\No>\Nomax$} and consequently, achieves IO performance. 
Unlike the previous case for $\betamin\leq\beta\leq\betamax$, for which \LAMA has \emph{two} regions of optimality, $\No>\Nomax$ and $\No<\Nomin$, for $\beta>\betamax$, \LAMA has only one optimal region: $\No>\Nomax$. 
As $\beta\rightarrow\betamax$, the low noise \mbox{$\No<\Nomin$} (or high \SNR) region of optimality disappears because $\Nomin\rightarrow0$ as $\beta\rightarrow\betamax$ from \fref{eq:beta_recover}.
%%%
The green, dash-dotted line and red, dotted lines in \fref{fig:SE_QPSK3} show \fref{eq:plotfixedfunction} for $\beta=\betamax$ with $\No=0$ and $0<\No\leq\Nomax$, respectively. 
We observe that the cSE of \LAMA converges to the suboptimal fixed point when $\beta=\betamax$ even with $\No=0$. The cyan, dashed line refers to $\beta=\betamax$ with $\No>\Nomax$. 
While the noiseless case enables the cSE of \LAMA to converge to the suboptimal fixed point, we observe that for high noise (or equivalently low \SNR), the cSE of \LAMA is able to achieve IO performance.
Therefore, if \mbox{$\beta>\betamax$}, then \LAMA achieves IO performance whenever the noise variance  exceeds the maximum guaranteed~noise~variance~$\Nomax$.

% 16PSK
We also note that for some constellations, the cSE of \LAMA may recover the optimal fixed point for $\beta\in[\betamin,\betamax]$ in some noise variance intervals $[\Nomin,\Nomax]$. 
An example case for $\beta=1.73$ with 16-PSK is shown in \fref{fig:SE_16PSK}, where cSE of \LAMA recovers the unique fixed-point with $\No=1.1\cdot 10^{-2}\in[\Nomin,\Nomax]$. 
These intervals exist for some constellations because in addition to $\sigma^2$ that result $\Nomin$ and $\Nomax$, there are multiple values of $\sigma^2$ that satisfy \mbox{$\frac{\textnormal{d}}{\textnormal{d}\sigma^2}g(\sigma^2,\beta,\No,\setO)=0$}, where $g(\sigma^2,\beta,\No,\setO)$ is defined in \fref{eq:plotfixedfunction}. 
As a result, there exist intervals between $\Nomin$ and $\Nomax$ that the cSE of \LAMA has one (optimal) fixed point. In such regions, LAMA enables IO performance.
% 
%However, as these results only hold for special constellations that satisfy has multiple $\sigma^2$ for $\frac{\textnormal{d}}{\textnormal{d}\sigma^2}g(\sigma^2,\beta,\No,\setO)=0$, we 
% disregard the details now and 
%leave it for future work.
%
We finally note that the MRT $\betamin$ and ERT $\betamax$ do not depend on the signal variance $\Varop_S[S]$. 
% of the original signal distributed $S\sim p(s)$
In contrast, the critical noise levels $\Nomin$ and $\Nomax$ depend on $\Varop_S[S]$.
% 

% \cj{Arian:There is a wonderful phenomenon that happens in the asymptotic that I guess it is worth mentioning. Consider asymptotic MSE of the algorithm as a function of $\sigma_w^2$. This function exhibits discontinuity. I think this is surprising and is probably worth mentioning it in the paper with one more graph!}

\subsection{Decomposing Complex-Valued Systems}\label{sec:highersystemSE}

We now analyze whether the cSE of LAMA with complex-valued constellations can equivalently be characterized by a real-valued SE with a real-valued constellation.
We note that while the loading factor limits were given in \cite{T2000} and \cite{TanakaCDMA} respectively, these results were pertinent to BPSK with real-valued systems and no results were given for other constellations.

We note that the standard way of dealing with complex-valued systems is via the real-valued decomposition (see footnote~1).
This approach, however,  violates the independent assumption on the MIMO channel. 
Since \LAMA operates directly on the complex plane, no transformation into the real-valued domain is required.
Nevertheless, we now provide conditions for which the complex-valued problem can be exactly characterized by a corresponding real-valued problem. 
For our analysis, we require the following definition.

%We first note that for certain constellations, the corresponding complex-valued state evolution equations in \fref{eq:SErecursion} and \fref{eq:SErecursion2} can be transformed into a real-valued state evolution. 
% 
%The following definition characterizes the types of complex-valued constellations that we will decompose into real-valued constellations for state evolution analyses.

\begin{defi}
\label{def:separable}
For all $s\in\setO$, express $s$ as $s = a + ib$,  where $a\in\realpart{\setO}$, $b\in\imagpart{\setO}$. Then, the constellation $\setO$ is called \emph{separable} if $p(s) = p(a)p(b)$ holds for all $s\in\setO$ and $\realpart{\setO}=\imagpart{\setO}$.
\end{defi}

For example, $M^2$-QAM with equally likely symbols is separable.
%, because $p(a) = 1/M^2 = 1/M \cdot 1/M = p(\ar) p(\ai)$, where $\ar$ and $\ai$ are from $M$-PAM constellations. 
In contrast, $M^2$-PSK is not separable (except for QPSK) as the real and imaginary parts dependent. 
We now present a result that allows us to transform the complex-valued cSE equations in \fref{eq:SErecursion} and \fref{eq:SErecursion2} into equivalent real-valued SE equations; the proof is given in \fref{app:separability}.

\begin{lem}\label{lem:separability}
Let the constellation $\setO$ be separable. Define $\Sr = \realpart{S}$ and denote the real-part of $\setO$ as $\setO^\textnormal{R}$. 
Define $\mathsf{F}^\textnormal{R}$ and $\mathsf{G}^\textnormal{R}$ as the message mean and variance function, respectively, with $\Sr \sim p(\realpart{S})$. Also define the MSE function $\Psi$ and the variance function $\Phi$ for the real-valued prior $\Sr$ as:
\begin{align*}
\Psi^\textnormal{R}(\sigma^2,\gamma^2) &= \Exop_{\Sr,Z_\textnormal{R}}\!\left[\!\left(\mathsf{F}^\textnormal{R}(\Sr + \sigma Z_\textnormal{R},\gamma^2) - \Sr\right)^{\!2}\right]\!,
\\
\Phi^\textnormal{R}(\sigma^2,\gamma^2) &= \Exop_{\Sr,Z_\textnormal{R}}\!\left[
\mathsf{G}^\textnormal{R}(\Sr + \sigma Z_\textnormal{R},\gamma^2)\right]\!,
\end{align*}
where $Z_\textnormal{R}\sim\setN(0,1)$. 
Then we have the following relation for $\Psi$ and $\Phi$ between the complex-valued constellation $\setO$ and the real-valued constellation $\setO^\textnormal{R}$:
\begin{align*}
\Psi(\sigma^2,\gamma^2) = 2 \Psi^\textnormal{R}\!\left(\frac{\sigma^2}{2},\frac{\gamma^2}{2}\right),
% \\
\quad
% \quad \text{and} \quad
\Phi(\sigma^2,\gamma^2) = 2 \Phi^\textnormal{R}\!\left(\frac{\sigma^2}{2},\frac{\gamma^2}{2}\right)\!.
\end{align*}
Therefore, the cSE recursions in \fref{eq:SErecursion} and \fref{eq:SErecursion2} are given by:
\begin{align}
\label{eq:psi_seprable}
\sigma_t^2 &= \No + \beta\Psi(\sigma^2_t,\gamma^2_t) = \No+2\beta \Psi^\textnormal{R}\!\left(\frac{\sigma^2_t}{2},\frac{\gamma^2_t}{2}\right)\!,
% \Psi^\textnormal{R}(\sigma^2,\gamma^2) = \Exop_{\Sr,Z_\textnormal{R}}\!\left[\!\left(\mathsf{F}^\textnormal{R}(\Sr + \sigma Z_\textnormal{R},\gamma^2) - \Sr\right)^{\!2}\right],
\\
\label{eq:phi_seprable}
\gamma_t^2 &= \Nopost + \beta\Phi(\sigma^2_t,\gamma^2_t) = \Nopost+2\beta \Phi^\textnormal{R}\!\left(\frac{\sigma^2_t}{2},\frac{\gamma^2_t}{2}\right)\!.
\end{align}
\end{lem}

We  note that LAMA operates \emph{simultaneously} on complex-valued signals by reducing $\sigma^2_t$ each iteration in both real and imaginary parts independently; this can be seen by noting that since $\setO$ is separable, $\Psi^\textnormal{R}$ is identical for both the real and imaginary parts of $\setO$. 
In addition, \fref{lem:separability} shows that if $\setO$ is separable, then the cSE can be transformed into a real-valued SE, hence validating the relation between the complex-valued constellation and the equivalent real-valued representation.  
This transformation implies that for certain constellations, the message mean $\mathsf{F}$ and variance function $\mathsf{G}$ can be computed (often more efficiently) in parallel for real and imaginary dimensions.

% \cs{is that correct? isn't it separate for real and imaginary? not just for one of the two?} by only considering the real dimension. 
% %

We note that in \cite{GV2005} Guo and Verd\'u used  a real-valued decomposition and the replica method for analyzing the performance of complex-valued signals for separable constellations and concluded that the error performance for complex signals is exactly same as that of real-valued system with transmit energy halved. 
\fref{lem:separability} supports this conclusion. 
Moreover, we emphasize that the cSE holds for \emph{general} constellations, such as higher-order PSK constellations, and LAMA can be used for data detection in such cases.

\subsection{ERT, MRT, and Critical Noise Levels}\label{sec:ERTMRT}

\begin{table}
\centering
\caption{ERTs $\beta_\setO^\text{max}$, MRTs $\beta_\setO^\text{min}$ and the critical noise levels $\Nominnobeta(\betamin)$ and $\Nomaxnobeta(\betamax)$
% and minimum critical $\SNR^\text{min}(\beta)$ at $\beta=\beta_min$ 
for \LAMA with common PSK, PAM, and QAM constellations}
\begin{tabular}{llcccc}
\toprule
\multicolumn{2}{c}{Constellation} & \multirow{2}{*}{$\betamin$} & \multirow{2}{*}{$\Nominnobeta(\betamin)$}
% & $\frac{\betamin E_s}{\Nominnobeta(\betamin)}$ [dB]
& \multirow{2}{*}{$\betamax$}
& \multirow{2}{*}{$\Nomaxnobeta(\betamax)$}
% & $\frac{\betamax E_s}{\Nomaxnobeta(\betamax)}$ [dB]
\\
$\complexset$ system & $\reals$ system \\
\midrule
%BPSK compute over 1e-4:1e-4:1
BPSK & -- & 2.951 & $3.00\cdot 10^{-1}$ & 4.171& $2.43\cdot 10^{-1}$ \\
%QPSK compute over 1e-4:1e-4:1
QPSK & BPSK & 1.475 & $1.50\cdot 10^{-1}$ & 2.086& $1.22\cdot 10^{-1}$\\
% 16qam compute over 1e-4:1e-4:0.3
16-QAM & 4-PAM & 0.983 & $3.00\cdot 10^{-2}$ & 1.363 & $2.45\cdot 10^{-2}$\\
% 64qam compute over 1e-4:1e-4:0.3
64-QAM & 8-PAM & 0.842 & $7.14\cdot 10^{-3}$& 1.157 & $5.87\cdot 10^{-3}$\\
% 256qam compute over 1e-5:1e-5:0.02
256-QAM& 16-PAM & 0.786 & $1.77\cdot 10^{-3}$ & 1.075 & $1.45\cdot 10^{-3}$ \\
% 1024qam compute over 1e-5:1e-5:0.005
% 1024-QAM & 32-PAM & 0.7608 & $4.400\cdot 10^{-4}$ & 1.0367& $3.628\cdot 10^{-4}$ \\
8-PSK & -- & 1.458 & $4.44\cdot 10^{-2}$ & 1.804 & $3.83\cdot 10^{-2}$\\
% 16psk betamin compute over ; betamax compute over 1e-4:1e-4:0.5
16-PSK & -- & 1.473 & $1.14\cdot 10^{-2}$ & 1.801& $9.95\cdot 10^{-3}$\\
% 64PSK betamin compute over 1.5e-3:1e-6:1.6e-3
64-PSK & -- & 1.474 & $7.23\cdot 10^{-4}$ & 1.801& $8.39\cdot 10^{-3}$ \\
% 256PSK betamin compute over 9.68e-5:1e-8:9.7e-5, betamax compute over 0.385:1e-4:0.387
256-PSK & -- & 1.474 & $4.52\cdot 10^{-5}$& 1.801 & $8.39\cdot 10^{-3}$ \\
% 1024PSK betamin compute over 0.4e-5:1e-7:1e-5, betamax compute over [0.148:1e-4:0.150 0.385:1e-4:0.387]
% 1024-PSK & -- & 1.4741 & $2.825\cdot 10^{-6}$ & 1.8005 & $8.389\cdot 10^{-3}$\\
% Unit-circle & -- & 1.7074 & 0.0176 & 19.8561 & 
% 1.8005 & 0.0084 & 23.3171\\
\bottomrule
\end{tabular}
\label{tbl:exact_recovery}
\end{table}

The ERT, MRT, as well as the critical noise levels $\Nomin$ and $\Nomax$ for common constellations and for real-valued as well as complex-valued systems are summarized in \fref{tbl:exact_recovery}.
We assume equally likely priors with the constellation sets normalized to $\Varop_S[S]=E_s=1$.
% .\footnote{The critical noise levels depend on $E_s$. We assume, without loss of generality, $E_s=1$.} 
%
We note that the calculations of ERT and MRT for the simplest case with BPSK involve computations of logistic-normal integrals for which no closed-form expressions are known \cite{P2013} but approximations exist \cite{P2013,C2013,D2005}.
The results in \fref{tbl:exact_recovery} were obtained via numerical integration to compute the MSE function $\Psi(\sigma^2,\sigma^2)$.\footnote{We used MATLAB's {\texttt{integral}} and {\texttt{integral2}} commands with {$\texttt{AbsTol}=\texttt{RelTol}=10^{-12}$.}}
Next Lemma
% \fref{lem:separable_const_same} 
shows that for real- and separable complex-valued constellations, the ERT and MRT are identical for real- and complex-valued systems, respectively; a short proof is  given in \fref{app:separable_const_same}.
For an example, BPSK for real-valued systems and QPSK for complex-valued systems have identical ERT and MRT of 1.475 and 2.086, respectively.

\begin{lem}
% [For separable constellations, real and complex-valued system have the same MRT and ERT.]
\label{lem:separable_const_same}Fix a separable constellation $\setO$ and denote $\betaminno_\complexset$, $\betamaxno_\complexset$ and $\betaminno_\reals$, $\betamaxno_\reals$ as MRT and ERT of the complex and real-valued constellation, respectively. Also, denote the critical noise levels $N_{0,\complexset}^\text{max}(\beta)$, $N_{0,\complexset}^\text{min}(\beta)$, $N_{0,\reals}^\text{max}(\beta)$, and $N_{0,\reals}^\text{max}(\beta)$ for the complex and real-valued constellation, respectively.
Then, $\betaminno_\complexset=\betaminno_\reals$, $\betamaxno_\complexset=\betamaxno_\reals$, $N_{0,\complexset}^\text{max}(\beta)=2N_{0,\reals}^\text{max}(\beta)$, and $N_{0,\complexset}^\text{min}(\beta)=2N_{0,\reals}^\text{min}(\beta)$.
\end{lem}

% \cs{I changed this sentence. OK?} 
\fref{lem:separable_const_same} implies that optimality results for $M^2$-QAM in a complex system with equally likely transmit symbols (shown in \fref{tbl:exact_recovery}) are the same for a real-valued $M$-PAM system.
Moreover, between BPSK  and QPSK in a complex system, we observe that all the thresholds differ by a factor of $2$, which is  expected.
As shown in the second row of \fref{tbl:exact_recovery} for QPSK with complex noise, or a real-valued BSPK system (with real noise) the ERT is $\betamaxno_\text{QPSK}\approx2.0855$, which corresponds exactly to the maximum loading factor for the IO data detector established in \cite{TanakaCDMA,GV2005}. Moreover, the MRT for QPSK is given as $\betaminno_\text{QPSK}\approx1.4752$\cite{TanakaCDMA}.\footnote{Note that $\betaminno_\text{QPSK}$ Tanaka provided in \cite{TanakaCDMA} is $1.49$, whereas we obtain a slightly more accurate value $1.4752$.}
The critical noise values in \fref{tbl:exact_recovery} refer to complex constellations as the critical noise values can be easily computed for the real constellation by \fref{lem:separable_const_same}.

The MRTs for 16-QAM and 64-QAM indicate that small system ratios $\beta<1$ are necessary to guarantee that \LAMA achieves IO performance. 
For instance, we require $\beta\leq\beta^\text{min}_\text{64-QAM}\approx0.8424$, i.e. $\MT\leq0.8424\MR$,  to ensure that \LAMA solves \fref{eq:IO} for 64-QAM. As $\beta\rightarrow\beta^\text{max}_\text{64-QAM}\approx1.1573$, \LAMA is only optimal in settings in which the noise level is rather high, i.e., where $\No>\Nomaxnobeta(\beta^\text{max}_\text{64-QAM})\approx5.868\cdot 10^{-3}$, or, equivalently, when $\SNR<22.9495$ dB. 
From \fref{tbl:exact_recovery}, we see that higher-order QAM or PSK constellations can be decoded optimally by LAMA in massive MIMO as one typically assumes $\MR\gg\MT$. 
% 
%
%% PSK
We also observe that as $M$ increases for $M$-PSK, $\betamin$ and $\betamax$ approaches to $1.4741$ and $1.8005$, respectively. 
%
% \cs{what's this sentence telling us? it kinda leads to no where} 
% The quantity $\betamax$ for $M$-PSK systems seems to converge to $1.8005$. As the number of constellation points are increased by a factor of four for 16-PSK, 64-PSK, 256-PSK, the corresponding $\Nominnobeta(\betamin)$ decreases by a factor of $16$, whereas $\Nomaxnobeta(\betamax)$ remains constant at $8.389\cdot 10^{-3}$.

% % unit circle
% We also compare the MRT and ERTs for an $M$-PSK system with that of the unit circle set $C=\{s:\vecnorm{s}_2=1\}$.
%  % and highlight the differences.
% %
% We first highlight that $M$-PSK as $M\rightarrow\infty$ differs from the unit circle set $C$, because the former is a countably infinite set while the latter is uncountable.
% %
% Because of this difference, $\frac{\text{d}}{\text{d}\sigma^2}\Psi(\sigma^2,\sigma^2)\rightarrow0$ for an $M$-PSK set as claimed in \fref{lem:slopeMSEgoes0}, but the same result does not hold for the unit circle, because it is an uncountable set \cite{WV2010}.
% %
% The $\betaminno=1.4741$ of $M$-PSK arises because there exists a sharp increase in $\Psi(\sigma^2,\sigma^2)$ as $\sigma^2$ increases from 0 satisfying $\frac{\text{d}}{\text{d}\sigma^2}\Psi(\sigma^2,\sigma^2)=\frac{1}{1.4741}$. As $M$ is increased by a factor of 2, the $\sigma^2$ satisfying $\betaminno=1.4741$ was measured to decrease by a factor of 4. We note that this holds as $M\rightarrow\infty$ because $M$-PSK is a discrete set \cite{WV2010}.
% %

%% file: sec5-conclusion.tex
% !TEX root = 18TIT-lama.tex

%%%%%%%%%%%%%%%%%%%%%%%%%%%%%%%%%% 
% new section
%%%%%%%%%%%%%%%%%%%%%%%%%%%%%%%%%%%%
\section{Numerical Results and Practical Considerations}\label{sec:numerical}

We now provide numerical results for \LAMA, discuss practical implementation aspects, and highlight the pros and cons. 
In what follows, we use the average received SNR defined in \fref{eq:SNR}.
 % by $\SNR=\beta\frac{E_s}{\No}$.
%Recall that all the aforementioned results were based on the large system limit assumption, i.e. letting $\beta = \MT/\MR$ and $\MT\rightarrow\infty$. While we showed that IO and JO-\LAMA are the individual and jointly optimal detectors for appropriate values of $\beta$ respectively, this only pertains to the large limit and thus, does not hold for practical systems (finite antenna configurations). In this section, we highlight some considerations of IO and JO-\LAMA for practical systems and their limitations.

\subsection{Achievable Rates and Error-Rate Performance}

As detailed in \fref{sec:LAMAdecouple}, the output of \LAMA enables one to represent each transmit stream by a single-input single-output AWGN channel with a equal noise variance $\sigma^2_t$ that can be computed via the cSE \fref{thm:CSE}.
Therefore, the performance of \LAMA in the large-system limit can be characterized by analyzing a single AWGN channel.

Figures~\ref{fig:IO_cap} and \ref{fig:IO_SER} show the achievable rate and symbol error rate for \LAMA after $100$ iterations for various system ratios $\beta$. 
While an infinite number of iterations would guarantee \LAMA to converge to a fixed point solution, our results show that much fewer than $100$ iterations are required for \LAMA to converge; we will further discuss this aspect in \fref{sec:performance_tradeoff}.

% Capacity

\begin{figure*}[tp]
\centering
\subfigure[]{\includegraphics[height=0.25\textheight]{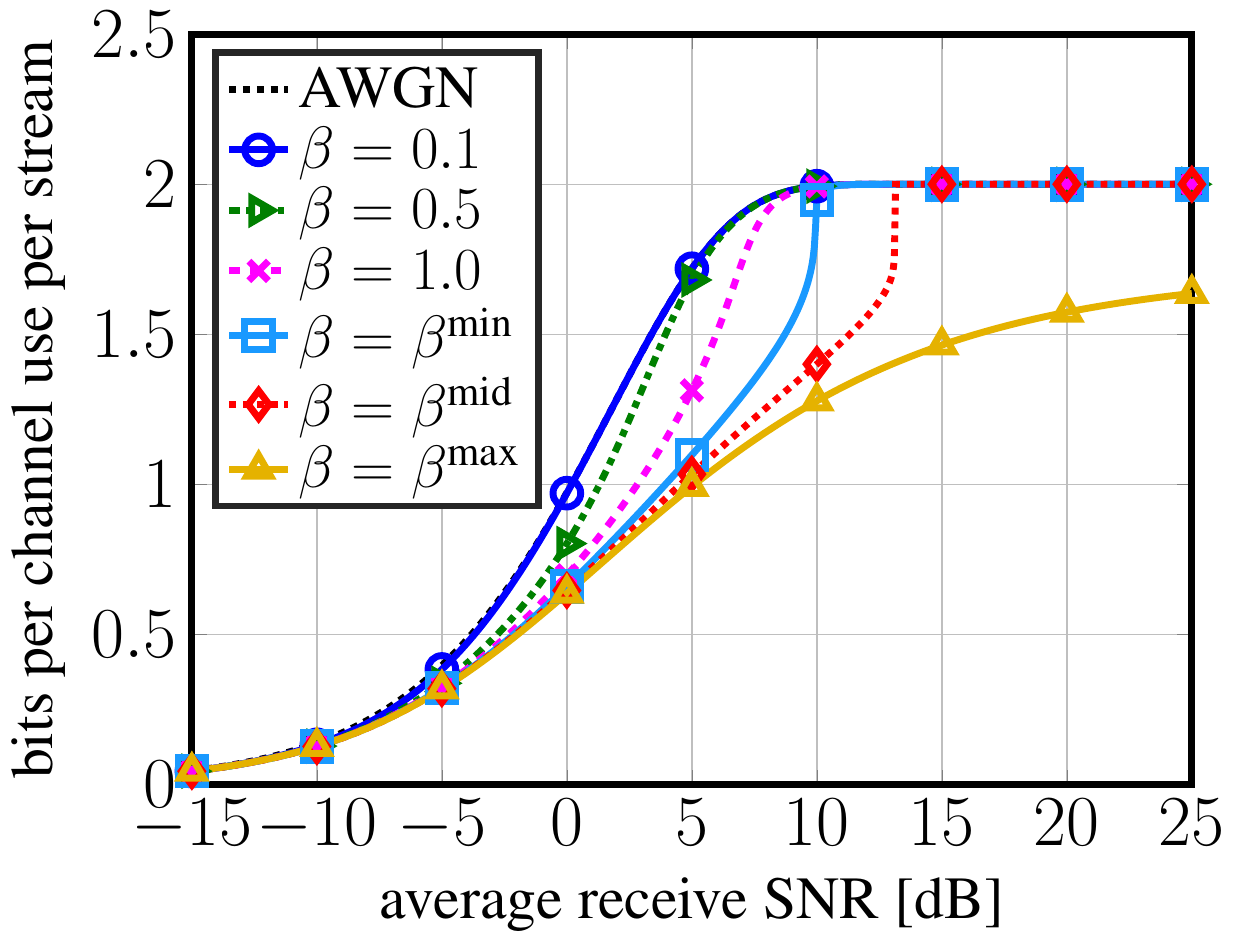}\label{fig:IO_cap}}
% \hspace{0.4cm}
% \subfigure[]{\includegraphics[width=0.45\columnwidth]{./source/JO_capacity.pdf}\label{fig:JO_cap}}
\hspace{0.4cm}
\subfigure[]{\includegraphics[height=0.25\textheight]{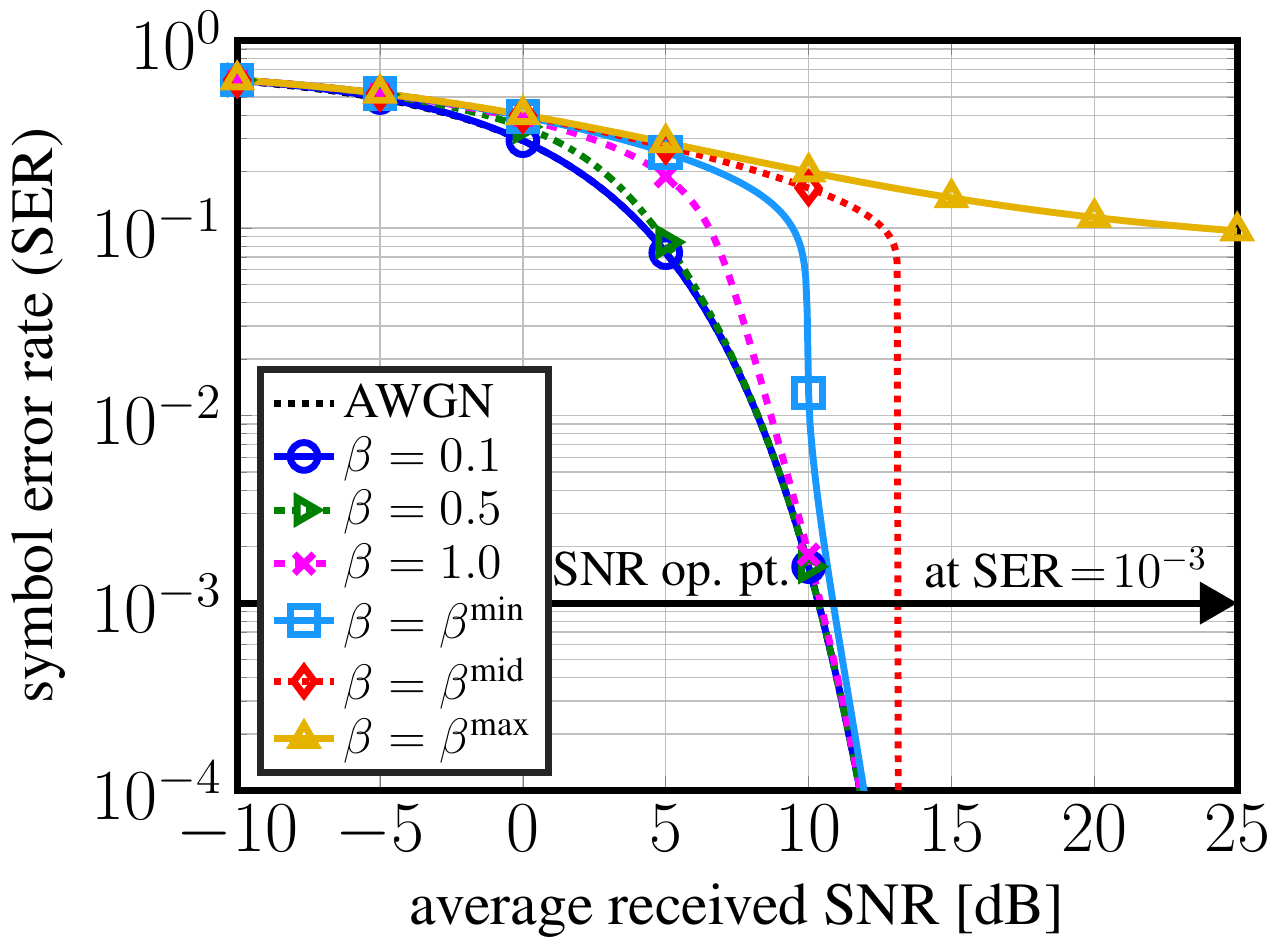}\label{fig:IO_SER}}
% \hspace{0.4cm}
% \subfigure[]{\includegraphics[width=0.45\columnwidth]{./source/JO_SER.pdf}\label{fig:JO_SER}}
% \hspace{0.1cm}
%% remove JO-LAMA, add complexity
\subfigure[]{\includegraphics[height=0.25\textheight]{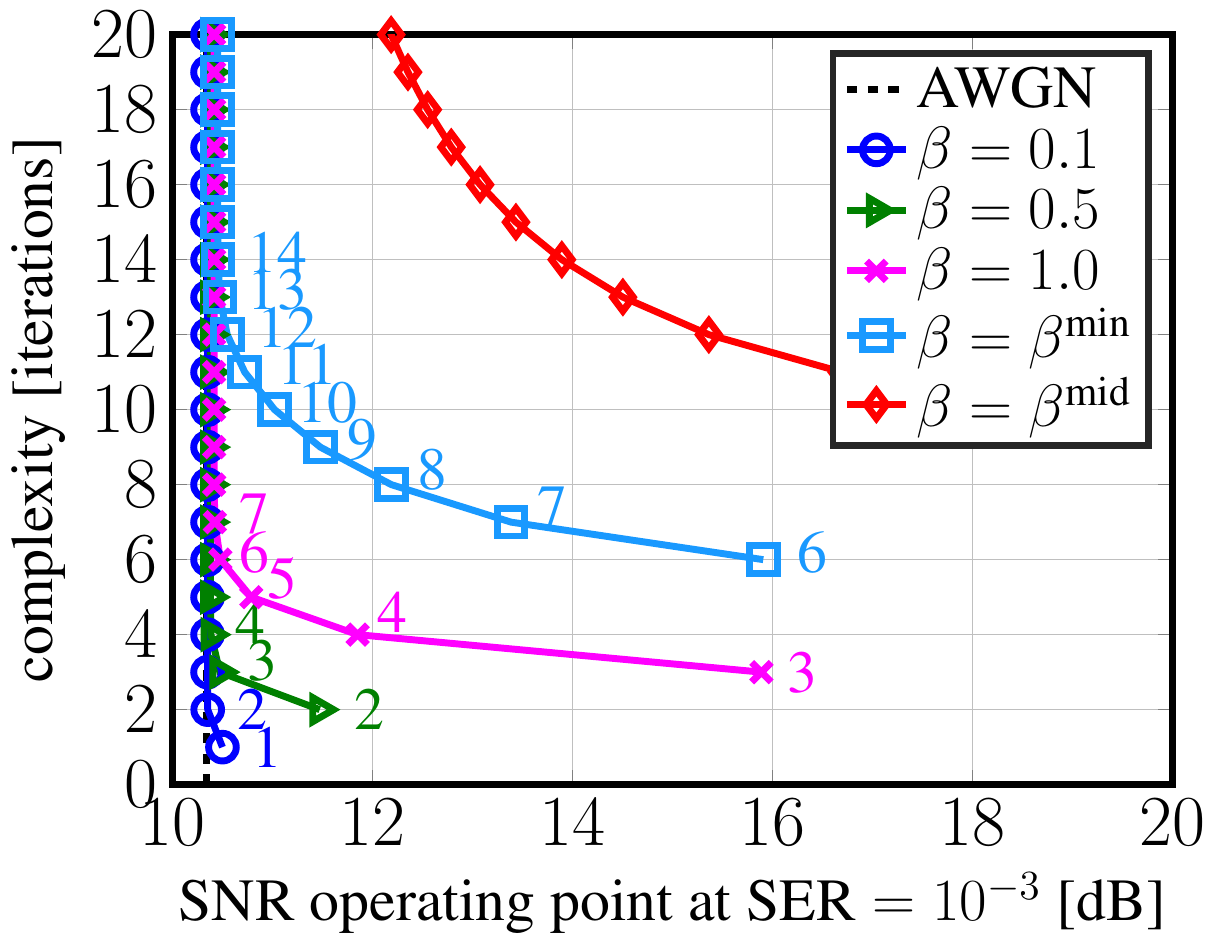}\label{fig:IO_iterative}}
\caption{
(a) Achievable rate (in bits per channel use per stream), (b) symbol-error rate (SER), and (c) performance/complexity trade-offs of \LAMA 
% andad JO-\LAMA (b) and (d) 
for $\beta\in\{0.1,0.5,\beta^\text{min}_\text{QPSK},\beta^\text{mid}_\text{QPSK},\beta^\text{max}_\text{QPSK}\}$,
where $\beta^\text{mid}_\text{QPSK}=(\beta^\text{min}_\text{QPSK}+\beta^\text{max}_\text{QPSK})/2$, and QPSK constellations. The sharp transitions in achievable rate and SER shown in (a) and (b) occur at $\beta>\beta^\text{min}_\text{QPSK}$ when the $\SNR=\beta\frac{E_s}{\No}$ with \No equaling the critical noise variance $\Nominnobeta(\beta)$.
(c) The dashed lines refer to the SNR operating point for an AWGN channel at $\textit{SER}=10^{-3}$ for each $\beta$.
The atomic complexity of \LAMA required to approach AWGN \SNR increases with the system ratio $\MT/\MR=\beta$.}
\label{fig:IOJO_plot}
\end{figure*}

\fref{fig:IO_cap} shows the achievable rate
 % (in bits per channel use per user)
% mutual information \cs{shouldn't this be bits per channel use per user? if so, pls also change in figure caption} 
of the decoupled AWGN channel per transmit stream for \LAMA, for various system ratios $\beta$. 
For small values of $\beta$, e.g. $\beta=0.1$, the achievable rate of \LAMA approaches to that of an AWGN channel, which agrees with \fref{lem:LAMAbeta0}.
We observe that the performance gap between \LAMA and that of an AWGN channel increases with $\beta$.
In particular, when $\beta=\beta^\text{min}_\text{QPSK}$, we see a sudden transition in the achievable rate of \LAMA
 % and JO-\LAMA 
to the achievable rate of an AWGN channel, which occurs approximately at $10$\,dB.
 % for IO and $12$\,dB for JO-\LAMA. 
%
This transition occurs exactly at the SNR regime for which the noise variance $\No$ becomes smaller than $\Nomin$, which was shown to ensure convergence of \LAMA to the unique optimal fixed point (cf. \fref{sec:LAMA_optimal}).
For $\beta\rightarrow\beta^\text{max}_\text{QPSK}$, we see that the achievable rate does not converge to that of an interference-free AWGN channel, irrespective of the SNR regime; this agrees with the 
% \cs{what was that result? can we add like 2 words to clarify what a recoverability result is?} 
perfect recoverability result in the large-system limit shown in \fref{thm:recovery} for ERT $\beta^\text{max}_\text{QPSK}$.

%% symbol error rate
\fref{fig:IO_SER} shows the symbol error rate (SER) of \LAMA.
Similar to the achievable rate in \fref{fig:IO_cap}, the SER for $\beta=0.1$ for \LAMA in a MIMO system approaches that of an interference-free AWGN channel. 
For $\beta>\beta^\text{min}$, we observe a waterfall behavior where the SER quickly drops and approaches that of an interference-free AWGN channel; this happens at exactly the point where the noise variance is smaller than the minimum critical noise $\Nominnobeta(\beta^\text{min}_\text{QPSK})$. 
We note that this waterfall behavior is consistent with the SNR regime that caused an upwards jump in the achievable rate curve shown in \fref{fig:IO_cap}.
When $\beta=\beta^\text{max}_\text{QPSK}$, we observe an SER floor at about $0.08$; this is due to the fact that as $\SNR\rightarrow\infty$, the cSE of \LAMA always converges to a suboptimal fixed point shown in \fref{fig:SE_QPSK3}.

\subsection{Performance/Complexity Trade-off}\label{sec:performance_tradeoff}

While only an infinite number of \LAMA iterations guarantee the cSE of \LAMA in \fref{thm:CSE} to converge to a fixed-point of \fref{eq:fixed_pt} and \fref{eq:fixed_pt_JO}, one can terminate the algorithm early with the goal of reducing its complexity. 
A straightforward approach is to terminate \fref{alg:LAMA}, if the parameter $\tau^t$ does not improve from one iteration to the next, e.g., if $\tau^t\leq \tau^{t+1}$ is met. 
Another approach is to terminate \LAMA after a predefined number of~$I$ iterations. The latter approach not only enables a deterministic throughput (which is critical in hardware implementations), but also enables us to study a fundamental performance/complexity trade-off of \LAMA.

Since the cSE analysis is only valid in the large system limit, common complexity measures, such as the number of additions and/or multiplications are not meaningful.
Nevertheless, for a given system, we see from \fref{alg:LAMA} that the computational workload of \LAMA per iteration remains constant. Hence, counting the maximum number of algorithm iterations provides a sensible way of measuring the complexity of \LAMA\footnote{In practice, one can multiply the atomic complexity with the number of arithmetic operations require per iteration; this enables one to obtain an accurate complexity measure that depends on the system configuration.}.

\begin{defi}
% [Atomic Complexity of \LAMA]
The \emph{atomic complexity} of \LAMA is defined by the maximum number of algorithm iterations $I$.
\end{defi}

We now study the performance of \LAMA depending on the atomic complexity~$I$. Put simply, we investigate by how much one can approach the performance of \LAMA with infinitely many iterations. 
We do so by first computing the output variance $\sigma^2_I$ of the equivalent AWGN channel for a fixed complexity $I$, and then computing the associated SER.

We first discuss the convergence speed of \LAMA to its fixed-point solution. 
% examine when \LAMA converges exponentially fast to the fixed-point solution.
% 
The following result, with proof in \fref{app:convergence_speed}, reveals that if $\beta<\betamin$, then \LAMA not only has a unique fixed point solution, but also converges exponentially fast; this ensures that \LAMA achieves near-IO performance with a small number of iterations.
\begin{lem}
% [Convergence speed of \LAMA is exponentially fast when $\beta<\betamin$ \cite{Maleki2010phd,DMM10b}]
\label{lem:convergence_speed}
Assume the initialization of \LAMA as in \fref{alg:LAMA}. If $\beta<\betamin$, then regardless of the noise variance $\No$, \LAMA converges exponentially fast to its unique fixed-point solution $\sigma^2_\star$.
\end{lem}

% \cj{stopped here}

% We now analyze the performance/complexity trade-offs \LAMA with QPSK for different system ratio $\beta$ in the large system limit. 
%
\fref{fig:IO_iterative} shows the required $\SNR$ to achieve SER of $10^{-3}$ for every iteration of \LAMA for various systems ratios $\beta$
% , $\beta\in\!\left\{0.1,0.5,\beta^\text{min}_\text{QPSK},\beta^\text{mid}_\text{QPSK},\beta^\text{max}_\text{QPSK}\right\}\!$ 
in the large-system limit. 
The colored dashed lines correspond to the SNR required to achieve an SER of $10^{-3}$ in an interference-free AWGN channel, which we call ``AWGN SNR.''
For $\beta=0.1$, $\beta=0.5$, only three and five iterations are required for \LAMA
% both IO and JO-\LAMA 
to closely approach the AWGN SNR.
We observe that as $\beta$ decreases, the number of iterations required to reach SNR of SER $10^{-3}$ also decreases.
This observation is in accordance with \fref{lem:LAMAbeta0}, where we demonstrated that in the extreme case where $\beta\rightarrow0$, one iteration (matched filter detection) is sufficient to converge to the AWGN SNR.
As $\beta$ increases, we start to see the performance differences between \LAMA and that of an interference-free AWGN channel. 
%
% When $\beta=1$, the differences are minimal as IO and JO-\LAMA needs only eight and nine iterations.
%
For $\beta=\betaminno_\text{QPSK}$, the SNR operating point of \LAMA closely approaches the AWGN SNR after $15$ iterations at a small performance loss (about $0.1$\,dB), which is visible from the SER plot in \fref{fig:IO_SER}.
%
% On the other hand, the \SNR operating point of JO-\LAMA converges to 12.95 dB after 19 iterations, which is approximately 2.5 dB higher than the AWGN \SNR from \fref{fig:JO_SER}.
%
The differences between the SNR operating point of \LAMA and AWGN SNR 
% JO-\LAMA 
are more pronounced when $\beta=(\beta^\text{min}_\text{QPSK}+\beta^\text{max}_\text{QPSK})/2$, as the SNR operating point of \LAMA converges to $13.5$\,dB after about 90 iterations, which is $0.6$\,dB higher than the AWGN $\SNR$ of $12.9$\,dB.
% , while that of JO-\LAMA converges to $17.55$\,dB after 25 iterations, roughly $4.5$\,dB greater than the AWGN \SNR. 
%
% Note that the \SNR operating point of IO-LAMA after 25 iterations is $14.6$\,dB (and decreasing with subsequent  iterations), which is still approximately 3 dB less than that of JO-\LAMA after the same number of iterations.
%
% We note that Tanaka also mentioned in \cite{TanakaCDMA} that the jointly optimal detector has worse performance for values of $\beta>1$ compared to that of the individually optimal detector.
%
For $\beta=\betamaxno_\text{QPSK}$, the complexity of \LAMA is not shown as it floors to an SER of approximately $0.08$ and hence, never achieves the target SER of $10^{-3}$.

\subsection{Performance in Finite-Dimensional Systems}\label{sec:estimator}

Since the design of \LAMA heavily relies on the large system limit, there are no optimality guarantees for finite-dimensional settings. 
For conventional, small-scale MIMO systems (with 8 antennas or less), the large-system assumption leads to a significant performance loss because (i) the statistics of $\bmz^t=\shat^t+\bH^\Herm\resid^t$ are not Gaussian and hence, (ii) the correct statistics of the Gaussian term $\bmz^t$ cannot be tracked in the \LAMA algorithm. 
The problem that arises in finite-dimensional systems becomes evident if we keep $\beta=1$ and increase $\SNR\rightarrow\infty$ for a small system. We see that \LAMA exhibits in an SER floor (see \fref{fig:sim2} for a $128\times128$ 16-QAM system). We note that this SER floor lowers as the system's dimension increases.
The performance loss of AMP-based algorithms for small-sized systems has been investigated in \cite{KT15,KT2016}.

In order to mitigate \LAMA's performance loss in finite dimensional systems, one can use estimators 
% for the message variance function $\mathsf{G}$ 
as opposed to the original message variance function in \fref{eq:G_compute} to estimate $\sigma^2_t$ each iteration.
For estimators in \LAMA to work universally when the antenna configurations are \emph{both} small and big, we need estimators of $\sigma^2_{t+1}$, which we will denote as $\hat\sigma^2_{t+1}$, that not only lower the error floor at high SNR in small antenna systems, but also converges to the true effective noise variance $\sigma^2_{t+1}$ in large antenna systems.

In \cite{andreaGMCS}, a series of estimators have been proposed for AMP in the context of sparse recovery. 
We adopt the same approach for LAMA for the case $\Nopost=\No$, where instead of computing the average of the exact message variance function as \fref{eq:input_var_eq}, 
% \cs{isn't there an equation where we showed this expression?} $\hat\sigma^2_{t+1}=\No+\beta\langle\mathsf{G}(\bmz^t,\No(1+\tau^t))\rangle$, 
% 
we estimate the variance of the Gaussian estimate $\bmz^t=\shat^t + \bH^\Herm\resid^t$ by:
 % $\hat\sigma^2_{t+1} = \vecnorm{\resid^t}_2^2/{\MR}$.
% 
% While various candidates were suggested for $\mathsf{est}(\shat^t + \bH^\Herm\resid^t)$ in \cite{andreaGMCS}, we consider the estimator $\vecnorm{\resid^t}_2^2/{\MR}$ 
% scaled by the noise variance that is given by
% 
\begin{align}\label{eq:estimator_tau}
\hat\sigma^2_{t+1} = \frac{1}{\MR}\vecnorm{\resid^t}_2^2.
\end{align}

% output_noise_var_eq
 
\fref{fig:sim2} shows the performance of \fref{eq:estimator_tau} for \LAMA in an $128\times128$ system with 16-QAM. 
We observe a decrease in the SER floor in high SNR regime compared to the original \LAMA without the estimator with no performance loss in the low SNR regime.
% 
%We note that the SER shown in \fref{fig:sim2} is for an uncoded system; the SER floor is likely below what is needed for reliable communication in a coded system.
% 

\subsection{Extension to General Channel Matrices $\bH$}\label{sec:gamp}

It is important to note that one of the limiting assumptions underlying AMP (and hence, for \LAMA) is that the entries of the channel matrix $\bH$ are i.i.d. zero-mean Gaussian or complex Gaussian with variance $1/\MR$ for AMP and complex-valued AMP respectively. 
In practical systems, however, the BS antennas may exhibit correlation and uneven power profiles, especially in multi-user scenarios, which makes LAMA less robust in these scenarios.
% 
% This requirement is rather unrealistic as practical systems may exhibit correlation and uneven power profiles (especially in multi-user scenarios). 
%
To address these limitations, Rangan \cite{GAMP2011} has developed Generalized AMP (GAMP), which extends AMP to arbitrary input and output noise distributions for real-valued systems, and can operate in channels with different power profiles. 
We note that in the large system limit with $\bH$ distributed according to (A2) with Gaussian noise, GAMP and AMP are equivalent. 
In addition, a modified GAMP that uses damping technique was proposed in \cite{RSF2014} to cope with non-zero mean, low-rank channels. The damping technique slows certain algorithmic parameter updates, but does so at the cost of increased iterations of the algorithm. 
Vila and Schniter furthermore included expectation-maximization into GAMP in \cite{VS2013,NSE14}, which further improves the performance of AMP-based methods in finite-dimensional systems. 
Recently, reference \cite{RSF2017} introduced vector AMP, which further generalizes GAMP to arbitrary matrices.
%
%Although extensions of the LAMA and (AMP-based) algorithms can help the algorithm operate on practical, finite-dimensional non-i.i.d. channels, analysis is hard, and often intractable.
% 
% \cs{we should also cite the new variant Vector AMP for completeness from Schniter and others. }

Generalized AMP has been used for practical MIMO-OFDM systems 
% We note that the GAMP algorithm for MIMO-OFDM systems has been shown in 
\cite{WKNLHG14} with variations introduced in \cite{NSE14,VS2013,S2011} to increase the detection performance for a finite-dimensional system. 
Reference \cite{WKNLHG14} primarily focused on simulations, whereas our paper concentrates on theoretical analysis in the large-system limit via the state-evolution framework.
 % our contributions are

% the theoretical analysis of the robustness of these methods in practical systems is still left for future work.

\subsection{Simulation Results}

Figures \ref{fig:sim1} and \ref{fig:sim2} show simulation results for large MIMO systems with 16-QAM. We fix the number of BS antennas to $128$ and the number of user antennas to $64$ and $128$.
%\tabularnewline
We compare the performance of \LAMA to unbiased linear MMSE detection, another message-passing-based receiver, i.e., channel hardening-exploiting message passing (CHEMP) \cite{indiachemp}\footnote{
We note that CHEMP has no theoretical performance guarantees and was primarily developed for massive MIMO, i.e., $\MR\gg\MT$ or small $\beta$.  
}, and IO data detection bound obtained by the cSE in the large-system limit. 
%

% \cs{the figures should not have ``ind. optimal'' this looks super bad. just write IO analytical or something}

For the $128\times64$ system in \fref{fig:sim1}, \LAMA performs very close to the IO bound with only 8 iterations. 
We note that CHEMP \cite{indiachemp} with 8 iterations performs worse than the linear MMSE detection, but approaches the performance of LAMA at 15 iterations.
Note that in the large-system limit for a system-ratio of $\beta=64/128$, $\beta<\beta_\text{16QAM}^\text{min}$, so \LAMA achieves IO data detection performance for any noise variance $\No$. 

% 128x128
For the $128\times128$ system in \fref{fig:sim2}, \LAMA with the estimator in \fref{eq:estimator_tau} exhibits a floor at around  $10^{-2}$ SER.
% about after \SNR of $15$\,dB.
%
We note that \LAMA with the estimator reduces the error floor while maintaining the performance at low SNR.
Because of the flooring behavior of \LAMA in finite dimensions, it performs worse than linear MMSE at high SNR (above $35$\,dB for this case). 
\LAMA with $20$ iterations outperforms CHEMP at the same number of iterations; CHEMP floors at an SER of $10^{-1}$ even after $100$ iterations.
In the $128\times128$ setting, we note that $\beta=1$ is larger than the ERT, $\beta_\text{16QAM}^\text{min}\approx0.9830$, from \fref{tbl:exact_recovery}, so \LAMA has two regions of optimality (cf.~\fref{tbl:IOLAMAoptimal_reg} for the regions) with $\No^\text{min}(\beta)\approx0.03$, or \SNR around 15\,dB. Note that this SNR happens where the sharp ``waterfall'' appears in the IO bound in \fref{fig:sim2}. 
We stress that for $\beta=1$ and $\MT\rightarrow\infty$, the SER of \LAMA will converge to that of the IO bound by cSE.
\begin{figure}[tp]
\centering
%
% \hspace*{0.4cm}
\subfigure[]{\includegraphics[width=0.45\textwidth]{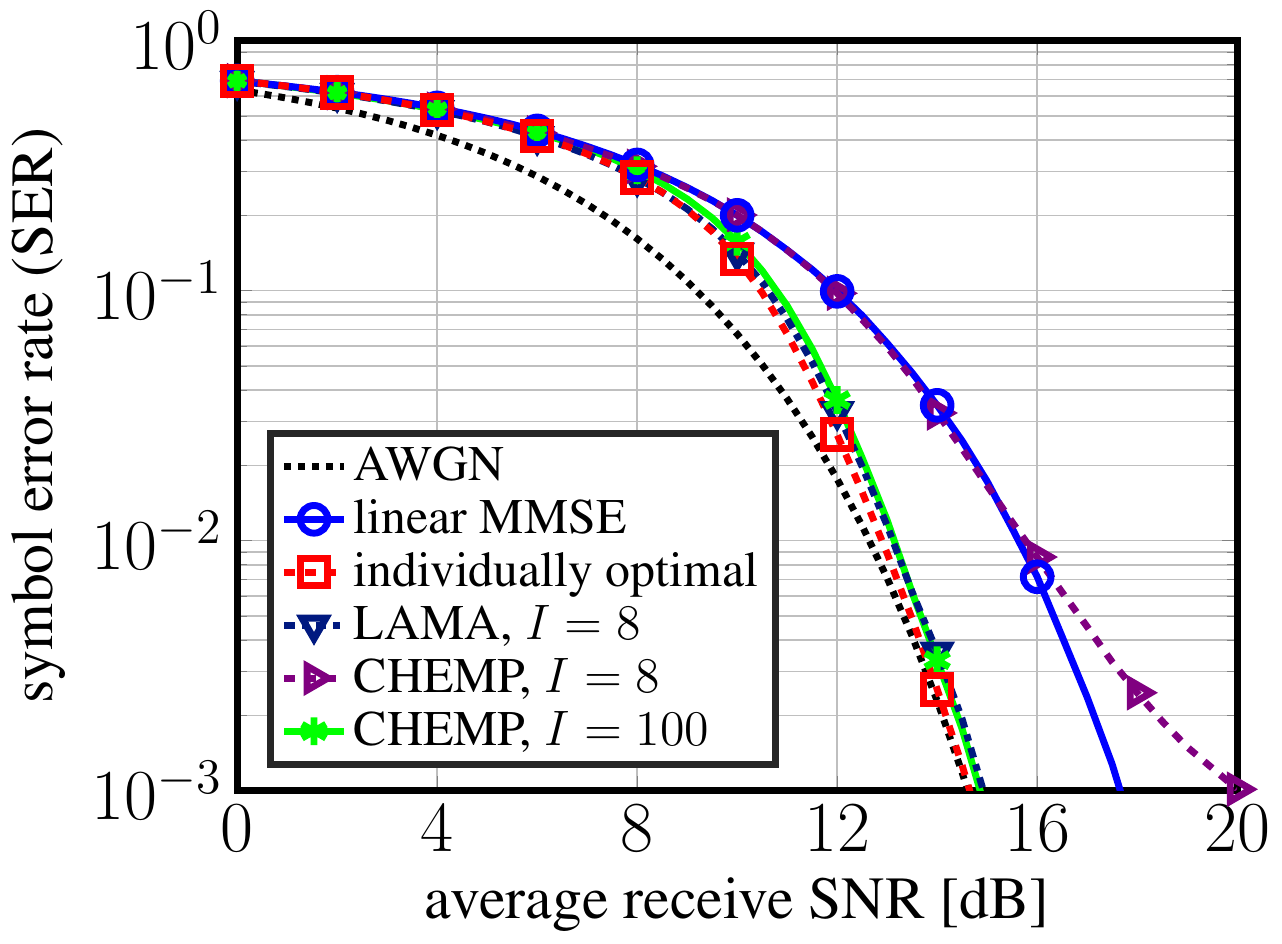}\label{fig:sim1}}
\hspace{0.4cm}
\subfigure[]{\includegraphics[width=0.45\textwidth]{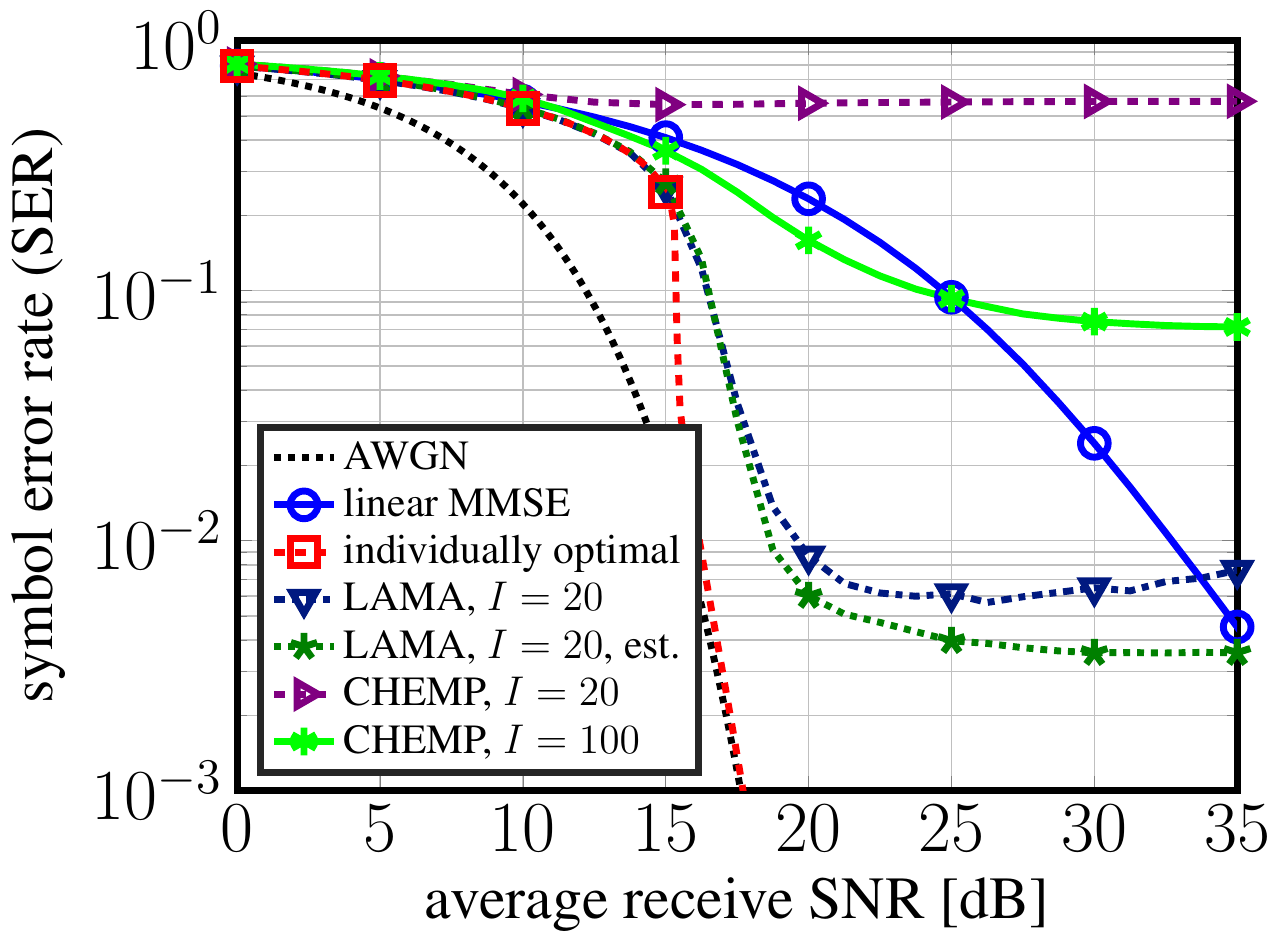}\label{fig:sim2}}
\caption{Symbol error rate (SER) performance of \LAMA for finite-dimensional systems with 16-QAM. (a) SER for a $128\times64$ system, compared to linear MMSE detection and CHEMP \cite{indiachemp}.
(b) SER for a $128\times128$ system, with the estimator in \fref{eq:estimator_tau} for \LAMA to mitigate the performance loss occurring from finite dimensions.
}
\label{fig:sims}
\end{figure}

\section{\LAMA and Prior Art}\label{sec:prior_art}

We now review existing results that are relevant for \LAMA and our analysis in \fref{sec:LAMA_optimal}.

\subsection{BPSK signaling in Randomly Spread CDMA systems}
We show that \LAMA for BPSK constellation in randomly spread CDMA systems coincides exactly to the detection algorithm put forth in \cite{K2003} and the cSE of \LAMA without noise variance mismatch is equivalent to that given by IO data detection bound derived from the replica method by Tanaka in \cite{TanakaCDMA}.
%
% Note that while the Tanaka characterized the individually optimal detection in \cite{TanakaCDMA}, no practical algorithm to achieve such detection has not been detailed.
%
Consider a real-valued randomly-spread CDMA system with equally likely BPSK symbols $\setO=\{-1,+1\}$ and the entries of the channel matrix $\bH$ are distributed $\setN(0,1/\MR)$.
In this case, \fref{eq:DAMP_F} and \fref{eq:DAMPrealG} are given by
\begin{align*}
\mathsf{F}(\shate_\ell,\tau) = \tanh \!\left(\frac{1}{\tau}\shate_\ell\right),\quad
\mathsf{G}(\shate_\ell,\tau) = 1 - \tanh^2 \!\left(\frac{1}{\tau}\shate_\ell\right).
\end{align*}
%which is equivalent to the nonlinear MMSE (IO) detector derived in \cite{TanakaCDMA}.
and thus, \LAMA corresponds to the following recursion:
\begin{align*}
\shat^{t+1} &= \tanh \!\left(\frac{\shat^t + \bH^T\resid^t}{\No(1 + \tau^t)}\right)\\
\tau^{t+1} &= \frac{\beta}{\No}\!\left\langle 1 - \tanh^2 \!\left(\frac{\shat^t + \bH^T\resid^t}{\No(1+\tau^t)}\right)\right\rangle\\
\resid^{t+1} &= \bmy - \bH\shat^{t+1} + \frac{\tau^{t+1}}{1+\tau^t}\resid^t,
\end{align*}
with the SE recursion from \fref{thm:CSE} given by
\begin{align}\label{eq:BPSKse}
\sigma_{t+1}^2 &= \No + \beta\Exop_{\srv,Z}\!\left[\!\left(\tanh\!\left(\frac{\srv+\sigma_t Z}{\sigma_t^2}\right)-\srv\right)^2\right],
\end{align}
where for a fixed $\beta$ and $\No$, the fixed point equation is 
\begin{align}\label{eq:BPSK_fixed}
\sigma^2  
% &=\No+ \beta\int_\reals\!\left[\tanh\!\left(\frac{1+\sigma z}{\sigma^2}\right)-1\right]^2 \frac{1}{\sqrt{2\pi}}\exp\!\left(-\frac{z^2}{2}\right) \text{d}z \notag\\
&= \No+\beta\int_\reals\!\left[1 - \tanh\!\left(\frac{1+\sigma z}{\sigma^2}\right)\right]\frac{1}{\sqrt{2\pi}}\exp\!\left(-\frac{z^2}{2}\right) \text{d}z.
\end{align}

We note that the fixed point equation in \fref{eq:BPSK_fixed} coincides exactly to Tanaka's fixed point equation in \cite{TanakaCDMA} and the optimal multiuser efficiency in \cite{GV2005} derived using the replica method.
Moreover, \LAMA coincides exactly to the method developed by Kabashima in 2003 for randomly-spread CDMA with BPSK signaling \cite{K2003}. 
Kabashima showed that the algorithm is consistent with the state evolution predictions obtained through numerical simulations.
Kabashima's algorithm in~\cite{K2003} was given for BPSK in real-valued systems only; in contrast, \LAMA is suitable for general constellations and complex-valued systems, and can be analyzed in the large-system limit.

In \fref{sec:ERTMRT}, we noted that the ERT of a BPSK system for LAMA is computed to be approximately 2.0855, which coincides exactly with Tanaka's recovery threshold in \cite{TanakaCDMA}, which was computed using the replica method. 
While we characterized the state of having multiple fixed point solutions by our definitions of MRT and ERT, Tanaka analogized the state of having multiple fixed points as having coexistence of phases in physical systems.
In this context, the MRT $\betaminno_\text{BPSK}$ and ERT $\betamaxno_\text{BPSK}$ corresponds to the boundary in which the instability of the phrases occur, and the boundary where the replica-symmetry solution becomes unstable, breaking the replica-symmetry assumptions \cite{TanakaCDMA}.
Although an analytical expression of the ERT has been given in  \cite{TanakaCDMA}, an exact characterization of the MRT was not included. 
Note that our \LAMA results generalize Tanaka's results to arbitrary constellations and provide a practical algorithm.

\subsection{Recovery of Antipodal Solutions via Convex Optimization}

Recall that from \fref{tbl:exact_recovery}, that a system ratio $\beta$ smaller than $2.0855$ is able to perfectly recover a BPSK vector in absence of noise. In this scenario, \LAMA is able to determine the unique solution to $\bmy=\bH\bms_0$ with $\bms_0\in\{-1,+1\}^\MT$ if $\bH$ is distributed (A2), and $\beta=\MT/\MR$ is fixed with $\MT\rightarrow\infty$.
A similar scenario was studied in \cite{DT2011,MR2011}, where the authors have provided necessary and sufficient conditions for the recovery of antipodal solutions from $\bmy=\bH\bms_0$.
In \cite{DT2011}, Donoho and Tanner showed that in the large system limit, $\beta<2$ guarantees the recoverability of the unique signal $\bms_0\in\{-1,+1\}^\MT$. 
The same threshold was recovered in \cite{MR2011}, by solving the following convex optimization problem \cite{SGYB14}:
\begin{align*}
(\text{P}_\infty) \quad \underset{\tilde\vecs\in\reals^\MT}{\text{minimize}}\, \|\tilde\vecs\|_\infty \quad \text{subject to} \,\,\vecy=\bH\tilde\vecs.
\end{align*}
In particular, the solution $\hat\vecs$ to $(\text{P}_\infty)$ corresponds to the antipodal vector $\{-\alpha,+\alpha\}$ for a given $\alpha>0$ if $\beta<2$ with high probability \cite{MR2011}. 
It is interesting to see that $(\text{P}_\infty)$ does not exploit magnitude information (i.e. $\alpha=1$), whereas LAMA requires this information. Quite surprisingly, the lack of this prior information only results in a slight improvement in terms of the system ratio $\beta$ that enables perfect recovery from $2$~to~$2.0855$. 
The error-rate performance of $\hat\vecs$ to $(\text{P}_\infty)$ for was recently investigated in \cite{TAXH2015,TAH2018,TXH2018}; \LAMA-based results were studied in \cite{JMS2016}.

\section{Conclusions}\label{sec:conclusion}

In this paper, we have developed the complex Bayesian approximate message passing (cB-AMP) framework with a possible mismatch in the postulated noise variance; 
%
% We have also presented the corresponding state evolution framework, which shows that the equivalent noise variance of cB-AMP at every iteration can be exactly computed by solving the state evolution recursion equations. 
% predicted by the variance of a single random variable mixed with Gaussian noise.
%
cB-AMP with appropriate priors enables a derivation of the \LAMA data detector.
In the large-system limit, we have shown that \LAMA decouples large MIMO systems into parallel AWGN channels with identical noise variance across all transmit streams every iteration.
Furthermore, cSE has been used to analyze the exact noise variance of the decoupled AWGN channel.

We have derived the specific conditions for which \LAMA achieves IO performance.
Based on the system ratio $\beta$, there exist three optimality regimes for \LAMA, where $\beta\leq\betaminno_\setO$, $\beta\in(\betaminno_\setO,\betamaxno_\setO)$, and $\beta\geq\betamaxno_\setO$ where the MRT $\betaminno_\setO$ and ERT $\betamaxno_\setO$ can be computed numerically for \LAMA. 
We have shown both asymptotic and finite-dimensional performance of \LAMA through analytical predictions and numerical simulations, which confirm our theoretical results. In addition, we have characterized the convergence behavior of \LAMA for system ratios smaller than the ERT.
For small system ratios $\beta$, we have shown that \LAMA exhibits similar achievable rate and error-rate performance to that of an AWGN channel for a low number of iterations, which makes \LAMA an excellent candidate for data detection in massive MIMO systems.

% \cs{that's not a good conclusion sentence; that sounds more like a contributions sentence; fix} We also present the achievable rates and error rate performance of \LAMA for various system ratios for QPSK and present simulation results for 16-QAM.

% We then analyze the achievable rates and error rates of \LAMA under various system ratios $\beta$ for QPSK constellation set. 
%
% When $\beta$ is greater than MRT $\betaminno$, we observe that there exists a ``jump'' and ``waterfall'' effect in the achievable rate and error rate respectively, where the corresponding performance of \LAMA suddenly jumps or falls to that of an AWGN channel.
%
% These effects happen when the noise variance is equal to the minimum critical noise $\Nomin$. 
%
% As $\beta$ increases to ERT $\betamaxno$, $\Nomin\rightarrow0$ and  \LAMA fails to converge to the achievable rate nor SER of an AWGN channel. 
%

% We then analyzed the iterative performance of \LAMA by the cSE framework and have shown that under $\beta<1$, in general, only a few iterations are required for \LAMA for its cSE to converge to the fixed point solution; in particular, as $\beta\rightarrow0$, only a single iteration is required. 
%
% We conclude by discussing the practical aspects of \LAMA and present an estimator that helps mitigate the performance loss that occurs in finite-dimensional systems.

There are numerous avenues for further work. 
A performance analysis of \LAMA with noise variance mismatch and a theoretical study in the finite dimensional setting as in \cite{RV2016} are open problems.
Moreover, theoretical analysis of LAMA and (AMP-based) methods for non-i.i.d. channel matrices could lead to a more robust algorithms that can operate more effectively on practical, real-world channel matrices.
% a complexity-performance tradeoffs for different algorithms. 
% 
Investigating the performance of LAMA in the presence of mismatched priors, as in \cite{JMS2016,MBB2015}, may lead to hardware-friendly data detection algorithms. 
A hardware implementation of LAMA would demonstrate the real-world efficacy of our algorithm.

\section*{Acknowledgments}
The work of C. Jeon, R. Ghods, and C. Studer was supported in part by Xilinx Inc., and by the US National Science Foundation (NSF) under grants ECCS-1408006, CCF-1535897,  CAREER CCF-1652065, CNS-1717559, and EECS-1824379. 
%
% \cj{Arian, do you have any funding agency that you would like to recognize?}

%% file: secc-app.tex
% !TEX root = 18TIT-lama.tex
\section{\cite[Lem. 5.56]{Maleki2010phd}}\label{app:lemma5_5}

For completeness, we include \cite[Lem. 5.56]{Maleki2010phd} and its proof for complex-valued MIMO systems.
We will use 
\fref{lem:Message} to derive the cB-AMP algorithm in \fref{app:damp-complex}.
\begin{lem}
% [Mean and variance of message $\nu_{\ell\rightarrow k}^t$]
\label{lem:Message}
Let $\shate_{\ell\rightarrow k}^t$ and $\Nopost\tau_{\ell\rightarrow k}^t$ be the mean and variance of the distribution $\nu_{\ell\rightarrow k}^t$ in \fref{eq:sumprod}, respectively. Suppose at iteration $t$, the messages are set to $\est\nu_{k\rightarrow \ell}^t\!\left(s_\ell\right) = \est\phi_{k\rightarrow \ell}^t\!\left(s_\ell\right)$, where $\est\phi_{k\rightarrow \ell}^t(s_\ell) $ is defined by
\begin{align}\label{eq:z1}
\est\phi_{k\rightarrow \ell}^t (s_\ell) &\triangleq 
\frac{\!\left\vert H_{k, \ell}\right\vert^2}{\pi \Nopost(1+\tau_{k \rightarrow \ell}^t)} \exp\left(
- \frac{\left\vert H_{k,\ell} s_\ell - \reside_{k \rightarrow \ell}^t \right\vert^2\!}{\Nopost(1+\tau_{k \rightarrow \ell}^t)}
\right)\!,
% \est\phi_{k\rightarrow \ell}^t (s_\ell) &\triangleq \frac{\left\vert H_{k, \ell}\right\vert^2}{\pi \Nopost(1+\tau_{k \rightarrow \ell}^t)} \exp\left(- \frac{1}{\Nopost(1+\tau_{k \rightarrow \ell}^t)} \left\vert H_{k,\ell} s_\ell - \reside_{k \rightarrow \ell}^t \right\vert^2\right)p\!\left(s_\ell \right),
\end{align}
and $\tau_{k\rightarrow \ell}^t = \tau^t$, where the residual and variance terms are given by
\begin{align*}
\reside_{k \rightarrow \ell}^t \triangleq y_k - \sum_{b \neq \ell} H_{k,b}\shate_{b \rightarrow k},\quad
 \tau_{k\rightarrow \ell}^t \triangleq \sum_{b \neq \ell} \abs{H_{k,b}}^2 \tau_{b\rightarrow k}^t.
 \end{align*}
 Then, at the next iteration $t+1$, the mean and the variance of the message $\nu_{\ell\rightarrow k}^{t+1}$ are given by
 \begin{align*}
 \shate_{\ell\rightarrow k}^{t+1} &= \mathsf{F}\!\left(\sum_{a\neq k} H_{a,\ell}^* \reside_{a\rightarrow \ell}^t,\Nopost(1+\tau^t)\right),\\
 \tau_{\ell\rightarrow k}^{t+1} &= \frac{1}{\Nopost} \mathsf{G}\!\left(\sum_{a\neq k} H_{a,\ell}^* \reside_{a\rightarrow \ell}^t,\Nopost(1+\tau^t)\right),
 \end{align*}
 \end{lem}
\begin{proof}
Suppose at iteration $t$, the messages from factor nodes to the variable nodes are set to be $\est\nu_{k\rightarrow \ell}^t = \est\phi_{k\rightarrow \ell}^t$. Then, 
\begin{align*}
\nu_{\ell\rightarrow k}^{t+1} &= \prod_{a\neq k}\est\phi_{a\rightarrow \ell}^t p\!\left(s_\ell \right) \\&= \exp \left(- \frac{
\sum_{a\neq k} \abs{H_{a,\ell}s_\ell-\reside_{a\rightarrow \ell}^t}^2}{\Nopost(1+\tau_{k\rightarrow \ell}^t)}
% \sum_{a\neq k} \abs{H_{a,\ell}s_\ell-\reside_{a\rightarrow \ell}^t}^2
\right)\!
p\!\left(s_\ell \right)
\\
&=
\frac{1}{Z} 
\exp\!
\left(-
  \frac{
    % \frac{\MR-1}{\MR}
    \abs{s_\ell}^2 - 
    2\sum_{a\neq k}
    \realpart{s_\ell^*H_{a,\ell}^*\reside_{a\rightarrow \ell}^t}}{\Nopost(1+\tau_{k\rightarrow \ell}^t)}
  %   \bigg(
  %   % \frac{\MR-1}{\MR}
  %   \abs{s_\ell}^2 - 
  %   2\sum_{a\neq k}
  %   \realpart{s_\ell^*H_{a,\ell}^*\reside_{a\rightarrow \ell}^t}
  %    % \left(
  %    %  s_\ell^*H_{a,\ell}^*\reside_{a\rightarrow \ell}^t+s_\ell H_{a,\ell}\reside_{a\rightarrow \ell}^{*t}
  %    % \right)
  % \bigg)\!
\right)\!
\\
&\quad \cdot
\exp
\bigg(
  \frac{\abs{s_\ell}^2}{\MR\Nopost(1+\tau_{k\rightarrow \ell}^t)}
\bigg)
p\!\left(s_\ell\right)
\\
&= \phi_{\ell\rightarrow k}^{t+1} (s_\ell) \left\{ 1+ O(\abs{s_\ell}^2/\MR)\right\}\!,
\end{align*}
where $Z$ is a normalization constant that ensures $\nu_{\ell\rightarrow k}^{t+1}$ is a probability density function.
Here, we defined $\phi_{\ell\rightarrow k}^{t+1}$ as
\begin{align*}
\phi_{\ell\rightarrow k}^{t+1}(s_\ell)=f\Bigg(s_\ell \,\bigg|\, \sum_{a\neq k} H_{a,\ell}^* \reside_{a\rightarrow \ell}^t, \Nopost(1+\tau^t)\Bigg),
\end{align*}
where $f\!\left(s_\ell \vert \hat{s}_\ell,\tau\right)$ is defined as the probability distribution in \fref{eq:gaussianmeasure} and $\tau_{k\rightarrow \ell}^t = \tau^t$. 
By definition, the mean $\mathsf{F}$ and variance $\mathsf{G}$ of $\nu_{\ell\rightarrow k}^{t+1}$ is given as mean and variance of the conditional probability distribution defined in \fref{eq:gaussianmeasure}:
\begin{align*}
\shate_{\ell\rightarrow k}^{t+1} &= \mathsf{F}\left(\sum_{a\neq k} H_{a,\ell}^* \reside_{a\rightarrow \ell}^t,\Nopost(1+\tau^t)\right),\\
\tau_{\ell\rightarrow k}^{t+1} &= \frac{1}{\Nopost} \mathsf{G}\left(\sum_{a\neq k} H_{a,\ell}^* \reside_{a\rightarrow \ell}^t,\Nopost(1+\tau^t)\right).
\end{align*}
\end{proof}

\section{Derivation of \fref{alg:cB-AMP}}\label{app:damp-complex}

We start by considering a factor graph $G=(V,F,E)$ with variable nodes $V=\{1,\ldots,\MT\}$, factor nodes $F=\{1,\ldots,\MR\}$, and edges $E = V \times F = \left\{ (\ell,k): \ell\in V, k\in F \right\}$. The sum-product message equations for \eqref{eq:factoredMAP} at every iteration $t$ are given by \cite{MM2009book},
\begin{align}
\label{eq:sumprod}
\nu_{\ell\rightarrow k}^t\!\left(s_\ell\right) 
&= 
\prod_{a\neq k} \est\nu_{a\rightarrow \ell}^{t-1}
\!\left(s_\ell\right) p\!\left(s_\ell\right)\!,\\
\label{eq:sumprod2}
\est\nu_{k\rightarrow \ell}^t\!\left(s_\ell\right) &= \int_\complexset p\!\left(y_k\,|\,\bms,\bmh_k^\row\right)
\prod_{b\neq \ell} \nu_{b\rightarrow k}^t \!\left(s_b\right)\text{d}y_k,
\end{align}
where $\nu_{\ell\rightarrow k}^t(s_\ell)$ and $\est\nu_{k\rightarrow \ell}^t(s_\ell)$ are probability density functions.
% 
% With \fref{lem:Message}, the mean and the variance of the sum-product messages in \fref{eq:sumprod} and \fref{eq:sumprod2} can be efficiently computed under the assumption of the Gaussian distribution in \fref{eq:gaussianmeasure}. However, the simplified update equations still need the computation of $\MR\MT$ messages, which would entail high computational complexity. 
% 
% 

Now, with \fref{lem:Message} we can simplify the sum-product algorithm shown in \fref{eq:sumprod} and \fref{eq:sumprod2}. 
We first expand the messages $\shate_{\ell\rightarrow k}^{t+1}$ and $\reside_{k\to\ell}^{t+1}$ into two parts (i) constant messages $\shate_{\ell}^{t+1}$ and $\reside_{k}^{t+1}$ which are independent of the edge $(\ell,k)$ and (ii) perturbed messages $\Delta \shate_{\ell\rightarrow k}^{t+1}, \Delta \reside_{k\rightarrow \ell}^{t+1}$ that depend on the edge.
As done in \cite[Eq. 5]{malekiCAMP}, we assume $\Delta \shate_{\ell\rightarrow k}^{t+1}, \Delta \reside_{k\rightarrow \ell}^{t+1} = O(1/\sqrt{\MT})$  such that 
\begin{align}\label{eq:z2}
\shate_{\ell\rightarrow k}^{t+1} &\triangleq \shate_\ell^{t+1} + \Delta \shate_{\ell\rightarrow k}^{t+1} + O(1/{\MT}),\\\label{eq:z3}
\reside_{k\rightarrow \ell}^{t+1} &\triangleq \reside_k^{t+1} + \Delta \reside_{k\rightarrow \ell}^{t+1} + O(1/{\MT}).
\end{align}
We then replace the complex-valued soft-thresholding function $\eta(\cdot)$ by the conditional mean $\mathsf{F}(\cdot)$ as in \cite[Prop. II.1]{malekiCAMP}, and use the decomposition in \fref{eq:z2} and \fref{eq:z3} to obtain
% 
 % we take the first-order complex expansion of \fref{lem:Message}, which reduces to the following set of equations.
% % \fref{lem:propcamp} will show that the decomposition of the messages in \fref{eq:z2} and \fref{eq:z3} will enable simpler expressions for $x_\ell^{t+1}$ and $\reside_k^{t+1}$.
% \begin{lem}[First-order complex expansion]\label{lem:propcamp}
% Suppose that \fref{eq:z2} and \fref{eq:z3} holds in every iteration $t$ of the message passing algorithm specified in \fref{lem:Message} \cs{do we need to show somewhere that these equations actually hold??? I think this is a missing piece.}. Moreover, assume that $\mathsf{F}$ is continuous \cs{how important is this?} \cj{F has to be continuous for the derivative to exist!} \cs{need to ask arian!}. Then, the mean $x_\ell^{t+1}$ and residual $\reside_k^{t+1}$ terms satisfy the following equations:
\begin{align}
\shate_\ell^{t+1} &= \mathsf{F}\!\left(\shate_\ell^t + \sum_{a=1}^{\MR}H_{a,\ell}^*\reside_a^t,\Nopost(1+\tau^t)\right)\label{eq:cBAMP-mean}\\
\reside_k^{t+1} &= y_k - \sum_{b=1}^{\MT}H_{k,b}\shate_b^{t+1} \notag\\
&\quad + \sum_{b=1}^{\MT} H_{k,b} \!\left(\partial_1\mathsf{F}^\text{R}\!\left(\shate_b^t + \sum_{a=1}^{\MR}H_{a,b}^*\reside_a^t\right)\!\right)\!\realpart{H_{k,b}^* \reside_k^t}\notag\\
&\quad + \sum_{b=1}^{\MT} H_{k,b} \!\left(\partial_2\mathsf{F}^\text{R}\!\left(\shate_b^t + \sum_{a=1}^{\MR}H_{a,b}^*\reside_a^t\right)\!\right)\!\imagpart{H_{k,b}^* \reside_k^t}\notag\\
&\quad + i\sum_{b=1}^{\MT} H_{k,b} \!\left(\partial_1\mathsf{F}^\text{I}\!\left(\shate_b^t + \sum_{a=1}^{\MR}H_{a,b}^* \reside_a^t\right)\!\right)\!\realpart{H_{k,b}^* \reside_k^t}\notag\\
&\quad + i\sum_{b=1}^{\MT} H_{k,b} \!\left(\partial_2\mathsf{F}^\text{I}\!\left(\shate_b^t + \sum_{a=1}^{\MR}H_{a,b}^*\reside_a^t\right)\!\right)\!\imagpart{H_{k,b}^* \reside_k^t}\label{eq:cBAMP-residual},
\end{align}
with
\begin{align*}
\partial _1\mathsf{F}^\text{R} &\triangleq \frac{\partial \realpart{\mathsf{F}(x + iy,\tau)}}{\partial x},\quad \partial _2\mathsf{F}^\text{R} \triangleq \frac{\partial \realpart{\mathsf{F}(x + iy,\tau)}}{\partial y},\\
\partial _1\mathsf{F}^\text{I} &\triangleq \frac{\partial \imagpart{\mathsf{F}(x + iy,\tau)}}{\partial x},\quad \partial _2\mathsf{F}^\text{I} \triangleq \frac{\partial \imagpart{\mathsf{F}(x + iy,\tau)}}{\partial y}.
\end{align*}

The final step to arrive at the cB-AMP algorithm is to compute the message-variance update equation and simplifying \fref{eq:cBAMP-residual} by the fact that $\bH$ satisfies (A1).
% , i.e., zero-mean with unit $\ell_2$-norm with independent real and imaginary entries.
% %
We note that the message-variance update equation was not provided in \cite{malekiCAMP} and hence, we include it for completeness. 
The variance update equation is computed by
\begin{align}\notag
\tau^{t+1} &= \sum_{b=1}^{\MT}\abs{H_{k,b}}^2\tau_{b\rightarrow k}^{t+1}\\
&= 
 \sum_{b=1}^{\MT} \frac{1}{\Nopost \MR} \mathsf{G}\!\left(\shate_b^t+\sum_{a=1}^{\MR} H_{a,b}^* \reside_a^t,\Nopost(1+\tau^t)\right)\!.\label{eq:tau_update_eq}
\end{align}
Since the columns of $\bH$ have unit norm with pairwise independence by (A1), each term in \fref{eq:cBAMP-residual} can be simplified in the large system limit as follows:
\begin{align*}
\sum_{b=1}^{\MT} H_{k,b} \partial_1 \mathsf{F}^\text{R} \realpart{H_{k,b}^* \reside_k^t} 
% \\&= \sum_{b=1}^{\MT} \partial_1 \mathsf{F}^\text{R} \!
% \left[\realpart{H_{k,b}}^2\realpart{\reside_k^t}
% + i\imagpart{H_{k,b}}^2 \imagpart{\reside_k^t}\right] \\
% &+i \sum_{b=1}^\MT \partial\mathsf{F}^\text{R}\realpart{H_{k,b}}\imagpart{H_{k,b}}\reside_k^{t*} 
% used avg notation \frac{1}{\MT}\sum_{b=1}^{\MT}\partial_1 \mathsf{F}^\text{R},\\
&= \frac{\beta}{2}  \reside_k^t\langle \partial_1 \mathsf{F}^\text{R}\rangle,\\
% Fr_2
\sum_{b=1}^{\MT} H_{k,b}\partial_2 \mathsf{F}^\text{R} \imagpart{H_{k,b}^* \reside_k^t}
% &= \sum_{b=1}^{\MT}\partial_2 \mathsf{F}^\text{R}\!
% \left[
% \realpart{H_{k,b}}^2\imagpart{\reside_k^t} 
%  - i \imagpart{H_{k,b}}^2\realpart{\reside_k^t}
% \right]\\
% &-\sum_{b=1}^\MT\partial_2 \mathsf{F}^\text{R} \realpart{H_{k,b}} \imagpart{H_{k,b}} \reside_k^{t*}
%
%
% used avg notation instead
% &\rightarrow \frac{\beta}{2i} \reside_k^t \frac{1}{\MT}\sum_{b=1}^{\MT}\partial_2\mathsf{F}^\text{R},
&= \frac{\beta}{2i} \reside_k^t 
\langle
% \frac{1}{\MT}\sum_{b=1}^{\MT}
\partial_2\mathsf{F}^\text{R}
\rangle,
%
% added these
%
%
%
\\
i\sum_{b=1}^{\MT} H_{k,b} \partial_1 \mathsf{F}^\text{I} \realpart{H_{k,b}^* \reside_k^t}
&=
\frac{\beta i}{2}\reside_k^t
\langle
\partial_1 \mathsf{F}^\text{I}
\rangle,
\\
i \sum_{b=1}^{\MT} H_{k,b}\partial_2 \mathsf{F}^\text{R} \imagpart{H_{k,b}^* \reside_k^t}
&=
\frac{\beta}{2}  \reside_k^t
\langle\partial_2 \mathsf{F}^\text{I}
\rangle,
\end{align*}
By using the Hadamard product, we arrive at \fref{alg:cB-AMP}.
% 

%% ----------------------------------
%% ----------------------------------
%% ----------------------------------
%% ----------------------------------
%% ----------------------------------
%% ----------------------------------
%% ----------------------------------
%% ----------------------------------
\section{Proof of \fref{lem:B_AMP_same_cB_AMP}}\label{app:B_AMP_same_cB_AMP}

We use the facts that $\partial_2\mathsf{F}^\text{R}$, $\partial_1\mathsf{F}^\text{I}$, and $\partial_2\mathsf{F}^\text{I}$ are all zero for real-valued systems. 
Moreover, since the $\ell_2$-norm of each column of $\bH$ is one according to (A1), the update \fref{eq:cBAMP-residual} simplifies to 
\begin{align}\notag
\reside_k^{t+1} - \left(y_k - \bmh_k^\row\shat^{t+1} \right) 
% &= \sum_{b=1}^{\MT} H_{k,b} \left[{\mathsf{F}'}\left(\shate_b^t + \sum_{a=1}^{\MR}H_{a,b}^*\reside_a^t\right)\right]\realpart{H_{k,b}^* \reside_k^t} \notag \\
&=  \reside_k^t\sum_{b=1}^{\MT}H_{k,b}^2{\mathsf{F}'}\left(\shate_b^t + \sum_{a=1}^{\MR}H_{a,b}\reside_a^t\right)
% =  \reside_k^t \sum_{b=1}^{\MT} \frac{1}{\MR} {\mathsf{F}'}\left(\shate_b^t + \sum_{a=1}^{\MR}H_{a,b}\reside_a^t\right) 
\notag \\
&= \beta \reside_k^t\!\left\langle {\mathsf{F}'}\!\left(\shat^t + \bH^T \resid^t\right)\right\rangle\!, \label{eq:realBAMP}
\end{align}
where $\mathsf{F}'$ is the derivative of the mean function $\mathsf{F}(\shate_\ell,\tau)$ taken with respect to $\shate_\ell$. 
The final comparison of \fref{eq:realBAMP} with \cite[Eq.~5.74]{Maleki2010phd} reveals equivalence of real-valued cB-AMP and B-AMP. 
%% -------------------------------------------
%% -------------------------------------------
%% -------------------------------------------
%% -------------------------------------------
%% -------------------------------------------
%% -------------------------------------------
%% -------------------------------------------
%% -------------------------------------------

\section{Intuitive derivation of \fref{thm:CSE}}\label{app:statederiv}
% \cs{arian: is this a rigorous derivation? or are there holes we need to fill?}

We present a non-rigorous derivation of \fref{thm:CSE} for complex-valued systems; a rigorous proof can be found in \cite{BM2011}.
Assume that the MIMO channel $\bH(t)$ changes each iteration $t$, where the elements are distributed $\setC\setN\left(0,1/\MR\right)$.
In addition, let $\mathsf{F}(z,\tau)$ and $\mathsf{G}(z,\tau)$ are functions defined in \fref{eq:F_compute} and \fref{eq:G_compute} according to the mean and variance of the distribution in \fref{eq:gaussianmeasure}, respectively.
Let $\bmy^t = \bH(t)\bms_0 + \bmn$ where the entries of $\bmn$ are circularly symmetric complex Gaussian with variance $\No$.
Assuming that we fix the postulated noise variance to $\Nopost$, then, in each iteration, the recursion is defined as:
\begin{align}\label{eq:statez}
\resid^{t} &= \bmy^{t} - \bH\shat^{t},\\
\shat^{t+1} &= \mathsf{F}\!\left(\shat^t+\bH^\Herm(t) \resid^t,\Nopost(1+\tau^t)\right)\!,\label{eq:statex}\\
\tau^{t+1} &= \frac{\beta}{\Nopost}\!\left\langle\mathsf{G}\left(\shat^t+\bH^\Herm(t) \resid^t,\Nopost(1+\tau^t)\right)\right\rangle\!.\label{eq:statetau}
\end{align}

By substituting $\resid^{t}$ in \fref{eq:statez} into $\shat^t+\bH^\Herm(t) \resid^t$, we have that
\begin{align}\label{eq:stateevoeta}
&\shat^t+\bH^\Herm(t) \resid^t = \bH^\Herm(t)\bmy^t + \!\left(\bI_{\MT} - \bH^\Herm(t)\bH(t)\right)\!\shat^t\notag\\
&= \bms_0 + \bH^\Herm(t)\bmn + \!\left( \bI_{\MT} - \bH^\Herm(t)\bH(t)\right)\!\!\left(\shat^t - \bms_0\right)\!.
\end{align}
The central limit theorem shows that each diagonal and non-diagonal entry in $\bI_{\MT} - \bH^\Herm(t)\bH(t)$ is distributed $\setN(0,1/{\MR})$ and $\setC\setN(0,1/{\MR})$ respectively, with pairwise independent entries. 
Also, for each $\ell$th entry in $\left( \bI_{\MT} - \bH^\Herm(t)\bH(t)\right)\left(\shat^t - \bms_0\right)$, the real and imaginary parts are normally distributed with zero mean and variance $\frac{\Vert\shat^t - \bms_0\Vert_2^2}{2\MR}+\delta_\ell^t$ and $\frac{\Vert\shat^t - \bms_0\Vert_2^2}{2\MR}-\delta_\ell^t$ respectively, with $\delta_\ell^t = \frac{\text{Re}\{(\shate^t_{\ell} - s_{0\ell})^2\}}{2\MR}$. 
With
\begin{align}\label{eq:stateevosigma}
\hat\sigma_{t}^2 = \lim_{\MT\rightarrow\infty}\left\Vert\shat^t-\bms_0\right\Vert^2/{\MT},
\end{align}
and noting that $\delta_\ell^t\rightarrow0$ as $\MT\rightarrow\infty$, we have that
\begin{align*}
\!\left(\bI_{\MT} - \bH^\Herm(t)\bH(t)\right)\!\!\left(\shat^t - \bms_0\right)\rightarrow \setC\setN\!\left(0,\beta\hat\sigma_{t}^2\right)\!.
\end{align*}
Moreover, by conditioning on $\bmn$, $\bH^\Herm(t)\bmn\rightarrow\setC\setN\left(0,\No\right)$ by the law of large numbers. 
By \fref{def:effective_var} of the effective noise variance of cB-AMP, we have the relation $\sigma_{t}^2 = \No+\beta\hat\sigma_{t}^2$ with $\hat\sigma_{t}^2$ defined in \fref{eq:stateevosigma}. 
Thus, each $\ell$th entry of \fref{eq:statex} converges to $\mathsf{F}\!\left(s_{0\ell} + \sigma_t Z,\Nopost(1+\tau^t)\right)$ where $Z\sim\setC\setN\left(0,1\right)$. Since we assume a fixed prior distribution for all $s_{0\ell}$, we obtain the following recursion for \fref{eq:statex}:
\begin{align*}
\sigma_{t+1}^2 &= \No + \beta\!\!\lim_{\MT\rightarrow\infty}\frac{1}{\MT}\left\Vert\shat^{t+1} - \bms_0\right\Vert^2\\
&= \No + \beta\Exop_{\srv,Z}\!\left[\abs{\mathsf{F}\left(\srv + \sigma_t Z,\Nopost(1+\tau^t)\right)- \srv}^2 \right]\!,
\end{align*}
with $S\sim p(S)$. 
Starting from \fref{eq:statetau}, we use \fref{eq:stateevoeta} and \fref{eq:stateevosigma}, and the law of large numbers to obtain:
\begin{align*}
\tau^{t+1} &= \frac{\beta}{\Nopost}\!\left\langle\mathsf{G}\!\left(\shat^t+\bH^\Herm(t) \resid^t,\Nopost(1+\tau^t)\right)\!\right\rangle\!
\\
&= \frac{\beta}{\Nopost}\Exop_{\srv,Z}\!\left[\mathsf{G}\!\left(\srv + \sigma_t Z,\Nopost(1+\tau^t)\right)\!\right]\!.
\end{align*}
By introducing the postulated variance $\gamma_t^2 = \Nopost(1+\tau^t)$, we obtain the final cSE:
\begin{align*}
\sigma_{t+1}^2 &= \No + \beta\Exop_{\srv,Z}\!\left[\abs{\mathsf{F}\!\left(\srv + \sigma_t Z,\gamma_t^2\right)\!- \srv}^2 \right]\\
\gamma_{t+1}^2 &= \Nopost+ \beta\Exop_{\srv,Z}\!\left[\mathsf{G}\!\left(\srv + \sigma_t Z,\gamma_t^2\right)\!\right]\!.
\end{align*}
%% ---------------------
%% ---------------------
We reiterate that the formulation of $\bH(t)$ to obtain cSE in \fref{thm:CSE} was non-rigorous; a rigorous proof can be found in \cite{BM2011}.

\section{Proof of \fref{lem:LAMA_alg}}\label{app:LAMA_alg}
We start with cB-AMP as detailed in \fref{alg:cB-AMP}. 
We simplify intermediate steps in cB-AMP using the definition of $\mathsf{F}(\shate_\ell,\tau)$ and $\mathsf{G}(\shate_\ell,\tau)$, and our knowledge of the prior distribution.
Recall that $\mathsf{F}(\shate_\ell,\tau)$ in \fref{eq:DAMP_F} was defined as
\begin{align*}
\mathsf{F}(\shate_\ell,\tau) = \sum_{a\in\setO} w_a(\shate_\ell,\tau)a,
\end{align*}
By taking partial derivatives of $\mathsf{F}(\shate_\ell,\tau)$ with the notations defined in \fref{alg:cB-AMP}, we have the following expressions, where we drop the notation $w_a = w_a(\shate_\ell,\tau)$ for simplicity.
\begin{align*}
\partial_1\mathsf{F}^\text{R} &= \frac{2}{\tau} \left[\sum_{a\in\setO} \realpart{a}^2 w_a - \!\left(\sum_{a\in\setO} \realpart{a} w_a\right)^{\!\!2} \right]\!,
\\
\partial_2 \mathsf{F}^\text{I} &= \frac{2}{\tau} \left[\sum_{a\in\setO} \imagpart{a}^2 w_a - \!\left(\sum_{a\in\setO} \imagpart{a} w_a\right)^{\!\!2} \right]\!,
\\
\partial_2\mathsf{F}^\text{I} &= \partial_1\mathsf{F}^\text{R} = \frac{2}{\tau}
\Bigg[
\sum_{a\in\setO} \realpart{a} \imagpart{a}  w_a \\
&\quad\quad\quad\quad - \!\bigg(\sum_{a\in\setO} \realpart{a} w_a\bigg)\!
\!\bigg(\sum_{a\in\setO} \imagpart{a}  w_a\bigg)\! 
\Bigg].
\end{align*}
Note that \fref{eq:DAMPrealG} can be separated in real and imaginary parts. Therefore,
\begin{align*}
\mathsf{G}(\shate_\ell,\tau) &= \sum_{a\in\setO}\abs{a}^2w_a - \abs{\sum_{a\in\setO}a w_a}^2
\\
&= \sum_{a\in\setO}\realpart{a}^2w_a - \!\left(\sum_{a\in\setO}\realpart{a} w_a\right)^{\!\!2}\\
&\quad + \sum_{a\in\setO}\imagpart{a} ^2w_a - \!\left(\sum_{a\in\setO}\imagpart{a}  w_a\right)^{\!\!2}\\
&= \frac{\tau}{2}\!\left[\partial_1 \mathsf{F}^\text{R} + \partial_2 \mathsf{F}^\text{I}\right]\!(\shate_\ell,\tau)
\end{align*}
Finally, observe that $\partial_1\mathsf{F}^\text{I}=\partial_2\mathsf{F}^\text{R}$ and $1/i = -i$, so \fref{lem:LAMA_alg} simplifies to:
\begin{align*}
\shat^{t+1} &= \mathsf{F}\!\left(\shat^t+\bH^\Herm\resid^t,\Nopost(1+\tau^t) \right)\!\\
\tau^{t+1} &= \frac{\beta}{\Nopost}\!\left\langle\mathsf{G}\!\left(\shat^t+\bH^\Herm\resid^t,\Nopost(1+\tau^t)\right)\!\right\rangle\!\\
\resid^{t+1} &= \bmy - \bH \shat^{t+1} + \frac{\tau^{t+1}}{1+\tau^t}\resid^t.
\end{align*}
% %%
%%%%%%%%%%%%%%%%%%%%%%%%%%%%%%%%%%%%%%%%%%%%%%%%%%%%
%%%%%%%%%%%%%%%%%%%%%%%%%%%%%%%%%%%%%%%%%%%%%%%%%%%%
%%%%%%%%%%%%%%%%%%%%%%%%%%%%%%%%%%%%%%%%%%%%%%%%%%%%
%%%%%%%%%%%%%%%%%%%%%%%%%%%%%%%%%%%%%%%%%%%%%%%%%%%%
%%%%%%%%%%%%%%%%%%%%%%%%%%%%%%%%%%%%%%%%%%%%%%%%%%%%
%%%%%%%%%%%%%%%%%%%%%%%%%%%%%%%%%%%%%%%%%%%%%%%%%%%%
%%%%%%%%%%%%%%%%%%%%%%%%%%%%%%%%%%%%%%%%%%%%%%%%%%%%
%%%%%%%%%%%%%%%%%%%%%%%%%%%%%%%%%%%%%%%%%%%%%%%%%%%%
%%%%%%%%%%%%%%%%%%%%%%%%%%%%%%%%%%%%%%%%%%%%%%%%%%%%
%%%%%%%%%%%%%%%%%%%%%%%%%%%%%%%%%%%%%%%%%%%%%%%%%%%%
%%%%%%%%%%%%%%%%%%%%%%%%%%%%%%%%%%%%%%%%%%%%%%%%%%%%
%%%%%%%%%%%%%%%%%%%%%%%%%%%%%%%%%%%%%%%%%%%%%%%%%%%%

\section{Proof of \fref{lem:LAMAtoMF}}\label{app:LAMA_MF}
Since $\Nopost\rightarrow\infty$, the recursions in \fref{alg:LAMA} are given by
\begin{align*}
\shat^{t} &= \!\!\lim_{\Nopost\rightarrow\infty}\!\!\mathsf{F}\!\left(\shat^{t-1}+\bH^\Herm\resid^{t-1},\Nopost\right)\!,\\
\resid^{t} &= \bmy - \bH\shat^{t}.
\end{align*}
First of all, notice that as $\Nopost\rightarrow\infty$, $w_a(\shate_\ell,\Nopost)\rightarrow p_a$ for any $\shate_\ell$. Therefore, for all $t$,
\begin{align*}
\shat^{t} &= \!\!\lim_{\Nopost\rightarrow\infty}\!\!\mathsf{F}\!\left(\shat^{t+1}+\bH^\Herm\resid^{t-1},\Nopost\right)\!\rightarrow \sum_{a\in\setO}a p_a = 0,\\
\resid^{t} &= \bmy - \bH\shat^{t} = \bmy,
\end{align*}
and thus, the Gaussian output $\bmz^{t}=\shat^{t}+\bH^\Herm\resid^t$ is equivalent to the matched filter output $\bH^\Herm\bmy$. To show that the non-linear MMSE output corresponds to the matched filter involves computing $\lim_{\Nopost\rightarrow\infty} \frac{\Nopost}{E_s} \shat^{t+1}$, which is given by
\begin{align*}
&\lim_{\Nopost\rightarrow\infty}\!\frac{\Nopost}{E_s}\shat^{t+1} = \lim _{\Nopost\rightarrow\infty}\!\frac{\Nopost}{E_s}\mathsf{F}\!\left(\bH^\Herm\bmy,\Nopost\right)\!\\
&= \lim_{\Nopost\rightarrow\infty}\!\frac{\Nopost}{E_s}\sum_{a\in\setO}a p_a\!\left(1 - \frac{1}{\Nopost}\abs{\bH^\Herm\bmy-a}^2\right)\!
\\&= -\frac{1}{E_s}\sum_{a\in\setO}a p_a \abs{\bH^\Herm\bmy-a}^2 
=\bH^\Herm\bmy.
\end{align*}
% so letting $\Nopost\rightarrow\infty$ corresponds to the matched filter output for both outputs for LAMA.

%%%%%%%%%%%%%%%%%%%%%%%%%%%%%%%%%%%%%%%%%%%%%%%%%%%%
%%%%%%%%%%%%%%%%%%%%%%%%%%%%%%%%%%%%%%%%%%%%%%%%%%%%
%%%%%%%%%%%%%%%%%%%%%%%%%%%%%%%%%%%%%%%%%%%%%%%%%%%%
%%%%%%%%%%%%%%%%%%%%%%%%%%%%%%%%%%%%%%%%%%%%%%%%%%%%
%%%%%%%%%%%%%%%%%%%%%%%%%%%%%%%%%%%%%%%%%%%%%%%%%%%%
%%%%%%%%%%%%%%%%%%%%%%%%%%%%%%%%%%%%%%%%%%%%%%%%%%%%

\section{Proof of \fref{lem:LAMAbeta0}}\label{app:LAMAbeta0}
First, note that as $\beta\rightarrow 0$ for a fixed $\Nopost$, we have that $\tau^t=0$ for all $t\geq1$. Therefore, we have the following recursions,
\begin{align*}
\shat^{t+1} &= 
\mathsf{F}\!\left(\shat^t+\bH^\Herm\resid^t,\Nopost\right)\!,\\
\resid^{t+1} &= \bmy - \bH\shat^{t+1}.
\end{align*}
Following the derivation of complex state evolution in \fref{app:statederiv}, as $\beta\rightarrow 0$, we have
\begin{align*}
\shat^t+\bH^\Herm\resid^t = \bH^\Herm\bmy+ (\bI_\MT - \bH^\Herm\bH)\shat^t \rightarrow \bH^\Herm\bmy,
\end{align*}
because the entries of $(\bI_\MT - \bH^\Herm\bH)\shat^t$ converge to a complex normal distribution with zero mean and variance $\beta\tilde\sigma^2_t$ with $\tilde\sigma^2_t = \lim_{\MT\rightarrow\infty}\frac{1}{\MT}\vecnorm{\shat^t}^2$. 
Since $\tilde\sigma^2_t$ is finite and $\beta\rightarrow 0$, we have that the Gaussian output of LAMA is $\bmz^t = \bH^\Herm\bmy$ (independent of the iteration index $t$). Hence, the non-linear MMSE output $\shat^{t+1}$ of \LAMA is given by $\mathsf{F}\!\left(\bH^\Herm\bmy,\Nopost\right)\!$ for all $t$.

We show that one iteration of \LAMA is sufficient to achieve AWGN performance by the cSE in \fref{thm:CSE}. 
Recall that previous paragraph demonstrated that $\bmz^t=\bH^H\bmy$ for all $t$.
Thus, the equivalent output noise variance 
% , computed as $\Varop[\bmz^t-\bms_0]$, 
is computed as $\sigma^2=\No+\beta\Varop_S[S]=\No$, where the last step comes from $\beta\rightarrow 0$.
Since each output of \LAMA is identical every iteration and the output noise variance is $\No$, one iteration is sufficient to achieve AWGN performance.

%%%%%%%%%%%%%%%%%%%%%%%%%%%%%%%%%%%%%%%%%%%%%%%%%%%%
%%%%%%%%%%%%%%%%%%%%%%%%%%%%%%%%%%%%%%%%%%%%%%%%%%%%
%%%%%%%%%%%%%%%%%%%%%%%%%%%%%%%%%%%%%%%%%%%%%%%%%%%%
%%%%%%%%%%%%%%%%%%%%%%%%%%%%%%%%%%%%%%%%%%%%%%%%%%%%
%%%%%%%%%%%%%%%%%%%%%%%%%%%%%%%%%%%%%%%%%%%%%%%%%%%%
%%%%%%%%%%%%%%%%%%%%%%%%%%%%%%%%%%%%%%%%%%%%%%%%%%%%

\section{Proof of \fref{lem:DAMPsolvability}}\label{app:DAMPsolvability}
Note that since $\No=0$, we have $\sigma^2 = \gamma^2$ by \fref{cor:cSE_IO}. 
Since the variance of $S$ is finite, denote $\Varop_S[S]=\sigma_s^2$. By \cite[Prop. 15]{GWSS2011}, we have the following upper bound for $\Psi(\sigma^2,\sigma^2)$:
\begin{align}\label{eq:MSIbeta1}
\Psi(\sigma^2,\sigma^2) \leq \frac{\sigma_s^2}{\sigma_s^2+\sigma^2}\sigma^2,
\end{align}
where equality is achieved for all $\sigma^2$ if and only if $S$ is complex normal with variance $\sigma_s^2$. Note that if $\sigma^2=0$, then \fref{eq:MSIbeta1} is achieved for any $\sigma_s^2$. If $\sigma^2>0$, then
\begin{align*}
\Psi(\sigma^2,\sigma^2)&\leq\frac{\sigma_s^2}{\sigma_s^2+\sigma^2}\sigma^2= \frac{1}{1+\sigma^2/\sigma_s^2}\sigma^2 < \sigma^2,
\end{align*}
and, hence, the proof follows.

The first part of \fref{lem:DAMPsolvability} is trivial from \fref{eq:MSIbeta1}, and thus, $\Psi(\sigma^2,\sigma^2)\rightarrow0$ as $\sigma^2\rightarrow0$. The second part is noting that as $\sigma^2\rightarrow\infty$, $\mathsf{F}(\cdot,\sigma^2)\rightarrow \sum_{a\in\setO}a p_a = \Exop_S[S]$, and hence we have
\begin{align*}
\lim_{\sigma^2\rightarrow\infty}\Psi(\sigma^2,\sigma^2) \rightarrow\Exop_{S}\!\left[
\abs{
	S - \Exop_S[S]
}^2
\right]\! = \Varop_S[S].
\end{align*}

%% =-----------------------

\section{Proof of \fref{thm:recovery}}\label{app:recovery}

We assume the initialization in \fref{alg:LAMA}. Since $\No=\Nopost=0$, if \LAMA perfectly recovers the true signal $\bms_0$, then the fixed-point \fref{eq:fixed_pt} is unique at $\sigma^2=0$.
This happens if the system ratio is strictly less than the ERT, $\betamax$ because otherwise, i.e., $\beta\geq\betamax$, there exists a non-unique fixed point to \fref{eq:fixed_pt} for some $\sigma^2>0$ by \fref{def:ERT}.

%% =-----------------------

\section{Proof of \fref{lem:MRTERT}}\label{app:MRTERT}

We show that for a fixed constellation $\setO$, $\betamin\leq\betamax$. 
For conciseness, define $\sigma^2_\star$ as the fixed-point $\sigma^2=\betamax\Psi(\sigma^2,\sigma^2)$.
The proof is straightforward as,
\begin{align*}
\betamin &\stackrel{(a)}{=} \min_{\sigma^2>0}\!\left\{\!
\left(\frac{\textnormal{d}\Psi(\sigma^2,\sigma^2)}{\textnormal{d}\sigma^2}
\right)^{\!\!-1}
\right\}\!
% \\&
\leq
\!\left.\!
\left(\frac{\textnormal{d}\Psi(\sigma^2,\sigma^2)}{\textnormal{d}\sigma^2}
\right)^{\!\!-1}
% \right\vert_{\sigma^2=\betamax\Psi(\sigma^2,\sigma^2)}
\right\vert_{\sigma^2=\sigma^2_\star}
\\&
\stackrel{(b)}{=}\!\left(\frac{1}{\betamax}
\right)^{\!\!-1} \!= \betamax,
\end{align*}
where (a) and (b) follow from the definitions of MRT and ERT, respectively.

%%%%%%%%%%%%%%%%%%%%%%%%%%%%%%%%%%%%%%%%%%%%%%%%%%%%
%%%%%%%%%%%%%%%%%%%%%%%%%%%%%%%%%%%%%%%%%%%%%%%%%%%%
%%%%%%%%%%%%%%%%%%%%%%%%%%%%%%%%%%%%%%%%%%%%%%%%%%%%

\section{Proof of \fref{lem:fp_IOLAMA_0}}\label{app:fp_IOLAMA_0}

As $\beta<\betamax$, there exists a value of $\No$, denote it as $\No^\star$, such that for $\No<\No^\star$, the fixed-point solution of \LAMA is unique. 
We note that $\Nomin$ is also a candidate for $\No^\star$ as the fixed-point solution of \LAMA is unique for all $\No<\Nomin$.
In addition, since $\setO$ is a constellation, by \cite[Thm. 10]{WV2010}, $\Psi(\sigma^2,\sigma^2)$ has a continuous derivative and $\lim_{\sigma^2\to0}\frac{\dd}{\dd \sigma^2}\Psi(\sigma^2,\sigma^2)=0$.
Hence, there exists a value $\sigma_\star^2$ such that for all $\sigma^2<\sigma_\star^2$, 
\begin{align}\label{eq:fp_IOLAMA_0}
\frac{\dd}{\dd \sigma^2}\Psi(\sigma^2,\sigma^2) < \frac{1}{2\beta}.
\end{align}
Now, suppose that $\No^\star<\sigma_\star^2/2$. Then, for all $\sigma^2<\sigma_\star^2$ we have:
\begin{align*}
\No+\beta\Psi(\sigma^2,\sigma^2)\stackrel{(a)}{<}\No+\frac{\sigma^2}{2},
\end{align*}
where $(a)$ follows from \fref{eq:fp_IOLAMA_0} and the mean value theorem.
Since $2\No<2\No^\star<\sigma_\star^2$, we have that:
\begin{align*}
\No+\beta\Psi(2\No,2\No)<\No+\No=2\No,
\end{align*}
and therefore, the fixed-point solution $\sigma^2$ has to be between $\No$ and $2\No$. As a result, as $\No\to0$, the fixed-point solution $\sigma^2\to0$.
The last part is apparent as:
\begin{align*}
\lim_{\No\to0}1 &= 
\lim_{\No\to0}\frac{\No}{\sigma^2} + \beta\lim_{\No\to0}\frac{\Psi(\sigma^2,\sigma^2)}{\sigma^2}
\\
&=\lim_{\No\to0}\frac{\No}{\sigma^2} + \beta
\lim_{\sigma^2\to0}\frac{\dd}{\dd\sigma^2}\Psi(\sigma^2,\sigma^2) = 
\lim_{\No\to0}\frac{\No}{\sigma^2},
\end{align*}

%%%%%%%%%%%%%%%%%%%%%%%%%%%%%%%%%%%%%%%%%%%%%%%%%%%%
%%%%%%%%%%%%%%%%%%%%%%%%%%%%%%%%%%%%%%%%%%%%%%%%%%%%
%%%%%%%%%%%%%%%%%%%%%%%%%%%%%%%%%%%%%%%%%%%%%%%%%%%%

%% =-----------------------
\section{Proof of \fref{lem:separability}}\label{app:separability}
Since the constellation $\setO$ is separable, we introduce a shorthand notation for $\setO^\textnormal{R} = \realpart{\setO}$ and $\setO^\textnormal{I} = \imagpart{\setO}$.
It is easy to observe that the weight scalar $w_a(\shate_\ell,\tau)$ can be rewritten as a product between the weight scalar of the real and imaginary constellation $w_{\ar}(\shate_\ell,\tau)$ and $w_{\ai}(\shate_\ell,\tau)$, i.e., $w_a(\shate_\ell,\tau )=w_{\ar}(\shate_\ell,\tau)w_{\ai}(\shate_\ell,\tau)  $ where
\begin{align}\label{eq:weight_separable}
w_{\ar}(\shate_\ell,\tau) = \frac{p_{\ar}\exp\! \left( - \frac{1}{\tau}(\realpart{\shate_\ell} - \ar)^2\right)\!}{\sum_{\ar\in\setO^\textnormal{R}}p_{\ar}\exp\!\left( - \frac{1}{\tau}(\realpart{\shate_\ell} - \ar) ^2\right)\!},
\end{align}
and likewise for $w_{\ai}$. 
Therefore, $\mathsf{F}$ is separable because
\begin{align}
\notag
&\mathsf{F}(\shate_\ell,\tau) = 
\sum_{a\in\setO} w_a(\shate_\ell,\tau) a 
\\
\notag
&= 
\sum_{a\in\setO} w_a(\shate_\ell,\tau) \ar + i \sum_{a\in\setO} w_a(\shate_\ell,\tau) \ai \\
\notag
&= \sum_{\ar\in\setO^\textnormal{R}} 
\ar \sum_{\ai\in\setO^\textnormal{I}}
w_a(\shate_\ell,\tau) 
+ 
i\sum_{\ai\in\setO^\textnormal{I}}
\ai\sum_{\ar\in\setO^\textnormal{R}} 
w_a(\shate_\ell,\tau) \\
\label{eq:F_seperable_complex}
&=  \sum_{\ar\in\setO^\textnormal{R}}
w_{\ar}(\shate_\ell,\tau) \ar + 
i  \sum_{\ai\in\setO^\textnormal{I}}
w_{\ai}(\shate_\ell,\tau) \ai.
\end{align}

Now, for a real-valued constellation $\setO^\textnormal{R}$, the message mean $\mathsf{F}^\textnormal{R}$ is given by:
\begin{align}
\label{eq:F_separable_real}
\mathsf{F}^\textnormal{R}(\shate_\ell,\tau) = \sum_{a\in\setO^\textnormal{R}} w_a^\textnormal{R}(\shate_\ell,\tau) a, 
\end{align}
where the weight scalar for the real-valued constellation is computed by
\begin{align*}
w_{a}^\textnormal{R}(\shate_\ell,\tau) = 
\frac{p_{a}\exp\! \left( - \frac{1}{2\tau}(\shate_\ell - a)^2\right)\!}{\sum_{a\in\setO^\textnormal{R}}p_{a}\exp\!\left( - \frac{1}{2\tau}(\shate_\ell - a) ^2\right)\!}.
\end{align*}

Therefore, we have that:
\begin{align*}
&\Psi(\sigma^2,\gamma^2) 
=  
\Exop_{\srv,Z}\!\left[\abs{\mathsf{F}\!\left(\srv + \sigma Z,\gamma^2\right) - \srv}^2 \right]\\
&\stackrel{(a)}{=}
\Exop_{\Sr,Z_\textnormal{R}}\!\left[\!\left(
\realpart{\mathsf{F}\Big(\Sr + \frac{\sigma}{\sqrt{2}} Z_\textnormal{R},\gamma^2\Big)}
- \Sr \right)^{\!2} \right]\\
&\quad + 
\Exop_{S_\textnormal{I},Z_\textnormal{I}}\!\left[\!\left(
\imagpart{\mathsf{F}\Big(S_\textnormal{I} +  \frac{\sigma}{\sqrt{2}} Z_\textnormal{I},\gamma^2\Big)}
- S_\textnormal{I} \right)^{\!2} \right]\\
&\stackrel{(b)}{=}
2\Exop_{\Sr,Z_\textnormal{R}}\!\left[\!\left(
\realpart{\mathsf{F}\Big(\Sr +  \frac{\sigma}{\sqrt{2}} Z_\textnormal{R},\gamma^2\Big)}
- \Sr \right)^{\!2} \right]\\
&\stackrel{(c)}{=}
2\Exop_{\Sr,Z_\textnormal{R}}\!\left[\!\left(
\sum_{a\in\setO^\textnormal{R}} w_a^\textnormal{R}\Big(\Sr + \frac{\sigma}{\sqrt{2}} Z_\textnormal{R},\frac{\gamma^2}{2}\Big) a 
- \Sr \right)^{\!2} \right]\\
&= 
2 \Psi^\textnormal{R}\!\left(\frac{\sigma^2}{2},\frac{\sigma^2}{2}\right)\!,
\end{align*}
where $(a)$ follows from \fref{eq:F_seperable_complex}, $(b)$ from definition of separable constellation, and $(c)$ follows from construction of \fref{eq:F_separable_real}.
We note that the case for variance function $\Phi$ is derived similarly.

\section{Proof of \fref{lem:separable_const_same}}\label{app:separable_const_same}
We show that for a separable constellation $\setO$, the MRT and ERT are equivalent.
Denote the complex-valued MSE function as $\Psi(\sigma^2,\sigma^2)= \Exop_{\srv,Z}\!\left[\abs{\mathsf{F}\!\left(\srv + \sigma Z,\sigma^2\right) - \srv}^2 \right]$, where $Z\sim\setC\setN(0,1)$, and  $\srv\sim p(\srv)$ for constellation $\setO$.
Denote the real-valued MSE function $\Psi^\text{R}(\sigma^2,\sigma^2)= \Exop_{\Sr,Z^\text{R}}\!\left[\!\left(\mathsf{F}\!\left(\Sr + \sigma Z^\text{R},\sigma^2\right) - \Sr\right)^2 \right]$, where $Z^\text{R}\sim\setN(0,1)$ and $\Sr\sim p(\Sr)$ for the real-valued constellation $\realpart{\setO}$.
We know from \fref{lem:separability} that $\Psi(\sigma^2,\sigma^2)=2\Psi^\text{R}(\sigma^2/2,\sigma^2/2)$. 
By denoting $\betaminno_\complexset$ and $\betaminno_\reals$ as the MRT of complex- and real-valued MSE function, respectively, we have:
\begin{align*}
\betaminno_\complexset 
&=
\min_{\sigma^2\geq0}
\!\left\{\!
\left(\frac{\textnormal{d}\Psi(\sigma^2,\sigma^2)}{\textnormal{d}\sigma^2}
\right)^{\!\!-1}\!
\right\}
\\
&=
\min_{\sigma^2\geq0}
\!\left\{\!
\left(\frac{\textnormal{d}\Psi^\text{R}(\frac{\sigma^2}{2},\frac{\sigma^2}{2})}{\textnormal{d}\sigma^2/2}
\right)^{\!\!-1}\!
\right\}
% =
\\&=
\min_{\bar\sigma^2\geq0}
\!\left\{\!
\left(\frac{\textnormal{d}\Psi^\text{R}(\bar\sigma^2,\bar\sigma^2)}{\textnormal{d}\bar\sigma^2}
\right)^{\!\!-1}\!
\right\}
\\&=
% = 
\betaminno_\reals
% \!\left.\!
% \left(\frac{\textnormal{d}\Psi(\sigma^2,\sigma^2)}{\textnormal{d}\sigma^2}
% \right)^{\!\!-1}
% \right\vert_{\sigma^2=\betamax\Psi(\sigma^2,\sigma^2)}
% \stackrel{(b)}{=}\!\left(\frac{1}{\betamax}
% \right)^{\!\!-1} \!= \betamax,
\end{align*}
The remaining quantities, $\betamaxno$ and the critical noise levels of \fref{lem:separable_const_same} can be derived similarly. 
% \cj{show that N0min/max for complex = 2N0min/max for reals}

% 
% 
% 
% 
% 
% 
% 
% 
% 
% 
% 
% 
% 
% 
% 
% 
% 
% 
% 
% \\

\section{Proof of \fref{lem:convergence_speed}}\label{app:convergence_speed}

Note that if $\beta<\betamin$, then the slope of the function $\beta\Psi(\sigma^2,\sigma^2)$ with respect to $\sigma^2$ is always less than 1, i.e., $
\beta\!\!\left.\frac{\textnormal{d}}{\textnormal{d}\sigma^2}\Psi(\sigma^2,\sigma^2)\right\vert_{\sigma^2=\sigma^2_\star}<1$ for any $\sigma_\star^2>0$.
In addition, since $\beta<\betamin$, the fixed-point solution to $\No+\beta\Psi(\sigma^2,\sigma^2)=\sigma^2$ is unique.
Now the exponentially-fast convergence result can be shown by using the bounding technique of \cite[Lem. 6.4.1]{Maleki2010phd}.

In \cite{Maleki2010phd}, the proof for showing exponentially-fast convergence of standard AMP to its largest fixed-point solution was shown by analyzing the stability constant $\text{SC}(\Psi)$ which was defined by:
\begin{align*}
\textnormal{SC}(\Psi) &= \beta\!\!\left.\frac{\textnormal{d}}{\textnormal{d}\sigma^2}\Psi(\sigma^2,\sigma^2)\right\vert_{\sigma^2=\sigma^2_\star},
\\
% \quad 
\sigma_\star^2 &= \max_{\sigma^2>0}
\!\left\{
\sigma^2: \No+\beta\Psi(\sigma^2,\sigma^2)\geq \sigma^2
\right\}
\end{align*}
Using the new notation of stability constant, the condition of $\textnormal{SC}(\Psi)<1$, i.e., slope at the largest fixed point is less than 1, was only needed in \cite{Maleki2010phd} to show exponential-fast convergence.
However, we note that this approach was viable in \cite{Maleki2010phd} due to concavity of the MSE function of $\Psi(\sigma^2,\sigma^2)$ for the soft-thresholding function; however, the MSE function of \LAMA does not have such properties. 
In fact, the MSE function for \LAMA for commonly used constellation in wireless is neither convex nor concave. 
However, as shown above, if $\beta<\betamin$, we have that not only $\textnormal{SC}(\Psi)<1$, but also $
\beta\!\!\left.\frac{\textnormal{d}}{\textnormal{d}\sigma^2}\Psi(\sigma^2,\sigma^2)\right\vert_{\sigma^2=\sigma^2_\star}<1$ for any $\sigma_\star^2>0$.
Therefore, if $\beta<\betamin$, LAMA converges exponentially fast to its unique fixed-point solution $\sigma_\star^2$.

% start by defining the stability coefficient $\textnormal{SC}(\Psi)$ \cite{Maleki2010phd} as slope of the function $\beta\Psi(\sigma^2,\sigma^2)$ with respect to $\sigma^2$ at its largest fixed point, i.e.,
% \begin{align*}
% \textnormal{SC}(\sigma_\star^2) = \beta\!\!\left.\frac{\textnormal{d}}{\textnormal{d}\sigma^2}\Psi(\sigma^2,\sigma^2)\right\vert_{\sigma^2=\sigma^2_\star},\quad \sigma_\star^2 = \max_{\sigma^2>0}
% \!\left\{
% \sigma^2: \No+\beta\Psi(\sigma^2,\sigma^2)\geq \sigma^2
% \right\}
% \end{align*}

% We note that $\sigma_\star^2$ exists and is finite for a fixed $\No$ and $\beta$ because the MSE function $\Psi(\sigma^2,\sigma^2)$ is upper bounded as shown in \fref{app:DAMPsolvability}.
% % 
% Now, to see the result, we note that by definition of the MRT $\betamin$, if $\beta<\betamin$ the fixed point is unique and $\textnormal{SC}(\sigma^2)<1$ for all $\sigma^2>0$.
% % 

% \section{Proof of \fref{cor:convergence_noiseless}}\label{app:convergence_noiseless}
% Since $\No=0$ and $\beta<\betamax$, \LAMA recover the unique fixed-point at $\sigma^2=0$ by \fref{thm:recovery}. 
% % 
% To show exponentially fast convergence, it suffices to show \fref{lem:convergence_speed} holds, i.e., $\textnormal{SC}(\Psi)<1$ for all discrete sets $\setO$. 
% % 
% To do so, we use the result from \cite[Theorem 9]{WV2010}, which shows that for 
% discrete (finite or countably infinite) sets $\setO$, the slope MSE $\frac{\textnormal{d}}{\textnormal{d}\sigma^2}\Psi(\sigma^2,\sigma^2)\rightarrow 0$ as $\sigma^2\rightarrow0$. Since the unique fixed point is at $\sigma^2=0$, we are done.
% 